\documentclass[twocolumn,10pt,showpacs,notitlepage,floatfix,superscriptaddress]{revtex4-1}

\usepackage{textcomp}
\usepackage[noend]{algpseudocode}
\usepackage{algorithm}
\usepackage[normalem]{ulem}
\usepackage{mathpazo}
\usepackage{graphicx}
\usepackage{amssymb}
\usepackage{amsmath,wasysym}
\usepackage{amsthm}
\usepackage{amsfonts}
\usepackage{mathtools}
\usepackage{xcolor}

\definecolor{cb-black}      {RGB}{  0,   0,   0}
\definecolor{cb-blue-green} {RGB}{  0,  073,  073}
\definecolor{cb-green-sea}  {RGB}{  0, 146, 146}
\definecolor{cb-rose}       {RGB}{255, 109, 182}
\definecolor{cb-salmon-pink}{RGB}{255, 182, 119}
\definecolor{cb-purple}     {RGB}{ 73,   0, 146}
\definecolor{cb-blue}       {RGB}{ 0, 109, 219}
\definecolor{cb-lilac}      {RGB}{182, 109, 255}
\definecolor{cb-blue-sky}   {RGB}{109, 182, 255}
\definecolor{cb-blue-light} {RGB}{182, 219, 255}
\definecolor{cb-burgundy}   {RGB}{146,   0,   0}
\definecolor{cb-brown}      {RGB}{146,  73,   0}
\definecolor{cb-clay}       {RGB}{219, 209,   0}
\definecolor{cb-green-lime} {RGB}{ 36, 255,  36}
\definecolor{cb-yellow}     {RGB}{255, 255, 109}

\usepackage[colorlinks,linkcolor=cb-burgundy,urlcolor=cb-brown,citecolor=cb-brown,destlabel]{hyperref}
\usepackage{pifont}
\usepackage{stmaryrd}

\newtheorem*{lem*}{Lemma}
\newtheorem*{thm*}{Theorem}
\newtheorem{thm}{Theorem}[section]

\newtheorem{dfn}[thm]{Definition}

\newcommand{\ket}[1]{|#1\rangle}

\newcommand{\eq}[1]{\hyperref[eq:#1]{Eq.~(\ref*{eq:#1})}}
\renewcommand{\sec}[1]{\hyperref[sec:#1]{Section~\ref*{sec:#1}}}
\newcommand{\secsm}[1]{\hyperref[sec:#1]{Sec.~\ref*{sec:#1}}}
\newcommand{\app}[1]{\hyperref[app:#1]{Appendix~\ref*{app:#1}}}
\newcommand{\theo}[1]{\hyperref[thm:#1]{Theorem~\ref*{thm:#1}}}
\newcommand{\algo}[1]{\hyperref[alg:#1]{Algorithm~\ref*{alg:#1}}}
\newcommand{\lemm}[1]{\hyperref[lem:#1]{Lemma~\ref*{lem:#1}}}
\newcommand{\defn}[1]{\hyperref[defn:#1]{Definition~\ref*{defn:#1}}}
\newcommand{\corr}[1]{\hyperref[cor:#1]{Corollary~\ref*{cor:#1}}}
\newcommand{\condition}[1]{\hyperref[cond:#1]{Condition~\ref*{cond:#1}}}
\newcommand{\fig}[1]{\hyperref[fig:#1]{Fig.~\ref*{fig:#1}}}
\newcommand{\tab}[1]{\hyperref[tab:#1]{Table~\ref*{tab:#1}}}
\newcommand{\tabsm}[1]{\hyperref[tab:#1]{Tab.~\ref*{tab:#1}}}
\newcommand{\propos}[1]{\hyperref[prop:#1]{Proposition~\ref*{prop:#1}}}
\newcommand{\propsm}[1]{\hyperref[prop:#1]{Prop.~\ref*{prop:#1}}}
\newcommand{\rema}[1]{\hyperref[rem:#1]{Remark~\ref*{rem:#1}}}

\newcommand{\nocontentsline}[3]{}
\newcommand{\tocless}[2]{\bgroup\let\addcontentsline=\nocontentsline#1{#2}\egroup}

\DeclareMathAlphabet\mathbfcal{OMS}{cmsy}{b}{n} 

\makeatletter
\let\OldStatex\Statex
\renewcommand{\Statex}[1][3]{%
  \setlength\@tempdima{\algorithmicindent}%
  \OldStatex\hskip\dimexpr#1\@tempdima\relax}
\makeatother

\newcommand{\Blank}{\Statex[-1]}

\makeatletter
\providecommand\theHALG@line{\thealgorithm.\arabic{ALG@line}}
\makeatother

\newcommand{\stab}{\mathcal{S}}

\newcommand{\faultset}{\mathcal{F}}
\newcommand{\cF}{F}

\newcommand{\weight}{\textrm{wt}\,}

\newcommand{\topochannel}{\mathcal{T}}

\newcommand{\patch}{\mathcal{R}}

\newcommand{\instructionset}{\mathbfcal{I}}

\newcommand{\checkset}{\Sigma}

\newcommand{\cleftsemicirc}{\put(3.5,2.5){\oval(4,4)[l]}\put(3.5,.5){\line(0,1){4}}\phantom{\circ}}

\newcommand{\newcirc}{\put(2.5,2.5){\oval(4,4)}\phantom{\circ}}

\preprint{APS/123-QED}

\begin{document}

\title{Fault tolerance of stabilizer channels}
\author{Michael E. Beverland}
\thanks{Present Address: IBM Quantum, IBM T.J. Watson Research Center, Yorktown Heights, NY 10598, USA}
\affiliation{Microsoft Quantum, Redmond, WA 98052, USA}
\author{Shilin Huang}
\thanks{Present Address: Department of Applied Physics, Yale University, CT 06511, USA}
\affiliation{Department of Electrical and Computer Engineering, Duke University, Durham, NC 27708, USA}
\author{Vadym Kliuchnikov}
\affiliation{Microsoft Quantum, Redmond, WA 98052, USA}

\begin{abstract}
Stabilizer channels are stabilizer circuits that implement logical operations while mapping from an input stabilizer code to an output stabilizer code.
They are widely used to implement fault tolerant error correction and logical operations in stabilizer codes such as surface codes and LDPC codes, and more broadly in subsystem, Floquet and space-time codes.
We introduce a rigorous and general formalism to analyze the fault tolerance properties of any stabilizer channel under a broad class of noise models.
This includes rigorous but easy-to-work-with definitions and algorithms for the fault distance and hook faults for stabilizer channels. 
The generalized notion of hook faults which we introduce, defined with respect to an arbitrary subset of a circuit's faults rather than a fixed phenomenological noise model, can be leveraged for fault-tolerant circuit design. 
Additionally, we establish necessary conditions such that channel composition preserves the fault distance.
We apply our framework to design and analyze fault tolerant stabilizer channels for surface codes, revealing novel aspects of fault tolerant circuits. 
\end{abstract}

\maketitle


\section{Motivation and overview}

Large-scale quantum computers, built from inevitably imperfect hardware, necessitate fault tolerance to reliably execute extensive computations. 
To address this, a variety of approaches to prove that circuits are fault tolerant. 
These range from high-level strategies using simple phenomenological noise models \cite{dennis2002tqm} and asymptotic scaling arguments to prove the existence of finite thresholds \cite{brown2020}, to detailed analyses with explicit circuit noise models \cite{Wang2009} that facilitate accurate overhead estimates for different fault tolerance schemes \cite{Fowler2012, beverland2021codeswitching,das2020scalable,Chamberland2020b}. 
Numerical simulations that sample noise events also contribute, albeit without definitive correctness proofs \cite{geher2023b}.
Anticipating the proposal of improved fault-tolerant circuits for well-established codes like surface and color codes, as well as for promising new code families such as Floquet codes \cite{hastings2021dynamically} and LDPC codes \cite{kovalev2013,tillich2013QLDPC}, motivates the development of tools to rigorously analyze and design a wide range fault-tolerant circuits consistently.

In this work, we develop tools for \emph{stabilizer channels}, which implement logical operations on information stored in stabilizer codes~\cite{GottesmanThesis} by applying circuits built from \emph{stabilizer operations} (preparations and measurements in the Pauli bases, along with Clifford unitaries).
A stabilizer channel can have different input and output codes and can act on arbitrary input states (not just those produced by stabilizer operations).
Despite excluding some important circuits, notably those generating encoded magic states \cite{Knill2004,Bravyi2005}, the domain of stabilizer channels remains expansive, covering numerous logical operations, from stabilizer extraction and lattice surgery \cite{Horsman2012} to topological manipulations like braiding punctures \cite{Raussendorf2006,raussendorf2007fault,Raussendorf2007b,Bombin2008,Bombin2009,fowler2009high,Fowler2012} and twist defects \cite{Bombin2010,hastings2014reduced} implemented by code deformation of topological stabilizer codes. 
This broad applicability extends to the error correction and logical stabilizer operations of subsystem and Floquet codes \cite{poulin2005,hastings2021dynamically}, to circuit-based constructions~\cite{bacon2017,delfosse2023spacetime} known as space-time codes, and to LDPC codes~\cite{kovalev2013,tillich2013QLDPC}.
As such, tools to rigorously analyze and design stabilizer channels can be expected to have wide applicability. 

In what follows, we summarize the main contributions in this paper.
    
\textbf{Formalism.}---
After reviewing the basic properties of stabilizer codes and channels in \sec{basic-defs}, we provide a rigorous and general formalism to analyze the fault tolerance properties of any stabilizer channel under a broad class of noise models in \sec{fault-tolerance}.
Our formalism relies upon and connects together many previous works, complimenting existing formalisms that focus on topological codes~\cite{dennis2002tqm,Brown2017,Bombin2021}.
By not requiring any geometric locality, our formalism applies to a very wide range of fault-tolerant protocols including general quantum LDPC codes~\cite{gottesman2013,tillich2013QLDPC,breuckmann2021balanced,kovalev2013,panteleev2021,lin2023,wang2023}, Floquet codes~\cite{hastings2021dynamically,Davydova2023,aasen2022,higgott2023} and space-time codes~\cite{bacon2017,delfosse2023spacetime,gottesman2022opportunities,Gidney2021stim}.
Among other things, we provide rigorous definitions of the fault distance~\cite{Bombin2021} $d(\faultset)$ of a channel given a set of allowed faults $\faultset$ in terms of a general form circuit that captures the channel's logical action~\cite{kliuchnikov2023}.
This overcomes limitations of fault-distance definitions that only track logical Paulis introduced by faults, which fail to capture logical measurement outcome flips among other features.
We find this operational approach more convenient to work with than the complementary approach in Ref.~\cite{Bombin2021}, where the stabilizers of the Choi state are used to capture the logical channel action.

\textbf{Generalized hook faults.}---
In \sec{hook-faults} we introduce definitions to generalize and characterize hook faults~\cite{dennis2002tqm} (also called `hook errors') for general stabilizer channels. 
Hook faults can limit the fault tolerance of a circuit and are typically identified through their deviation from a simple noise model such as i.i.d Pauli noise on individual qubits.
We expand the concept by declaring that any of the faults $\faultset$ that can occur in the channel to be hook faults \emph{with respect to a fault subset} $\faultset_\text{sub} \subset \faultset$ if they are not in the subset.
Since all the configurations that can occur for $\faultset_\text{sub}$ can also occur for $\faultset$, the fault distance $d(\faultset) \leq d(\faultset_\text{sub})$.
If $d(\faultset) < d(\faultset_\text{sub})$, then some set of hook faults must be responsible for the gap. 
This allows for useful characterizations of problematic faults that limit the fault-distance of a stabilizer channel, which by varying $\faultset_\text{sub}$, can be used to iteratively improve circuit designs and make them fault tolerant. 

\textbf{Analytic analysis.}---
In \sec{code-def-channels}, we first observe that any stabilizer channel can be re-expressed as a sequence of code deformation rounds.
Leveraging the computation in Ref.~\cite{kliuchnikov2023} of the logical action of a single code deformation round allows us to analytically describe stabilizer channels in terms of the stabilizer groups of the input and output codes, and the group of measurements in each deformation round.

\textbf{Channel composition.}---
In \sec{combining-stabilizer-channels}, we provide explicit sufficient conditions for a set of stabilizer channels that ensure that larger channels formed by composing them together will inherit their fault tolerance properties. 
It is possible to extend our analysis of stabilizer channels to include compositions with fault-tolerant non-stabilizer states prepared by techniques that are not stabilizer channels (such as magic state distillation or code switching).
This extension, outlined in \sec{concl}, is quite straightforward since the framework puts no restrictions on the input states to stabilizer channels.

\textbf{Algorithms.}---
We provide an algorithm in \sec{calculating-fault-dist} to compute the fault distance of any stabilizer channel, which is similar to algorithms in Refs.~\cite{Gidney2021stim,Breuckmann2017}.
We also provide a new algorithm to identify problematic hook faults in a channel which we find to be very useful in the practical design of fault-tolerant stabilizer channels.
These algorithms apply when the channel check matrices are graph-like. 
We also discuss efficient algorithms for determining if a channel matrix is graph-like and various approaches that help bound the fault distance for channels with non-graph-like channel check matrices.

\textbf{Examples.}---
In \sec{examples}, we illustrate these ideas by walking through the design and fault-tolerance analysis of two example channels in surface codes.
The first example is the lattice surgery operation to implement a logical $XX$ measurement between a pair of code patches, which recovers the standard analysis in Refs.~\cite{ChamberlandCampbell2022}.
The second example is the logical Hadamard on a single patch, leaving it in an adjacent location.
While we developed our Hadamard implementation independently of the recent work~\cite{geher2023b}, we point out that our approach is quite similar to the first of the two examples therein, although our circuit-level implementation has a worse depth than theirs (asymptotically $3d$, compared with their $2d$ error correction cycles for distance $d$).

These examples, particularly the Hadamard, highlight a number of novel features that can arise in fault-tolerant circuit design (some of which were also highlighted in Ref.~\cite{geher2023b}).
In particular, we find that stabilizer extraction circuits exhibit:
\begin{enumerate}
    \item \emph{Spatial variation.}--- Our Hadamard channel involves some irregular surface code patches in intermediate steps which require non-uniform stabilizer extraction circuits (for example min-weight X type logical operators may be horizontal in one region of the lattice, and vertical in another).
    \item \emph{Temporal variation.}--- 
    Due to the existence of `diagonal' faults in the circuit-level syndrome graph, logical operators can be deformed in the time dimension without increasing their weight. 
    This results in the \emph{same} stabilizer code in our channel being measured using different circuits during different rounds to avoid problematic hook faults in that align with logical operators that have been deformed from other times.
\end{enumerate}

Lastly, we conclude in \sec{concl}, highlighting a number of future directions.

\section{Basic definitions}
\label{sec:basic-defs}

Here we set notation and review some important material that is relevant for this work.

\subsection{Stabilizer codes}

We find it convenient to use the notation
that for Pauli operators $A$ and $B$, the commutator $\llbracket A, B \rrbracket = 0$ if $A$ and $B$ commute and  $\llbracket A, B \rrbracket = 1$ otherwise.
We write the set of integers from $1$ to $k$ as $[k]$, and the set of integers from $j$ to $k$ as $[j,k]$.
We write $H^{\perp}$ to denote the set of Pauli operators that commute with all elements of a Pauli operator set $H$.

If $P \ket{\psi} = \ket{\psi}$, we say that the Pauli operator $P$ \textit{stabilizes} the state $\ket{\psi}$, and also that the state $\ket{\psi}$ is \textit{stabilized} by the Pauli operator $P$.
We say that a group of commuting operators $\stab$ stabilizes a state $\ket{\psi}$ if all elements of $\stab$ stabilize $\ket{\psi}$.
An Abelian Pauli group $\stab$ specifies a stabilizer code, where the codespace consists of those states which are stabilized by $\stab$.
Measuring a complete set of generators of the stabilizer group $\stab$ produces a set of outcomes that we call the \textit{syndrome}, which is trivial if the state is in the codespace.
We also define the syndrome $s(P)$ of a Pauli operator $P$ with respect to a complete set of generators $(S_1, S_2, \dots, S_m)$ of a stabilizer group $\stab$, which is the list $s(P) = (\llbracket S_1, P \rrbracket, \llbracket S_2, P \rrbracket, \dots, \llbracket S_m, P \rrbracket)$.

Given a Pauli group $G$ and a Pauli operator $P$, we often write $PG$, which denotes the set $PG = \{Pg | g \in G \}$.
Given a pair of Pauli groups $G$ and $H$, we often write $GH$, which denotes the set $GH = \{gH | g \in G \}$. 

For the measurement of Hermitian Pauli operators, we use the convention that the possible outcomes are labeled by $0$ and $1$ (not $\pm 1$).
We often consider the addition of measurement outcomes or sets of measurement outcomes, which is always performed modulo two unless otherwise stated.

\subsection{Stabilizer circuits and general form}
\label{sec:stab-circuits}

In this work, we consider circuits built of elements of an explicit \textit{instruction set} $\instructionset$.
We require the operations in $\instructionset$ are \emph{stabilizer operations}:
\begin{itemize}
    \item preparations of qubits in the computational basis,
    \item non-destructive measurements of Pauli operators,
    \item destructive measurements of single-qubit Pauli operators,
    \item Clifford unitaries,    
    \item the generation of classical random bits,
    \item \emph{conditional Pauli operations controlled on the parities of prior measurement outcomes and random bits.}
\end{itemize}
The instruction set $\instructionset$ could in principle consist of of all possible stabilizer operations, but it may be useful to further restrict to those which match the capabilities of a specific hardware platform. 
For example $\instructionset$ may contain all single-qubit stabilizer operations and CNOT gates between specified pairs of qubits.

The operations in a stabilizer circuit $\mathcal{C}$ occur at specified time steps, with compatible sets of operations being allowed in the same time step.
(Note that precisely which sets of operations are considered compatible can depend on the setting, such as the capabilities of a hardware platform that the circuit is to be run on).
We find it convenient to assign a specific sequence to all the operations in the circuit, which respects the ordering of time steps such that operations in the same time step of the circuit form a consecutive subsequence.
We call the sequence of measurement outcomes and random bits produced by the circuit the \emph{circuit outcome vector}.
In Ref.~\cite{kliuchnikov2023}, it is shown that any stabilizer circuit $\mathcal{C}$ is equivalent to a simple general form circuit $\mathcal{C}_\text{gen}$ as shown in \fig{general-form-circuit}. 
By \emph{equivalent} we mean that when both circuits are fed the same input state $\ket{\psi}$, then the state $\ket{\psi'(O)}$ output by $\mathcal{C}$ when its outcome vector is $O$ is the same as the state $\ket{\psi'_\text{gen}(V)}$ output by $\mathcal{C}_\text{gen}$ when its outcome vector is $V$, for some known (affine) relation between $O$ and $V$. 
As we will see, this general form is particularly helpful in characterizing and analyzing the logical action of stabilizer channels.

We include the conditional Paulis in italics in the list of stabilizer operations because in this work we treat them differently from the other stabilizer operations. 
In particular, we take the stabilizer circuit $\mathcal{C}$ to be specified by the other stabilizer operations it contains, and we include or remove conditional Paulis at our convenience.
This is because any conditional Paulis can be done off-line (by a Pauli-frame update) and/or postponed to the end of the stabilizer circuit without changing any of the circuit's other instructions. 
Moreover, given a general form circuit for a stabilizer circuit, changing the circuit's conditional Paulis only results in a change of the conditional Paulis of the general form circuit.

\begin{figure}[ht]
    \includegraphics[width=0.7\linewidth]{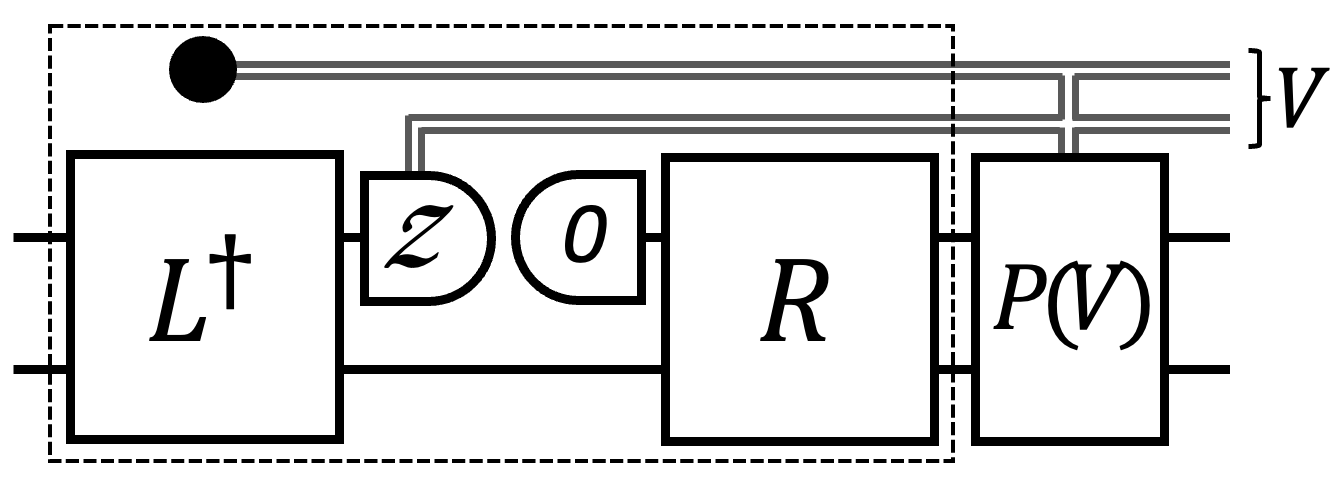}
    \caption{
    Any stabilizer circuit $\mathcal{C}$ is equivalent to a general form circuit $\mathcal{C}_\text{gen}$ as shown here, which consists of five steps: 
    i) initial Clifford $L$,
    ii) destructive measurements,
    iii) ancilla preparation,
    iv) final Clifford $R$,
    v) conditional Pauli.
    The conditional Pauli is controlled by the general form's \emph{circuit outcome vector} $V$, consisting of random bits (represented by a black disk) and measurement outcomes.
    Changing the conditional Paulis in $\mathcal{C}$ only changes the conditional Pauli in $\mathcal{C}_\text{gen}$ (not the parts in the dashed box).
    }
    \label{fig:general-form-circuit}
\end{figure}

\subsection{Stabilizer channels}
\label{sec:stab-circuits-logical}

Following the language of Ref.~\cite{Bombin2021}, we say that a stabilizer circuit $\mathcal{C}$ is a \emph{stabilizer channel} from $\stab_\text{in}$ to $\stab_\text{out}$ if, for any input state encoded in $\stab_\text{in}$ and for every possible channel outcome vector, it outputs a state encoded in $\stab_\text{out}$. 
Note that to achieve this for all outcome vectors, a circuit $\mathcal{C}$ will typically make use of carefully chosen conditional Paulis.
Unless otherwise stated, we will assume throughout this work that all conditional Paulis are placed at the end of the stabilizer channel as shown in \fig{enc-and-dec}(a).

Given a stabilizer channel from $\stab_\text{in}$ to $\stab_\text{out}$, it is natural to ask what logical action it has on the encoded information.
To make the logical action of the channel uniquely defined, we must specify an explicit logical basis for each of $\stab_\text{in}$ and $\stab_\text{out}$, which can be achieved by providing encoding Cliffords $C_\text{in}$ and $C_\text{out}$ for the input and output codes.
An encoding Clifford $C$ for a $k$-qubit stabilizer code $\stab$ maps $X_i$ and $Z_i$ for $i \in [k]$ to a basis of logical operator representatives, while it maps $Z_i$ for $i \in [k+1,n]$ to a set of independent stabilizer generators of the code.
The logical action of a stabilizer circuit $\mathcal{C}$ can then be succinctly characterized by finding a general form circuit equivalent to the composite stabilizer circuit formed by sandwiching $\mathcal{C}$ between encoding and unencoding Cliffords for the input and output code as shown in \fig{enc-and-dec}(b).
An appropriate choice of conditional Paulis at the end of the stabilizer channel yields the simple logical action general form shown, with no random bits or conditional Paulis.
One such choice is to use the conditional Paulis obtained from finding a general form for the circuit formed by the ancilla state preparation, $C_\text{in}$ and the unconditional part of the stabilizer circuit.
An efficient algorithm exists to find a general form for any stabilizer channel~\cite{kliuchnikov2023}.

\begin{figure}[ht]
    (a)\includegraphics[width=0.8\linewidth]{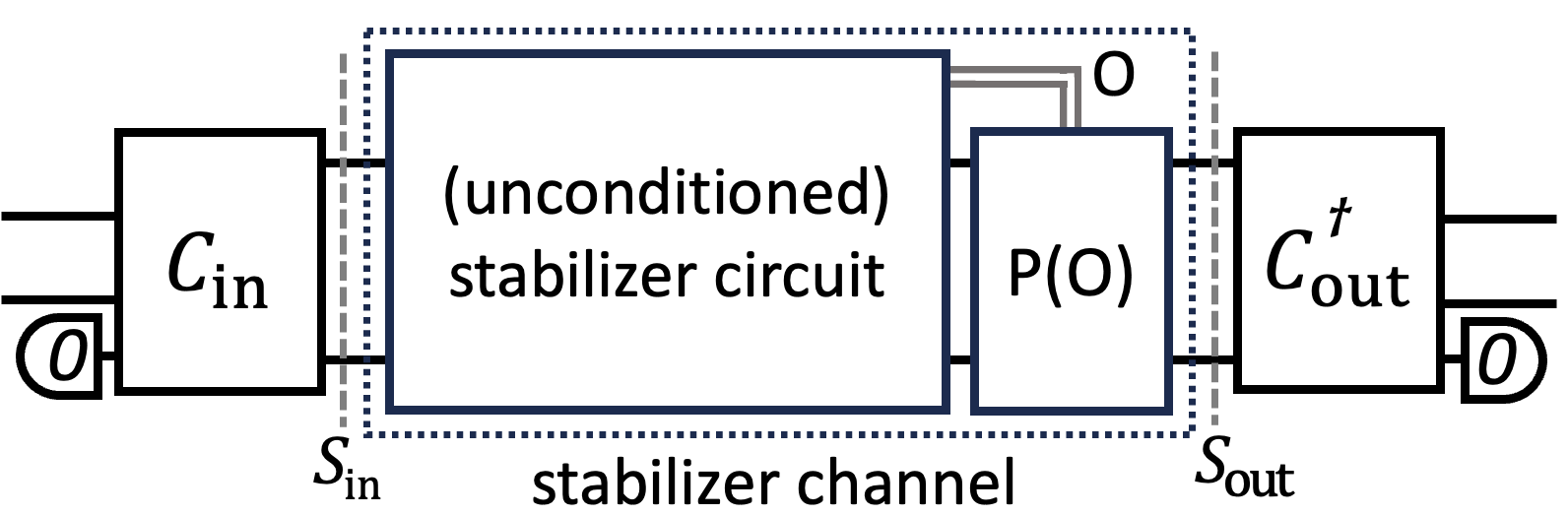}
    (b)\includegraphics[width=0.6\linewidth]{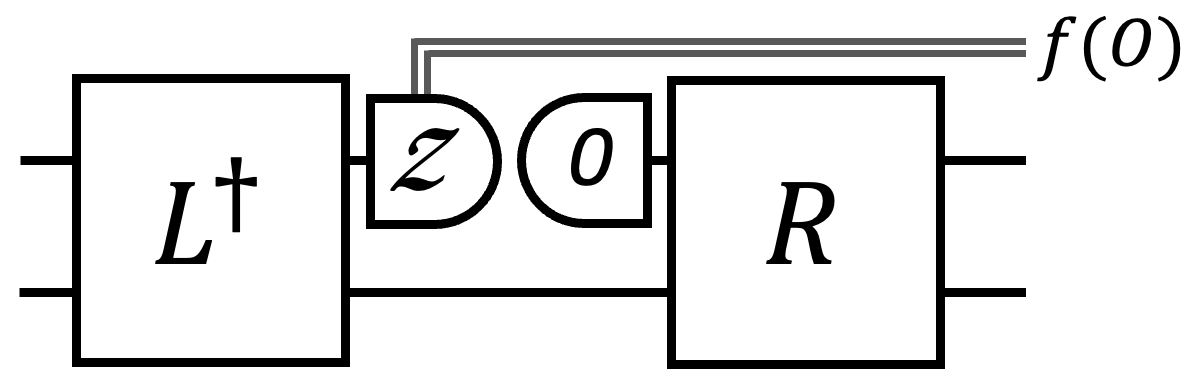}
    \caption{
    (a) The logical action of a \emph{stabilizer channel}, which is a stabilizer circuit that takes input states encoded in the code $\stab_\text{in}$ and produces output states encoded in $\stab_\text{out}$.
    We take the convention that all conditional Paulis are at the end of the stabilizer channel.
    The Clifford unitaries $C_\text{in}$ and $C_\text{out}$ fix the logical basis of $\stab_\text{in}$ and $\stab_\text{out}$.
    (b) A general form circuit equivalent to the logical action.
    Note there are no random bits or classically controlled Paulis in this general form, which can be ensured by an appropriate choice of classically controlled Paulis for the stabilizer channel.
    The outcome vector $\mathfrak{f}(O)$ of the general form is a function of the outcome vector $O$ of the stabilizer channel.
    }
    \label{fig:enc-and-dec}
\end{figure}

We will make use of the following concepts and definitions.
There is a natural correspondence between a Clifford unitary $C$ and the symplectic basis for $n$ qubits: $CX_1C^\dagger, C Z_1C^\dagger, CX_2C^\dagger, C Z_2C^\dagger, \dots, CX_nC^\dagger, C Z_nC^\dagger$.
We say that a symplectic Pauli basis $\mathcal{B}$ \emph{respects} a stabilizer group $\stab$ if the basis contains a complete set of generators for $\stab$.
The encoding Cliffords $C_\text{in}$ and $C_\text{out}$ thereby specify the symplectic Pauli bases $\mathcal{B}_\text{in}$ and $\mathcal{B}_\text{out}$ which in turn respect $\stab_\text{in}$ and $\stab_\text{out}$.
A symplectic basis $\mathcal{B}$ that respects the code with stabilizer group $\stab$ is a convenient tool for determining a complete set of logical operators of $\stab$.
A symplectic basis $\mathcal{B}$ can be used to define a symplectic basis $[\mathcal{B}]_{\stab}$ for the logical operators of $\stab$ using a simple two-step procedure.
First remove all anti-commuting pairs from the basis for which one of the elements in the pair is in $\stab$, second multiply all the remaining elements by $\stab$.
We omit the subscript in $[\mathcal{B}]_{\stab}$ when the code is clear from the context.

\begin{figure}[ht]
\includegraphics[width=0.7\linewidth]{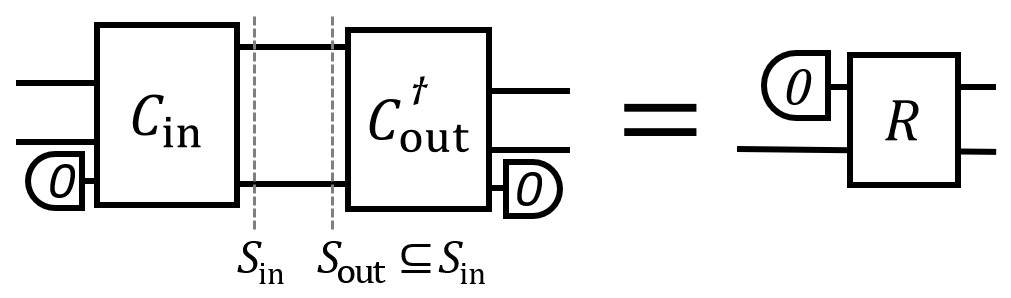}
    \caption{
    A \emph{code basis conversion} is the logical action of a trivial circuit with respect to different code bases.
    This requires that the stabilizer group of the output code is contained in that of the input code $\stab_\text{out} \subseteq \stab_\text{in}$.
    The logical action is a simple general form circuit specified by a $k_\text{out}$-qubit Clifford $R$, where $k_\text{out}$ is the number of logical qubits in $\mathcal{S}_\text{out}$. 
    }
    \label{fig:code-basis-conversion}
\end{figure}

The simplest example of a stabilizer channel is the trivial circuit.
However, even though the circuit itself is trivial, when it is viewed as a channel with respect to specified input and output code bases it can have some non-trivial logical action if the two codes or their bases are different as shown in \fig{code-basis-conversion}.
We call this a \emph{code basis conversion} channel. 
The parameters of a code basis conversion channel must satisfy certain properties.
The stabilizer group of the output code must be contained in that of the input code.
The unitary $R$ specifying the logical action must map each independent logical operator of $\mathcal{S}_\text{out}$ from $\stab_\text{in} / \mathcal{S}_\text{out}$ to  $\bar Z_1,\ldots, \bar Z_l$ in $[\mathcal{B}_\text{out}]_{\mathcal{S}_\text{out}}$, where $l$ is the number of the independent generators of $\mathcal{S}_\text{in} / \mathcal{S}_\text{out}$.
Another condition on $R$ concerns the ordered set formed by taking $\mathcal{B}_\text{in}$ and removing all anti-commuting pairs that have any of the elements in $\mathcal{S}_\text{in}$ and multiplying the remaining pairs by $\mathcal{S}_\text{out}$. 
Then $R$ must map the logical operators in this set to $\bar Z_{l+1},\bar X_{l+1}, ... , \bar Z_{k_\text{out}}, \bar X_{k_\text{out}}$ from $[\mathcal{B}_\text{out}]_{\mathcal{S}_\text{out}}$ where $k_\text{out}$ is the number of logical qubits in $\mathcal{S}_\text{out}$.

\section{Fault tolerance definitions}
\label{sec:fault-tolerance}

In this section we consider the effect of noise on stabilizer channels and define key objects that we will work with later.
Our formalism is similar to other approaches in the literature, especially those that represent channels as space-time codes~\cite{bacon2017,delfosse2023spacetime,gottesman2022opportunities,Gidney2021stim}.
In summary, we work with a binary vector $v$ that represents the configuration of faults in the channel, a binary matrix $A_\checkset$ that captures the syndrome $\sigma = A_\checkset v$ available for error correction, and a binary matrix $A_L$ that captures the logical action of $v$ on the channel via $A_L v$.
A key object we will work with is the fault distance, which is then given by the smallest $v$ such that $A_\checkset v = 0$ while $A_L v \neq 0$.
A novel aspect of our treatment is that we define $A_L$ in terms of a general form representation of the logical channel action described in Ref.~\cite{kliuchnikov2023}.
This automatically captures logical channel failures not just due to logical Paulis of the output code being introduced, but also flips of logical measurement outcomes, and avoids incorrectly identifying the introduction of a logical output operator when that operator is stabilized by the logical channel action (for example suppose $XX$ is applied due to faults in a channel that measures $XX$).
We find this operational approach more convenient to work with than the complementary approach in Ref.~\cite{Bombin2021}, where the stabilizers of the Choi state are used to capture the logical channel action.

\subsection{Channel outcomes}
\label{sec:channel-outcomes}

Throughout this section we consider a stabilizer channel $\mathcal{C}$ with an observed outcome vector $\tilde{O}$ from an input code with stabilizer group $\stab_\text{in}$ to an output code with stabilizer group $\stab_\text{out}$, with respect to Pauli bases $\mathcal{B}_\text{in}$ and $\mathcal{B}_\text{out}$. 
The logical action of the channel is captured by the general form circuit as in \fig{code-deformation-channel-logical} with Clifford unitaries $L$, $R$ and the outcomes $m$ of the logical measurements being encoded in a linear outcome function $\mathfrak{f}(O)$, where $O$ is the fault-free outcome vector (which may differ from the observed outcome vector $\tilde{O}$ due to noise).
The following concepts will be used throughout this section.

{\bf Extended outcome vectors}.---
In the remainder of this section we will work with an extended version of the observed circuit outcome vector $\tilde{O} \oplus s_\text{out}$, where $s_\text{out}$ is the syndrome of the output code at the end of the channel.
$\tilde{O} \oplus s_\text{out}$ contains all of the information available to a decoder to perform error correction, which can include not only the outcomes $\tilde{O}$ that are observed in the channel itself, but also extra information that could potentially be obtained by acting on the state output by the channel.
This extra information is at most that which is contained in $s_\text{out}$, but in practice is likely to be noisy.
Here we include $s_\text{out}$ with no noise, which is equivalent to the standard practice of including a `perfect round' of syndrome extraction at the end of the channel.
It is important to take care that this perfect output syndrome information does not lead to spurious analysis results, and we will return to this point in \sec{combining-stabilizer-channels} when considering the composition of channels. 
For notational convenience, unless otherwise stated, we use extended versions of outcome vectors from here on but drop the `extended' descriptor and write $\tilde{O}$ rather than $\tilde{O} \oplus s_\text{out}$.

\subsection{Independent stochastic Pauli noise}
\label{sec:stochastic-noise}

In this work we assume noise occurs according to an \textit{independent stochastic Pauli noise model}, whereby each operation in the circuit has an associated set of allowed \emph{faults}, each of which occurs independently with some fixed probability.  
Each fault consists of the application of a specific Pauli operator to the physical qubits after the instruction, and flips of specific measurement outcomes.
We assume that conditional Pauli operations never fail, and that they are applied at the end of the circuit. 
This assumption is justified because conditional Paulis can be applied offline by a Pauli-frame update and because the circuit, including any faults, is built from Stabilizer operations. 
Note that although the conditional Paulis are fault-free, whether or not they are applied can be flipped by faults which change the bitstring upon which they are controlled.

For each operation in $\mathcal{C}$, and for each fault it can experience according to the noise model, we assign an element to the \emph{elementary fault set} $\faultset$.
For each element of $\faultset$, we specify the probability $\Pr$ that the elementary fault occurs, and its \emph{fault effect}, which is the effect the fault has on the channel.
The fault effect consists of the \emph{residual error} Pauli $E$ which is applied to the qubits at the end of the stabilizer channel as a result of the fault, and the \emph{residual outcome flip} bitstring $\delta$ which specifies how the channel outcome vector $O$ is modified to $O + \delta$ as a result of the elementary fault.
Both $E$ and $\delta$ can be calculated for a given fault by inserting the fault into the stabilizer circuit, and propagating it to the end.
In \sec{3b} we will put some restrictions on the elementary fault set.

Multiple elementary faults can occur together, forming a \emph{fault configuration} $F$ which is a subset $F \subset \faultset$.
Since elementary faults are sampled independently from $\faultset$ according to their probabilities, which forms a fault configuration $F \subset \faultset$ with probability $\Pr(F)$, which is the probability that those elementary faults in $F$ occur, and all elementary faults in $\faultset \setminus F$ do not.
For any fault configuration $F$, we define the residual error $E(F)$ from the product of the residual errors of each elementary fault in $F$, and the residual outcome flip $\delta(F)$ by taking the sum over the residual outcome flips of each elementary fault in $F$.

We also find it useful to define the effect of a Pauli error $P$ acting on the input qubits: $E_\text{in}(P)$ and $\delta_\text{in}(P)$ are the residual error at the end of the channel and residual outcome flip that would occur if $P$ was applied at the start of the channel without any faults.

We can associate a weight $\weight(F)$ with each fault configuration.
Unless otherwise stated, we take $\weight(F) = \log \Pr(\emptyset)  -\log \Pr(F)$ for every fault configuration $F \subset \faultset$,
so that $\weight(F_1 \sqcup F_2) = \weight(F_1) + \weight(F_2)$.
We assume throughout this paper that the weights of all elementary faults (and therefore all faults) are non-negative.
At times, we will set all weights to have unit weight (and will clearly state when doing so).
This is useful in some analysis since if for example all elementary faults occur with a probability proportional to $p$, then when assuming unit weights, the probability of some fault $F$ occurring is $\Theta(p^{\text{wt}(F)})$.

Our analysis in this paper holds for any independent stochastic Pauli noise model, but we will often refer to the following standard circuit noise model.

\textbf{Circuit noise}.---
The instruction set consists of arbitrary single-qubit stabilizer operations, and arbitrary two-qubit stabilizer operations on specified pairs of qubits. 
We also allow arbitrary conditional Paulis.
The faults are then:
\begin{itemize}
    \item Three faults for each each single-qubit instruction. 
    Each of $\{X,Y,Z\}$ occurs with probability $p/3$.
    \item Fifteen faults for each two-qubit instruction (each of which is a non-trivial two-qubit Pauli). 
    Each occurs with probability $p/15$.
    \item An outcome flip of each measurement instruction, which occurs with probability $p$.
\end{itemize}

There are a number of variants of circuit noise models which are used in the literature which can be captured by slight modifications of the parameters in the above models. 
For example, some models are phrased in terms of exclusive rather than independent events (for example when an idle qubit failure occurs some models say that \emph{either} an $X$, \emph{or} a $Y$ \emph{or} a $Z$ error occurs).
The standard exclusive phenomenological and circuit noise models can be re-expressed as independent stochastic Pauli noise models of the type here with specific fault probability values~\cite{Chao2020}.
The allowed multi-qubit operations should be informed by hardware constraints, and this connectivity will have a significant impact on the details of circuits that can be implemented~\cite{Chamberland2020a,Chao2020,Tremblay2021}.

When analysing circuit noise, we also consider restricted circuit noise where the Pauli faults that follow one or two-qubit instructions are restricted to a subset of the three or the fifteen faults correspondingly.

\subsection{Error correction}
\label{sec:3b}

Here we primarily review some standard material on fault-tolerant quantum error correction for stabilizer channels, although we point the expert reader to our definition of undetectable fault configurations which is new.
Many quantum error correction analyses focus on the scenario of memory where the goal is to fault-tolerantly store information, and do not address subtleties that can arise when the goal is instead to fault-tolerantly perform logical operations including measuring out some logical qubits.
Here we use the general form circuit to capture the logical action, which complements the approach in Ref.~\cite{Bombin2021} where the stabilizers of the Choi state are used. 

In what follows, we will use the convention that the fault-free channel outcome vector is $O$, and that given a fault configuration $F$, and an input Pauli error $P$, the observed channel outcome vector is $\tilde{O} = O + \delta(F) + \delta_\text{in}(P)$, and the output Pauli error at the end of the channel is $E(F) \cdot E_\text{in}(P)$.

{\bf Channel checks}.---
Each \emph{channel check} is either a parity (or the parity plus one) of a subset of elements of the channel outcome vector $\tilde{O}$, with the property that its value is zero for every possible fault-free channel outcome vector $O$.
The sum of two checks is also a check, such that the checks form a vector space over $\mathbb{F}_2$.
We use $\checkset$ to denote a complete basis of channel checks.
Note that the checks are defined without any reference to the noise model - they are a property of the channel irrespective of the kind of noise it undergoes.
We define the \emph{syndrome} $\sigma(\tilde O)$ to be the observed outcomes of the checks from $\checkset$.

Note that any fault-free outcome vector $O$ has trivial syndrome $\sigma(O) = 0$, and as such, $\sigma(\tilde O) = \sigma(\delta(F)) + \sigma(\delta_\text{in}(P))$. 
We assume that the syndrome contains all information used by the decoding algorithm to perform error correction.
We can think of $\sigma$ as an organized and potentially compressed representation of the data contained in $\tilde O$ which retains information about faults in the channel in a form that is conveniently consumed by the decoder.
Different choices of check bases for the check space, while containing equivalent information, may be more or less convenient for decoding purposes.

{\bf Undetectable fault configurations}.---
We say that a fault configuration $F$ is \emph{undetectable} if it has trivial syndrome, i.e. if $\sigma(\delta(F)) = 0$.
Note that since $s_\text{out}$ is included at the end of $\delta(F)$, the residual error $E(F)$ for an undetectable fault $F$ must be a logical operator with respect to the output code basis $\mathcal{B}_\text{out}$, which we write as $\bar{E}(F)$.
We write the \emph{logical action} of the fault $F$ as the bitstring $\mathcal{L}(F) = \mathfrak{f}(\delta(F)) \oplus a(\bar{E}(F)) \oplus b(\bar{E}(F)) \oplus c(\bar{E}(F))$, where $a(\bar{E}(F))$ $b(\bar{E}(F))$ and $c(\bar{E}(F))$ are the bitstrings that uniquely define the effect of $\bar{E}(F)$ on the logical action of the channel as defined in \fig{logical-effect-bitstring}.

If $F$ is undetectable and has trivial logical action we say it is \emph{trivial}.
We say two fault configurations $F$ and $F'$ are \emph{equivalent} iff $\sigma(\delta(F+F')) = 0$ and $\mathcal{L}(F+F') = 0$.

\begin{figure}[h]
    \includegraphics[width=1.0\linewidth]{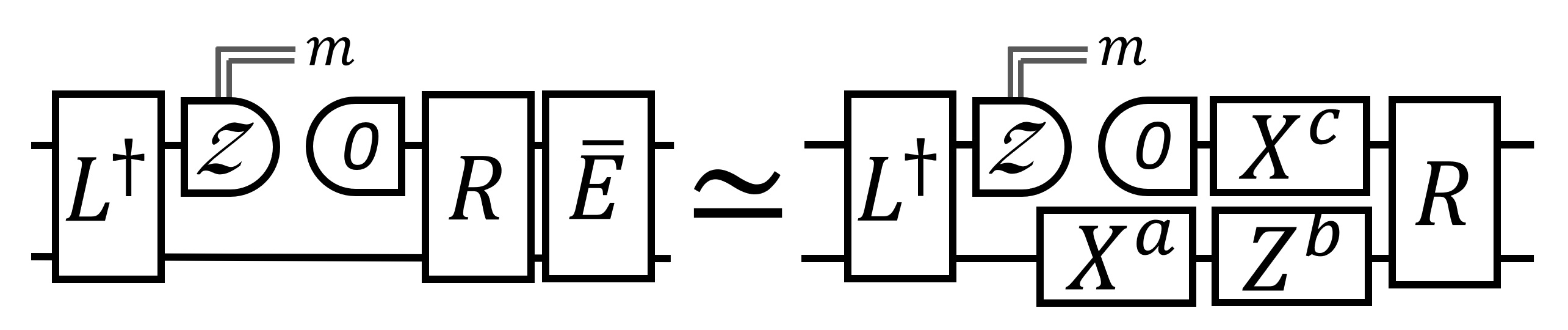}
    \caption{
    The effect on the logical action of a logical operator $\overline{E}$ for the output code.
    Applying $\overline{E}$ at the end of a stabilizer channel is equivalent to inserting Pauli operators $X^a$, $Z^b$, $X^c$ in the the general form circuit as shown.
    We use the bitstrings $a(\overline{E})$, $b(\overline{E})$ and $c(\overline{E})$ to uniquel capture the action of $\overline{E}$ on the channel.
    }
    \label{fig:logical-effect-bitstring}
\end{figure}

{\bf Minimum-weight decoding}.---
In this work we assume the \emph{minimum-weight} (MW) decoder $\mathcal{D}$, which finds a lowest-weight fault configuration $\hat{F} = \mathcal{D}(\sigma)$ which matches the observed channel syndrome $\sigma = \sigma(\tilde{O})$.
This is a most-probable fault configuration with the observed syndrome. 
A MW decoder may not be the optimal decoder, because the latter requires finding the most probable \emph{set} of equivalent fault configurations.
From $\hat{F}$, we can obtain the Pauli correction $\hat E = E(\hat{F})$ to be applied to the data qubits at the end of the channel and also a corrected version of the channel outcome $\hat{O} = \tilde{O} + \delta(\hat{F})$ that is used to obtain the logical outcome via $\mathfrak{f}(\hat{f})$.
We will assume that the decoder is deterministic, such that the same correction is produced for a given syndrome $\sigma$, and we also that when the syndrome $\sigma$ is trivial, $\mathcal{D}(\sigma)$ is also trivial.

In the presence of noise, error correction is first applied, producing a fault configuration $\hat{F}$.
This defines a corrected syndrome the corrected outcomes $\hat{O} = \tilde O + \delta(\hat{F})$, for which all channel checks are satisfied.
The overall logical action of the faults and error correction is then $\mathcal{L}(F + \hat{F})$, and we say that error correction has succeeded if this is trivial and that it has failed otherwise.
Since the logical outcome is obtained via $\mathfrak{f}(\hat{O})$, this ensures that we can use any choice of logical outcome function $f$, because all the choices of $f$ are equivalent up to checks and after the outcomes are corrected, all the checks are zero. 

{\bf Restrictions on the elementary fault set}.---
We can use some of the above definitions to impose useful restrictions on the elementary fault set $\faultset$.
Firstly, we exclude undetectable elementary faults from $\faultset$, since they either have no effect on quantum error correction if their logical action is trivial or render quantum error correction futile otherwise.
When convenient, we can furthermore assume that each elementary fault in $\faultset$ has a unique effect pair $(E,\delta)$.
Note that if this is not the case for a given elementary fault set, we can modify it to enforce this condition by retaining only one elementary fault in $\faultset$ with a particular effect $(E,\delta)$. 
When doing this, the $\Pr$ assigned to the retained elementary fault in the new elementary fault set is obtained from the sum over the probabilities of all combinations of an odd number of the equivalent elementary faults occurring in the original elementary fault set.

{\bf Matrix representations}.---
It is often convenient to represent a fault $F$ as a vector $v_F \in \mathbb{F}_2^{|\faultset|}$
that we call a fault vector.
Motivated by this definition, we use $F + F'$ to represent the symmetric difference of two faults.
Similarly, we represent an input Pauli error $P$ as an input error vector $w_\text{in} \in \mathbb{F}_2^{2 n_\text{in}}$, which is expressed in the input code basis $\mathcal{B}_\text{in}$. 
We can then represent the linear dependence of checks $\checkset$ on fault vectors $v_F$ as a \emph{channel check matrix} $A_{\checkset}$ defined via the equation $A_\checkset v_F = \sigma(\delta(F))$.
Every undetectable fault configuration $F$ corresponds to an element $v_F$ of the kernel of $A_\checkset$.
Similarly, it is  convenient to represent the linear dependence of logical action of fault configurations on fault vectors $v_F$ as a \emph{logical effect matrix} $A_{L}$ defined via the equation $A_L v_F = \mathcal{L}(F)$.
Every trivial fault configuration $F$ corresponds to $v_F$ of the kernels of $A_\checkset$ and $A_L$, 
and non-trivial undetectable fault configuration $F$ corresponds to $v_F$ in the kernels of $A_\checkset$ and not in the kernel of $A_L$.

\subsection{Fault distance}
\label{sec:logical-fault-configurations}

Here we introduce the fault distance, which tells us how many faults can be corrected in a channel.

\begin{dfn}[Fault distance]
    Consider a stabilizer channel $\mathcal{C}$ affected by a stochastic Pauli noise with elementary fault set $\faultset$.
    The fault distance $d(\faultset)$ is the minimum weight of a non-trivial undetectable fault configuration of $\mathcal{C}$,
    that is, the minimum weight of an element of the kernel of the channel check matrix $A_{\checkset}$ that is not in the kernel of the logical effect matrix $A_L$.
\end{dfn}

We will often refer to a non-trivial undetectable fault configuration as a \emph{logical fault configuration}.
First, we show that the channel achieves its fault distance on a simple fault configuration, where we say that a fault configuration $\cF$ is \emph{simple} if there is no undetectable fault configuration $\cF'$ that is a proper subset of $\cF$.
Suppose that the fault configuration $\cF$ achieving the distance is not simple. 
We can represent it as a disjoint union of two undetectable fault configurations: $\cF = \cF' \sqcup (\cF \setminus \cF')$.
Consider the case when the channel distance is greater than zero. 
Suppose that $\cF'$ is a non-trivial undetectable fault configuration; then, its weight must be the same as that of $\cF$. 
For this reason, $\cF \setminus \cF'$ must be a weight-zero fault configuration and, therefore, trivial. 
Otherwise, this would contradict the minimal weight of a non-trivial undetectable fault configuration being positive.
Therefore, we can replace $\cF$ with a non-trivial undetectable fault configuration $F'$ of strictly smaller size.
A similar argument applies when $\cF \setminus \cF'$ is a non-trivial undetectable fault configuration.
The second case is when the channel distance is zero. 
Both $\cF'$ and $\cF$ have a weight of zero.
If $\cF'$ is a non-trivial undetectable fault configuration, we can replace $\cF$ with a non-trivial undetectable fault configuration $\cF'$ of strictly smaller size. 
A similar argument applies when $\cF \setminus \cF'$ is a non-trivial undetectable fault configuration. 
We can continue this process until we obtain a simple non-trivial undetectable fault configuration.

Next we show an important implication that the fault distance of a channel has on its error correction properties.
We assume that the MW decoder has access to the channel checks and the output syndrome of the channel
\footnote{This is equivalent to appending a perfect round of syndrome extraction to the end of the channel and extending channel checks with the perfect round outcomes.}.
This assumption is motivated by the fact that the channel is usually a part of a larger computation and the later steps 
of the computation can be used to infer the output code syndrome.
With the access to the output code syndrome, the MW decoder can correct any error with weight less than $d(\faultset)/2$.
To see this, suppose that a fault configuration $F$ occurs with syndrome $\sigma(F)$ and output code syndrome $s_\text{out}(F)$.
The MW decoder will output a minimum weight fault configuration $F'$, where $\weight(F') \le \weight(F)$, $\sigma(F')= \sigma(F)$
and $s_\text{out}(F')= s_\text{out}(F)$.
The fault configuration $F + F'$ is undetectable and has weight less than $2\weight(F)$ and is therefore trivial (since $2\weight(F)$ is less than $d(\faultset)$).
For this reason, fault configurations $F$ and $F'$ have the same logical action and we can correct $F$ by applying the correction $F'$.

Now we consider the lowest weight errors that lead to failure. 
Given the above bound the smallest such error must be of weight $d(\faultset)/2$.
To describe a case when a fault configuration of weight $d(\faultset)/2$ cannot be corrected we need to distinguish between 
odd and even simple fault configurations. 
We say that a simple fault configuration $F$ is \emph{even} if it can be written as a disjoint union of fault configurations $F_1 \sqcup F_2$ 
such that $\weight(F_1) = \weight(F_2) = \weight(F)/2$.
Every fault configuration that is not even, is \emph{odd}.
Suppose the channel fault distance is achieved on an even fault configuration $F_1 \sqcup F_2$, with $\weight(F_1) = \weight(F_2) = d(\faultset)/2$.
Note that at least one of $s_\text{out}(F_1)$, $\sigma(F_2)$ is non-zero because $F_1 \sqcup F_2$ is simple.
The logical actions of  $F_1$ and $ F_2$ must be different because $F_1 \sqcup F_2$ is a non-trivial undetectable fault configuration.
For syndrome $\sigma(F_1) = \sigma(F_2)$ and output code syndrome $s_\text{out}(F_1) = s_\text{out}(F_2)$
the deterministic MW decoder will find a minimum-weight fault configuration $F'$.
The logical action of $F'$ must be different from either $F_1$ or $F_2$ and so either $F_1$ or $F_2$ followed by the correction
$F'$ corresponds to a logical error.
We see that when the distance is achieved on an even fault configuration, $d(\faultset)/2$ is a tight (non-inclusive) upper bound 
on the weight of correctable errors. 

Let us now discuss a more nuanced case, when the channel distance is achieved only on odd fault configurations.
Suppose that $F$ is a distance-achieving odd fault configuration and $F_e$ is a positive-weight 
elementary fault configuration $F_e \subset F$.
(Here we use the term elementary fault configuration to describe a fault configuration with a single element.)
Let us write $F = F_1 \sqcup F_2 \sqcup F_e$ as a disjoint union of three sets,
and assume without loss of generality that $\weight(F_1) \ge \weight(F_2)$.
fault configuration $F_1 \sqcup F_e$ will cause error correction to fail.
Suppose that $F_1 \sqcup F_e$ is the fault configuration that actually occurs. 
Note that $F_2$ 
and  $F_1 \sqcup F_e$ will have the same channel syndrome and output code syndrome,
at least one of which should be non-zero. 
The MW decoder will provide a correction $F'$ with weight at most $\weight(F_2) \le (d(\faultset) - \weight(F_e))/2$
with the same syndromes and the same logical action as $F_2$.
If the logical actions were different, this would contradict that $d(\faultset)$ is the fault distance of the channel.
Correcting for $F'$ will cause the logical error equal to the logical action of $F$.
In case of odd fault configurations, there exists a non-correctable fault configuration with weight 
$\weight(F_1 \sqcup F_e) \in (d(\faultset)/2, d(\faultset)]$.
When all the elementary faults have the same weight $w_0$, we can choose $\weight(F_1) = \weight(F_2)$ 
and the above construction leads to a non-correctable fault configuration of weight $(d(\faultset) + w_0)/2$.

\section{Generalized hook faults}
\label{sec:hook-faults}

Here we provide a formal general definition for hook faults which is useful for designing fault-tolerant channels.

An important concept when designing fault tolerant circuits is the notion of a hook fault (sometimes referred to as a `hook error'~\cite{dennis2002tqm}).
To our knowledge, there is no generally applicable definition of a hook fault, but instead they are understood in the context of specific examples.
In the scenarios where hook faults are typically discussed, there is a well-established simplified representation of the circuit, which experiences a simplified `phenomenological model' of faults, where any single fault results in a Pauli errors of weight at most one.
Then, loosely speaking, a hook fault is a single elementary fault in the circuit which has an effect which requires at least two elementary faults in the phenomenological noise model to reproduce. 
Such faults are important because they can have an enhanced corrupting effect on the encoded information, for example by introducing higher-weight Pauli errors.
The name `hook' is used because the two phenomenological faults that reproduce the effect of a hook fault are often represented as edges in a decoding graph, and can visually resemble hooks.
In this section we formalize the concept of a hook fault for general stabilizer channels, and also formally define hazardous hook faults.

We define hook faults for a general stabilizer channel $\mathcal{C}$ with a given elementary fault set $\faultset$ with respect to a fault subset $\faultset_\text{sub} \subset \faultset$, with unit weights for all elementary faults.
Let $d(\faultset)$ and $d(\faultset_\text{sub})$ be the fault distances of the stabilizer channel $\mathcal{C}$ with respect to elementary fault sets $\faultset$ and $\faultset_\text{sub}$ respectively.
Clearly $d(\faultset_\text{sub}) \geq d(\faultset)$, since any non-trivial undetectable fault configuration $F' \subset \faultset_\text{sub}$ with weight $d(\faultset_\text{sub})$ has an equivalent fault configuration $F \subset \faultset$ with weight at most $d(\faultset_\text{sub})$.
We are interested in those elementary faults in $\faultset \setminus \faultset_\text{sub}$ which have the potential to cause $d(\faultset)$ to be strictly less than $d(\faultset_\text{sub})$.
We say that an elementary fault $F \in \faultset$ is a \emph{hook fault} with respect to $\faultset_\text{sub} \subset \faultset$ if there is no equivalent elementary fault of $\faultset_\text{sub}$.
The set of hook faults of $\faultset$ with respect to $\faultset_\text{sub}$ is contained within $\faultset \setminus \faultset_\text{sub}$, and is equal to $\faultset \setminus \faultset_\text{sub}$ if all elementary faults in $\faultset$ are inequivalent.
Note that by this definition, a hook fault can be defined with respect to \emph{any} fault subset, which is a substantial generalization beyond the standard notion of hooks.

\begin{dfn}[Hazardous and brazen hook faults]
Let $F$ be a hook fault of $\faultset$ with respect to $\faultset_\text{sub}$ for the stabilizer channel $\mathcal{C}$, where all faults have unit weight.

We call $F$ a \emph{hazardous} hook fault of $\faultset$ with respect to $\faultset_\text{sub}$ if $F$ is contained in a minimum-weight non-trivial undetectable fault configuration for $\mathcal{C}$ with respect to $\faultset$.

We call $F$ a \emph{brazen} hook fault of $\faultset$ with respect to $\faultset_\text{sub}$ if $F$ is equivalent to a subset of a minimum-weight non-trivial undetectable fault configuration for $\mathcal{C}$ with respect to $\faultset_\text{sub}$.
\end{dfn}

Note that if there are no hazardous hook faults for $\faultset$ with respect to $\faultset_\text{sub}$, then it must be that $d(\faultset) = d(\faultset_\text{sub})$.
Moreover, if $d(\faultset) < d(\faultset_\text{sub})$, it must be the case that there is at least one hazardous hook fault in $\faultset$ with respect to $\faultset_\text{sub}$.
The existence of any brazen hook faults for $\faultset$ with respect to $\faultset_\text{sub}$ implies that $d(\faultset) < d(\faultset_\text{sub})$.

\begin{figure}[h]
    (a)\includegraphics[width=0.4\linewidth]{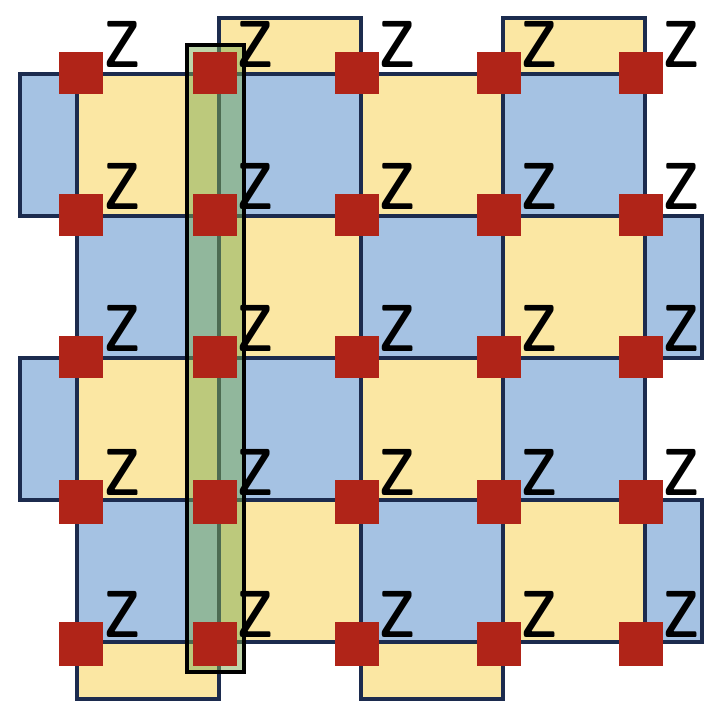}
    (b)\includegraphics[width=0.4\linewidth]{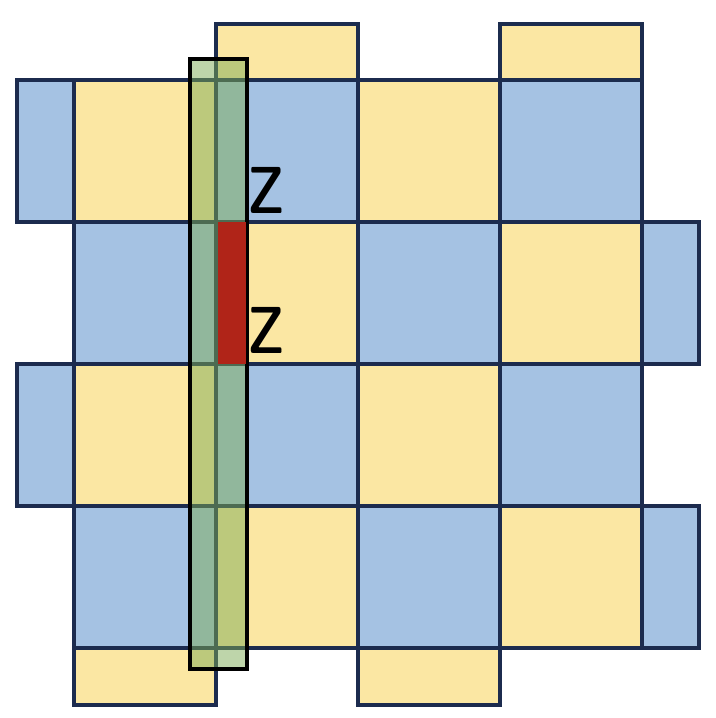}
    (c)\includegraphics[width=0.4\linewidth]{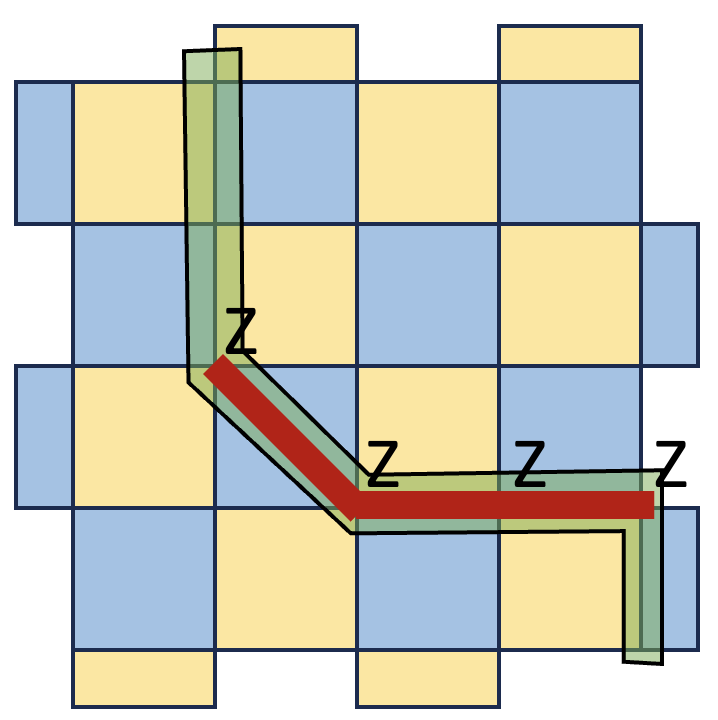}
    (d)\includegraphics[width=0.4\linewidth]{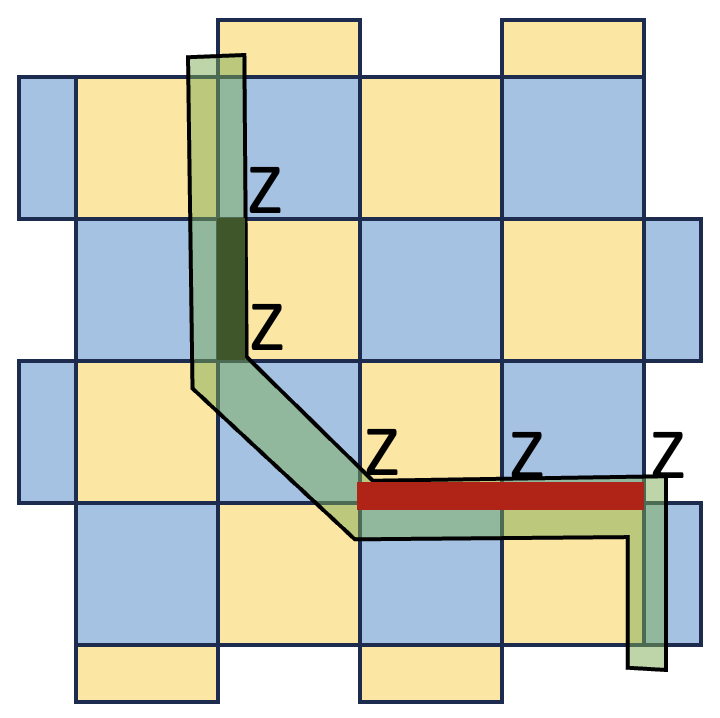}
    \caption{
        Different classes of hook faults illustrated using a single round of surface code stabilizer measurements, with $X$ and $Z$ stabilizers corresponding to blue and yellow plaquettes respectively.
        (a) $\faultset_\text{sub}$ is single-qubit $Z$ faults (red squares). 
        $F$ is a minimum-weight logical fault configuration for $\faultset_\text{sub}$ (green), such that $d(\faultset_\text{sub})=5$. 
        Now consider adding an in-equivalent fault $f$ to $\faultset_\text{sub}$ to form $\faultset$.
        (b) If $f$ (red) is equivalent to the combination of $w>1$ faults in $F$, then the fault distance is reduced to $d(\faultset) = d(\faultset_\text{sub}) - w +1$ and we say $f$ is brazen for $\faultset$ with respect to $\faultset_\text{sub}$.
        (c) Non-brazen faults can also reduce the fault distance however.
        If $f$ (red) is in a logical fault configuration $F'$ for fault set $\faultset$ as shown, we say $f$ is hazardous for $\faultset$ with respect to $\faultset_\text{sub}$.
        In this example $F'$ does not correspond to a minimum-weight logical fault for $\faultset_\text{sub}$, and the hazardous $f$ reduces the distance. 
        (d) Hazardous faults may not reduce the fault distance however.
        In this example we add two faults $f_1$ (green) and $f_2$ (red) to $\faultset_\text{sub}$ to form $\faultset$, which results in $d(\faultset) = d(\faultset_\text{sub})-1$.
        Removing $f_1$ from $\faultset$ (but leaving $f_2$) recovers the distance.  
    }
    \label{fig:hook-fault-definitions}
\end{figure}

These concepts can be used in the design of fault-tolerant stabilizer channels. 
A typical scenario is to work with $\faultset$ as corresponding to the full circuit noise model for a candidate stabilizer channel $\mathcal{C}$, and $\faultset_\text{sub}$ be a subset of fault locations that correspond to a phenomenological noise model that is insensitive to small changes in the circuit.
The standard approach is then to design circuit $\mathcal{C}$ such that $d(\faultset_\text{sub})$ achieves the target fault distance, and if the initial choice of $\mathcal{C}$ does not have a fault distance that matches $d(\faultset_\text{sub})$, one can consider small modifications to $\mathcal{C}$ which change $\faultset$ without changing $\faultset_\text{sub}$ to eliminate different kinds of hook faults until the desired distance is achieved.
However, defining hooks with respect to an arbitrary subset of faults allows us to work in a more general setting, for example not just by changing $\faultset$ while keeping $\faultset_\text{sub}$ fixed, but by also increasing the size of $\faultset_\text{sub}$ in iterations to find and remove brazen hook faults in each iteration.

We make these approaches more concrete in examples in \sec{examples} and consider the computational aspects of these different kinds of hook faults in \sec{calculating-fault-dist}.

\section{Analytic analysis of stabilizer channels}
\label{sec:code-def-channels}

Code deformation involves modifying the code in which information is stored by measuring stabilizer generators for the new code.
Many techniques in the field of quantum error correction are naturally described in terms of code deformation.
In this section, we show that in fact any stabilizer channel can be understood as a sequence of code deformation steps.
By expressing a stabilizer channel in terms of code deformations, the circuit's action is made more transparent, which can make it easier to improve the design of the circuit. 
The code deformation viewpoint is complementary to other approaches that also help elucidate how circuits work in different ways, such as the topological picture described in Ref.~\cite{Bombin2021} for surface codes.

In \sec{code-def-step} we review a single code deformation step and its logical action when considered as a stabilizer channel.
Then, in \sec{code-def-viewpoint-sub} we describe how any stabilizer channel can be reformulated as a sequence of code deformation steps.
In \sec{code-def-channel-logical-action}, we show how one can easily identify the logical action of a stabilizer channel by composing the logical action of each deformation step.
Lastly, in \sec{flip-fault-effect}, we provide an approach to map the effect of an outcome-flip fault in a stabilizer channel to a fault configuration of Pauli errors, which will be useful in our examples in \sec{examples}.

\subsection{An elementary code deformation step}
\label{sec:code-def-step}

Here we review the analysis of a single code deformation step from Ref.~\cite{kliuchnikov2023}.
In \fig{code-def-round}(a) we show an elementary code deformation step, which has an input state stabilized by a Pauil group $\stab_\text{in}$ which undergoes a set of commuting Pauli measurements which generate a group $\mathcal{M}$.
It is convenient to add a conditional Pauli correction $\hat{P}(m)$ to `fix up' the code deformation step to ensure that the output state is independent of the outcome vector $m$ of a set of generators of $\mathcal{M}$~\cite{kliuchnikov2023}.
We call the result a \emph{phase-fixed code deformation step} as shown in \fig{code-def-round}(b).
The output state's stabilizer group contains $\stab_\text{out} = (\stab_\text{in} \cap \mathcal{M}^\perp) \cdot \mathcal{M}$.

\begin{figure}[ht]
    (a)\includegraphics[width=0.43\linewidth]{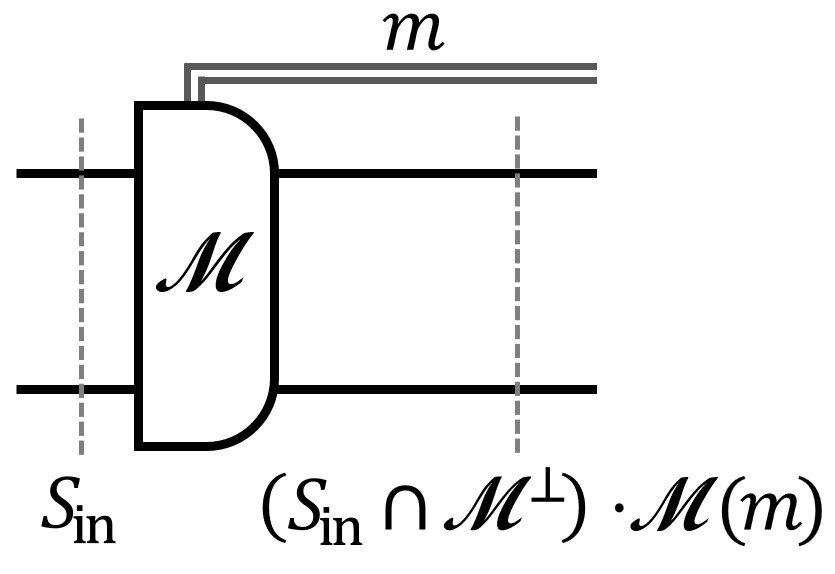}
    (b)\includegraphics[width=0.47\linewidth]{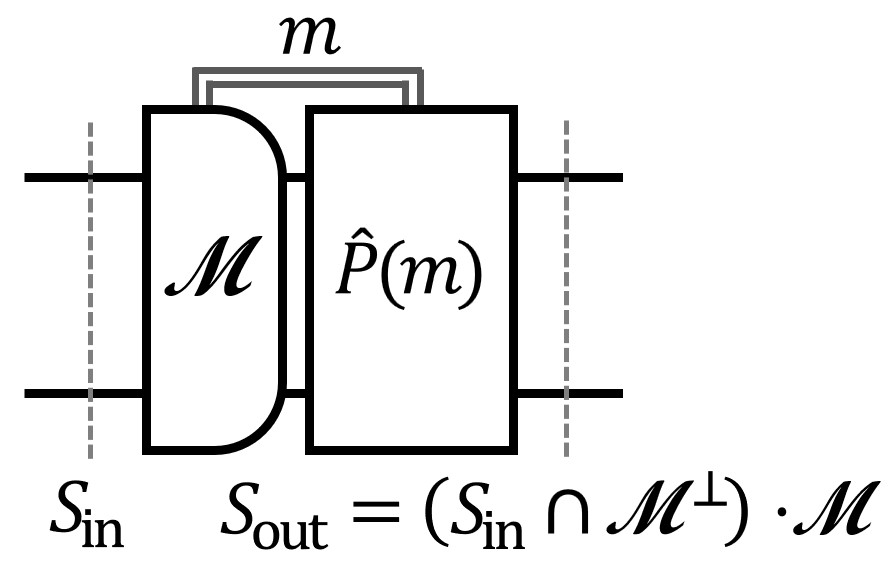}
    (c)\includegraphics[width=0.95\linewidth]{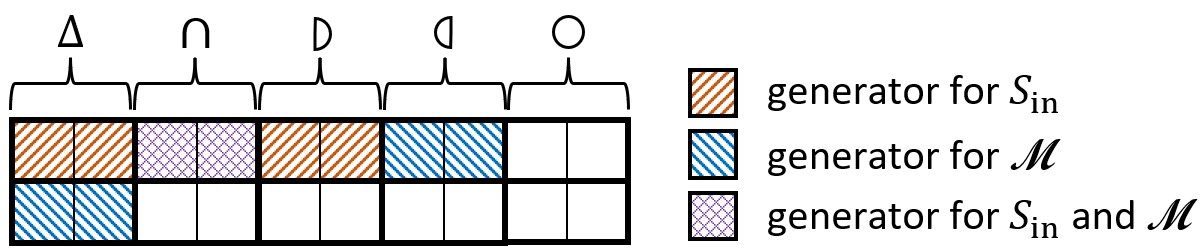}    
    \caption{
    (a) A \emph{code deformation step} consists of a set of commuting Pauli measurements (which generate the group $\mathcal{M}$, producing outcome vector $m$) given an input stabilized by group $\stab_\text{in}$.
    The output state stabilizer group depends on $m$, in particular it contains $\mathcal{M}(m)$ which is the group $\mathcal{M}$, but with some phases changed.
    (b) We can \emph{phase-fix} a code deformation step by adding a conditional Pauli $\hat{P}(m)$ to ensure the output is independent of $m$ and is stabilized by $\stab_\text{out} = (\stab_\text{in} \cap \mathcal{M}^\perp) \cdot \mathcal{M}$. 
    (c) A \emph{common symplectic Pauli basis} $\mathcal{B}(\stab_\text{in},\mathcal{M})$ exists for any two stabilizer groups~\cite{kliuchnikov2023}. 
    The basis has $2n$ Pauli generators (boxes) separated into five sets (symbols) depending on their membership of $\stab_\text{in}$ and $\mathcal{M}$.
    Generators pairwise-commute except for vertical neighbors.
    }
    \label{fig:code-def-round}
\end{figure}

In what follows, we use the terminology and notation from \sec{stab-circuits-logical} on Pauli bases.
The fix-ups $\hat{P}(m)$ are naturally defined in terms of a \emph{common symplectic basis} $\mathcal{B}(\stab_\text{in},\mathcal{M})$ which respect both $\stab_\text{in}$ and $\mathcal{M}$; see \fig{code-def-round}(c).
In this basis, the group $\mathcal{M}$ is generated by the elements of $\mathcal{X}^\Delta$, $\mathcal{Z}^\cap$ and $\mathcal{Z}^{\cleftsemicirc}$, and $\hat{P}(m)$ is chosen such that for each element of $\mathcal{X}^\Delta$ or $\mathcal{Z}^{\cleftsemicirc}$ which has non-trivial outcome, its anticommuting partner from $\mathcal{Z}^\Delta$ or $\mathcal{X}^{\cleftsemicirc}$ is applied respectively.

From the common symplectic basis $\mathcal{B}(\stab_\text{in},\mathcal{M})$, the subset of generators $\mathcal{Z}^{\newcirc}$ and $\mathcal{X}^{\newcirc}$ can each be multiplied by the stabilizer group $\stab_\text{out}$ to form a symplectic basis for the logical operators of the code $\stab_\text{out}$, which we call $[\mathcal{B}(\stab_\text{in},\mathcal{M})]_{\stab_\text{out}}$.
Similarly, we can use $\mathcal{Z}^{\cleftsemicirc},\mathcal{X}^{\cleftsemicirc},\mathcal{Z}^{\newcirc}$ and $\mathcal{X}^{\newcirc}$ to define a symplectic basis $[\mathcal{B}(\stab_\text{in},\mathcal{M})]_{\stab_\text{in}}$ for the logical operators of $\stab_\text{in}$.

The logical action of a phase-fixed code deformation step has a simple general form circuit as shown in \fig{code-def-round-logical}. 
This general form circuit starts with a logical basis change $L^\dagger$ of the code $\stab_\text{in}$ from $[\mathcal{B}_\text{in}]_{\stab_\text{in}}$ to $[\mathcal{B}(\stab_\text{in},\mathcal{M})]_{\stab_\text{in}}$.
Next, single-qubit $Z$ measurements are applied (which measure all the logical operators of $S_\text{in}$ that have a representative in $\mathcal{M}$).
Lastly, a logical basis change $R$ of the code $\stab_\text{out}$ is applied from $[\mathcal{B}(\stab_\text{in},\mathcal{M})]_{\stab_\text{out}}$ to $[\mathcal{B}_\text{out}]_{\stab_\text{out}}$.

An alternative approach to analyzing an elementary code deformation step is discussed in Ref.~\cite{Vuillot2019}.

\begin{figure}[ht]
    \includegraphics[width=1.0\linewidth]{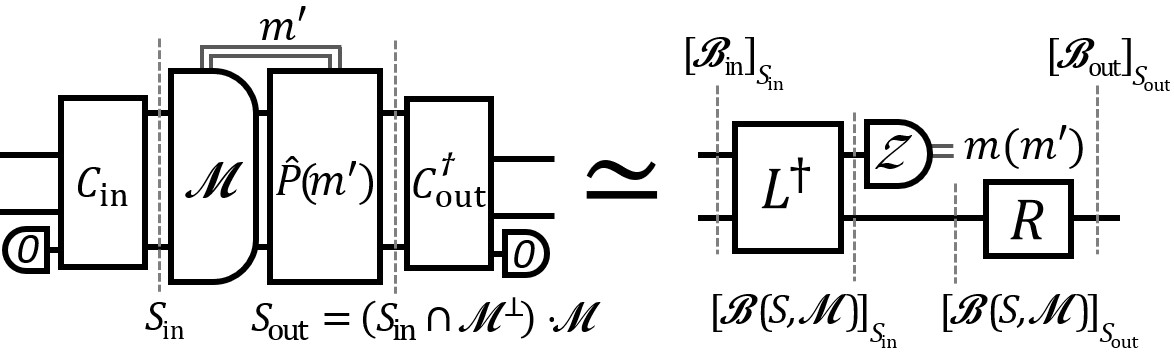}
    \caption{
    The logical action of a phase-fixed code deformation step is captured by a simple general form.
    The $k_\text{in}$-qubit unitary $L$ maps the logical operators of $\stab_\text{in}$ from the common symplectic basis to the basis specified by $C_\text{in}$.
    The $k_\text{out}$-qubit unitary $R$ maps the logical operators of $\stab_\text{out}$ from the common symplectic basis to the basis specified by $C_\text{out}$.
    }
    \label{fig:code-def-round-logical}
\end{figure}

\subsection{Stabilizer channels are code deformations}
\label{sec:code-def-viewpoint-sub}

Here we begin by defining a \emph{code deformation sequence}, which consists of a sequence of code deformation steps.
We go on to explain that in fact all stabilizer channels can be considered as code deformation sequences.

We define a code deformation sequence as depicted in \fig{code-deformation-channel}.
Each of the $T$ steps has a Clifford unitary $U_i$ followed by a set of measurements of commuting Pauli operators which generate a group $\mathcal{M}_i$.
We take $m_i$ to be the outcomes of a generating set for the group $\mathcal{M}_i$.
We also find it convenient to include a conditional phase fix-up $\hat{P}_i(m_i)$ at the end of each step to ensure that the stabilizer group after each step is independent of the measurement outcomes.

\begin{figure}[ht]
    \includegraphics[width=1.0\linewidth]{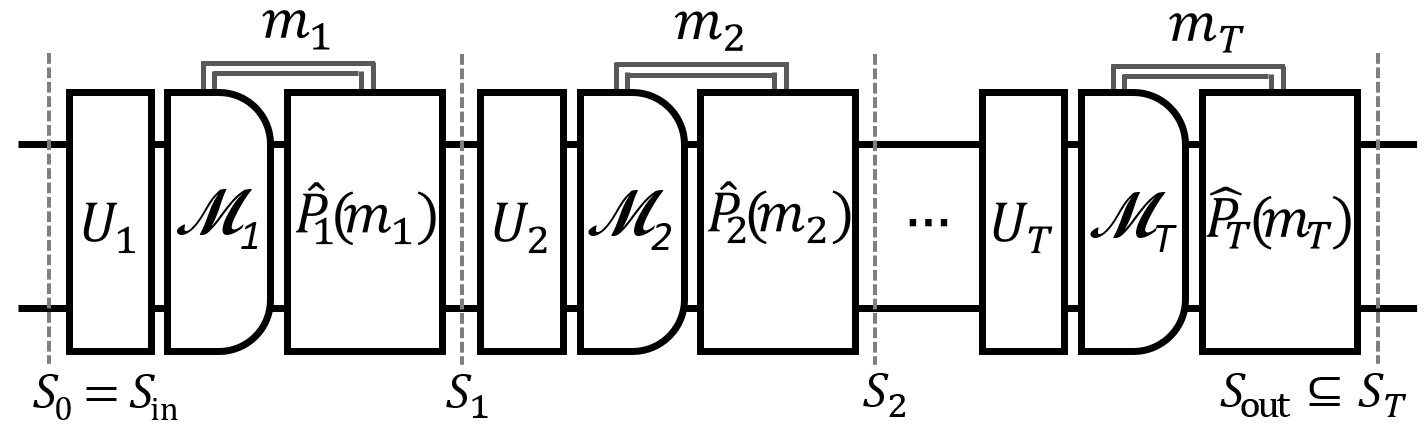}
    \caption{Any stabilizer channel can be expressed as a \emph{code deformation sequence} as shown here, consisting of a sequence of phase-fixed code deformation steps, where the $i$th step involves applying the Clifford unitary $U_i$ before measuring a set of commuting Pauli operators giving outcomes $m_i$ followed by a fix-up $\hat{P}_i(m_i)$ as described in \fig{code-def-round}(b).
    The input and output states of the channel are encoded in the stabilizer codes $\stab_\text{in}$ and $\stab_\text{out}$ respectively.
    }
    \label{fig:code-deformation-channel}
\end{figure}

We consider the channel's logical action with respect to an input code with stabilizer group $\stab_\text{in}$ and an output code with stabilizer group $\stab_\text{out}$ (with their logical bases fixed by encoding circuits $C_\text{in}$ and $C_\text{out}$).
The state after the $i$th phase-fixed code deformation step has a stabilizer group containing $\stab_i$, where $\stab_{i} = ((U_{i} \stab_{i-1} U_{i}^\dagger) \cap \mathcal{M}_i^\perp) \cdot \mathcal{M}_i$ for each $i = 1,\dots, T$ and $\stab_0 = \stab_\text{in}$.
We require for consistency between the channel and the specified output code $\stab_\text{out}$ that $\stab_T \supseteq \stab_\text{out}$.

There are many ways to express a stabilizer circuit $\mathcal{C}$ as an equivalent code deformation sequence $\mathcal{C}_{\text{def}}$.
For example, one can split $\mathcal{C}$ into a sequence of $T$ sub-circuits $\mathcal{C}_j$ each of which has the same number of input and output qubits 
such that $\mathcal{C} = \mathcal{C}_1 \circ \ldots \circ \mathcal{C}_T$.
Each sub-circuit $\mathcal{C}_i$ is equivalent to a code deformation step (not necessarily phase fixed).
This is because each sub-circuit $\mathcal{C}_i$ is equivalent to a general form circuit~\fig{general-form-circuit} (with unitaries $L_i,R_i$ and $k_i$ destructive measurements),
and the general form circuit with the same number of input and output qubits is equivalent to a code deformation step
with $U_i =  R_i L_i^\dagger$ and $\mathcal{M}_i = \langle R_i Z_j R_i^\dagger : j \in [k_i] \rangle$ and appropriately chosen conditional Paulis.
To get the channel in \fig{code-deformation-channel} it remains to adjust the conditional Paulis to ensure that every step is phase-fixed.

\subsection{Logical action of stabilizer channels}
\label{sec:code-def-channel-logical-action}

When a stabilizer channel has been expressed as a code deformation sequence, it is straight-forward to compute its logical action analytically.
For any code deformation sequence, the logical action is equivalent to the simple general form circuit shown on the left of \fig{code-deformation-channel-logical}.
Note that this is a restriction of the general form circuit shown in \fig{general-form-circuit}, having no random bit outcomes or conditional Pauli (by virtue of the phase fixing in the channel).
The logical action of a code deformation sequence is fully specified by a $k_\text{in}$-qubit Clifford unitary $L$, an integer $k$ that specifies the number of independent logical Paulis measured by the channel, and a $k_\text{out}$-qubit Clifford unitary $L$.
The \emph{logical outcome function} $\mathfrak{f}(O)$ identifies what the $k$-bit logical outcome vector $m$ of the channel is given the channel outcome vector $O$.

\begin{figure}[ht]
    \includegraphics[width=1.0\linewidth]{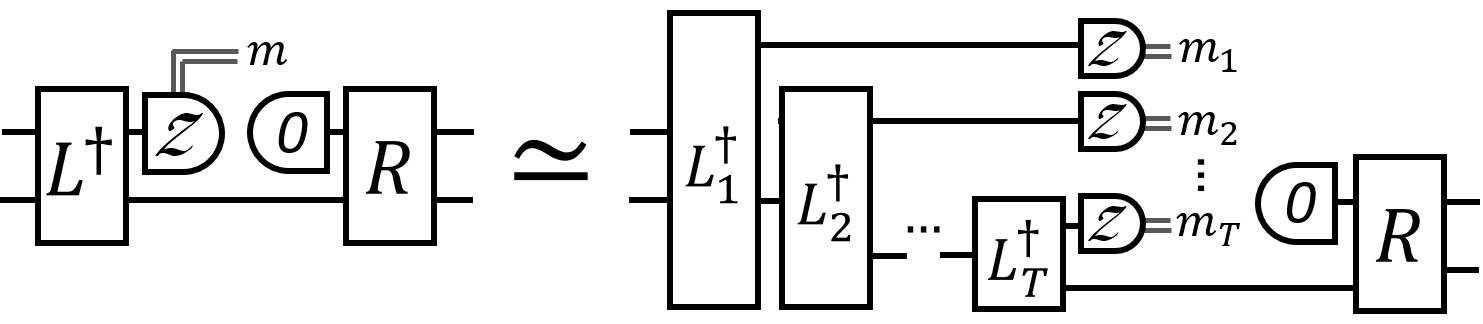}
    \caption{
    The logical action of any code deformation sequence (with phase fixes) is captured by a general form circuit as shown on the left, with no random bits or conditional Pauli.
    The Clifford unitary $L$ can be expressed in terms of products of the logical basis change Cliffords $L_j$ (expressed in carefully chosen bases as described in the text) of each non-trivial code deformation step.
    The Clifford unitary $R$ can be expressed as a code basis conversion from $\stab_T$ tp $\stab_\text{out}$.
    There exists a map from the channel outcomes for each step $m_1, m_2, \dots m_T$ to the logical outcome vector $m$. 
    }
    \label{fig:code-deformation-channel-logical}
\end{figure}

The general form circuit for the logical action of the channel makes it clear that any code deformation sequence simply measures out a set of commuting logical Pauli operators of the input code, and then re-encodes the remaining logical information in the output code in some basis.
The inclusion of the zero-state ancilla in the general form allows for some logical degrees of freedom of the output code to be fixed.

As illustrated in \fig{code-deformation-channel-logical}, the Cliffords $L$ and $R$ which specify this general form can be found from more elementary logical basis change matrices.
To define these basis changes, we introduce a sequence of symplectic Pauli bases $\mathcal{B}_0, \mathcal{B}_1, \dots, \mathcal{B}_T$, where for $j = 1,2, \dots T$, we define $\mathcal{B}_j = \mathcal{B}(U_{j} \stab_{j-1}U_{j}^\dagger,\mathcal{M}_j)$ as a common symplectic basis for $U_{j} \stab_{j-1}U_{j}^\dagger$ and $\mathcal{M}_j$, and we take $\mathcal{B}_0 = \mathcal{B}_\text{in}$ to be the Pauli basis specified by the input code encoding Clifford $C_\text{in}$.
To use~ \fig{code-def-round-logical} we need to specify with respect to which input and output symplectic bases each of the measurement rounds is analyzed. 
For round $j$, we choose the input and output pair to be $U_j \mathcal{B}_{j-1} U_j^\dagger$ and $\mathcal{B}_j$.
With the above choice of input and output symplectic bases, for all $j = 1,2, \dots T$, $L^\dagger_j$ is a logical basis change from $[U_j \mathcal{B}_{j-1} U_j^\dagger]_{U_j S_{j-1} U_j^\dagger}$
to $[\mathcal{B}_j]_{U_j S_{j-1} U_j^\dagger}$,
and $R_j$ is a logical basis change of the code $\stab_{j}$ from $[\mathcal{B}_{j}]$ to $[\mathcal{B}_{j}]$ and is therefore equal to the identity.
The $R$ Clifford in the general form on the left of \fig{code-deformation-channel-logical} is then identified as a code basis conversion from $\stab_T$ (with encoding Clifford specified by $\mathcal{B}_{j}$) to $\stab_\text{out}$ (with encoding Clifford specified by $C_\text{out}$).

Note that with these choices of bases, the logical action of $U_j$ is trivial with respect to the input code $S_{j-1}$ with logical operator basis $[\mathcal{B}_{j-1}]_{\stab{j-1}}$
and output code $U_j S_{j-1} U^\dagger_j$ with logical operator basis $[U_j \mathcal{B}_{j-1} U^\dagger_j]_{U_j \stab{j-1} U^\dagger_j}$.

\subsection{Fault effects of measurement outcome flips}
\label{sec:flip-fault-effect}

Re-expressing the fault effect of measurement outcome flips in a circuit as fault configurations of Pauli operators with the same effect can be a useful technique. 
The main tool for this is the notion of an \emph{absolutely trivial fault configuration}. 
Given a stabilizer channel $\mathcal{C}$ and a sub-circuit $\mathcal{C}'$ of $\mathcal{C}$, a fault configuration $F$ within $\mathcal{C}'$ is absolutely trivial if the fault-free sub-circuit $\mathcal{C}'$ implements the same linear map as the sub-circuit $\mathcal{C}'$ with the fault configuration $F$ for all possible measurement outcomes of $\mathcal{C}'$. 
Every absolutely trivial fault configuration is also a trivial fault configuration.

For example, consider a sub-circuit $\mathcal{C}'$ that measures a Pauli $Z$ on some qubit, with a fault configuration of a Pauli $Z$ right before the measurement. 
This is a simple example of a trivial fault configuration.

In Section \ref{sec:examples}, we use absolutely trivial fault configurations to analyze the fault-tolerant properties of code deformation rounds. Specifically, consider a code deformation round that consists of measuring a commuting set of Pauli operators $P_1,\ldots,P_m$. 
Let $Q_j$ be a Pauli operator that anti-commutes with $P_j$ and commutes with the rest of the measured Pauli operators. The fault configuration that consists of (1) a Pauli $Q_j$ before measuring $P_1,\ldots,P_m$, (2) an outcome flip of $P_j$, and (3) a Pauli $Q_j$ after measuring $P_1,\ldots,P_m$, is an absolutely trivial fault configuration. 
This allows us to analyze the outcome flip of $P_j$ in terms of an equivalent fault configuration consisting of Pauli operators.

\section{Combining stabilizer channels}
\label{sec:combining-stabilizer-channels}

The goal of this section is to establish sufficient conditions on a set of stabilizer channels, so that 
a stabilizer channel composed from the set of stabilizer channels has fault distance $d$.
The key technical condition we require of the stabilizer channels is that they have a specific time-locality property.
In what follows, we provide a rigorous definition of this property and show that for a stabilizer channel which satisfies it, the simple fault configurations must be contiguous in time.
This will allow us to localize any undetectable fault configuration in a composed stabilizer channel to a pair of consecutive stabilizer channels in the composition.

\begin{dfn}[Time-local stabilizer channels]
Consider a stabilizer channel with elementary fault set $\faultset$.
Let $\faultset = \faultset_1 \sqcup \faultset_2 \sqcup \ldots \sqcup \faultset_T$ be a partition of the elementary fault set into disjoint subsets. 
We say that the stabilizer channel with checks $\checkset$ is \emph{time-local} with respect to this elementary fault set partition if the check set can be partitioned as $\checkset = \checkset_1 \sqcup \checkset_2 \sqcup \ldots \sqcup \checkset_T$ such that: 
(a) checks in $\checkset_1$ depend only on the input code syndrome $s_\text{in}$ of the input error $E_\text{in}$ and elementary faults in $\faultset_1$,
and $s_\text{in}$ can be calculated from checks in $\checkset_1$ when there are no faults,
(b) checks in $\checkset_j$ depend only on elementary faults in $\faultset_{j-1}\sqcup \faultset_j$ for $j \in [2,T]$
(c) $s_\text{out}$ depends only on the checks $\checkset$ and elementary faults from $\faultset_T$.
\end{dfn}

Let us next show that simple fault configurations must be contiguous in time in a time-local stabilizer channel.
Consider a simple fault configuration $\cF$ and suppose that it is a subset of 
$\faultset_{j_1} \sqcup \ldots \sqcup \faultset_{j_m}$ for $j_1 < j_2 < \ldots < j_m$, and that $\cF \cap \faultset_{j_k}$ is non-empty for all $k \in [m]$.
We show that for each $k \in [2,m]$ it must be that $j_k = j_{k-1} + 1$, that is the fault configuration $\cF$ is contiguous in time.
Suppose that this is not the case for some $k$, 
then fault configurations $\cF \cap (\faultset_{j_1} \sqcup \ldots \sqcup \faultset_{j_{k-1}})$ 
and $\cF \cap (\faultset_{j_{k}} \sqcup \ldots \sqcup \faultset_{j_m})$ must both both be undetectable, which is a contradiction to $\cF$ being simple.

The composition of two time-local stabilizer channels is again a time-local stabilizer channel.
To see this, consider stabilizer channels $\mathcal{C}^\alpha$ and $\mathcal{C}^\beta$ that are time-local with respect to partitions of the elementary fault sets
$\faultset^\alpha = \{ \faultset^\alpha_j \}_{j\in[T^\alpha]}$
and
$\faultset^\beta = \{ \faultset^\beta_j \}_{j\in[T^\beta]}$,
and of the check bases 
$\checkset^\alpha = \{ \checkset^\alpha_j \}_{j\in[T^\alpha]}$ 
and
$\checkset^\beta = \{ \checkset^\beta_j \}_{j\in[T^\beta]}$.
We require that $\stab^\beta_\text{in} = \stab^\alpha_\text{out}$, and consider them with respect to the same symplectic Pauli basis.
We seek to show that the composed channel $\mathcal{C}^\alpha \circ \mathcal{C}^\beta$ is time-local with respect to the partition 
$\{ \faultset^\alpha_j \}_{j\in[T^\alpha]} \sqcup \{ \faultset^\beta \}_{j\in[T^\beta]}$
of elementary fault set $\faultset^\alpha \sqcup \faultset^\beta$.
We must also specify the set of checks and their partition for the composed channel $\mathcal{C}^\alpha \circ \mathcal{C}^\beta$.
Recall that the checks in $\checkset^\beta_1$ depend on the input code syndrome $s^\beta_\text{in}$ of the input error $E^\beta_\text{in}$ and faults in $\faultset^\beta_1$.
Also note that $s^\beta_\text{in} = s^\alpha_\text{out}$ is a function of faults in $\faultset_{T^\alpha}$ and checks in $\checkset^\alpha$.
For this reason $s^\beta_\text{out}$ can be expressed as a function of checks in $\checkset^\alpha \sqcup \checkset^\beta$ and faults in $\faultset^\beta_{T^\beta}$.
We can modify $\checkset^\beta_1$ into $\tilde \checkset^\beta_1$ by combining it with checks in 
$\checkset^\alpha$, such that $\tilde \checkset^\beta_1$ depends on $\faultset^\alpha_{T^\alpha} \sqcup \faultset^\beta_1$.
The check basis partition $\{ \checkset^\alpha_j \}_{j\in[T^\alpha]} \sqcup \{  \tilde \checkset^\beta_1 \} \sqcup \{ \checkset^\beta \}_{j\in[2,T^\beta]}$
satisfies the properties needed to establish the time-locality of $\mathcal{C}^\alpha \circ \mathcal{C}^\beta$.

Note that when $T = 1$ the time-locality condition becomes very simple. 
It requires that the input syndrome $s_\text{in}$ can be inferred from the checks when no faults occur in the channel.
Indeed, if this is the case we can re-express the output syndrome as a linear function of the fault vector and checks.

For the rest of this section, let us consider the case of uniform weights. 
In this case, a simple fault configuration of weight $d$ in a time-local stabilizer channel can span at most $d$ consecutive elements of the partition $\faultset_{j+1} \sqcup \ldots \sqcup \faultset_{j+d}$ for some $j$.
Otherwise its weight would be greater than $d$.
This observation is crucial to establish the conditions under which stabilizer channel composition is distance preserving.

\begin{thm}
\label{thm:stabilizer channel-composition}
Let $\mathcal{C}^\alpha$, $\mathcal{C}^\beta$, $\mathcal{C}^\gamma$ be a sequence of compatible~\footnote{
We say that the stabilizer channels in a sequence are compatible if the output stabilizer code of each channel in the sequence matches the input stabilizer code of the next channel in the sequence.
} time-local stabilizer channels with unit-weight faults.
Suppose that compositions $\mathcal{C}^\alpha \circ \mathcal{C}^\beta$, $\mathcal{C}^\beta \circ \mathcal{C}^\gamma$ are stabilizer channels
with fault distance at least $d$ 
and the partition with respect to which $\mathcal{C}^\beta$ is time local is at least of size $d$, 
then stabilizer channel $\mathcal{C}^\alpha \circ \mathcal{C}^\beta \circ \mathcal{C}^\gamma$ has fault distance at least $d$. 
\end{thm}

\begin{proof}
We prove the result by contradiction.
Suppose that the stabilizer channel composition has a non-trivial undetectable fault configuration $\cF$ of weight $d-1$.
Without loss of generality, we choose $\cF$ to be simple. 
Note that that stabilizer channel composition $\mathcal{C}^\alpha \circ \mathcal{C}^\beta \circ \mathcal{C}^\gamma$ is time-local stabilizer channel
with respect to the partition
$\{ \faultset^\alpha_j \}_{j\in[T^\alpha]} \sqcup \{ \faultset^\beta_j \}_{j\in[T^\beta]} \sqcup \{ \faultset^\gamma_j \}_{j\in[T^\gamma]}$
of the elementary fault set 
$\faultset^\alpha \sqcup \faultset^\beta \sqcup \faultset^\gamma$.
Fault configuration $\cF$ must be supported on at most $d-1$ contiguous time partitions, and $T^\beta \ge d$,
therefore $\cF$ must be supported on either (A) $\faultset^\alpha \sqcup \faultset^\beta$
or (B) $\faultset^\beta \sqcup \faultset^\gamma$.

Next we show that the fault configuration $\cF$ corresponds to a simple non-trivial fault configuration of either $\mathcal{C}^\alpha \circ \mathcal{C}^\beta$ in case (A), or of $\mathcal{C}^\beta \circ \mathcal{C}^\gamma$ in case (B).
In both cases, this forms a contradiction since $\cF$ has weight $d-1$ but both $\mathcal{C}^\alpha \circ \mathcal{C}^\beta$ and $\mathcal{C}^\beta \circ \mathcal{C}^\gamma$ have fault distance at least $d$.
In case (A), the output syndrome of $\mathcal{C}^\beta$ is zero because it is equal to the input syndrome of $\mathcal{C}^\gamma$ in the stabilizer channel composition and
time locality of $\mathcal{C}^\gamma$ implies that the only way for the output syndrome of $\mathcal{C}^\gamma$ to be zero in the absence of faults in $\mathcal{C}^\gamma$ is for the input syndrome to be zero.
Similarly, the only way $\cF$ can have non-trivial logical effect with respect to $\mathcal{C}^\alpha \circ \mathcal{C}^\beta \circ \mathcal{C}^\gamma$ is if it has non-trivial logical effect with respect to $\mathcal{C}^\alpha \circ \mathcal{C}^\beta$. 
In case (B), the input syndrome of $\mathcal{C}^\beta$ is zero because there are no faults in the part $\mathcal{C}^\alpha$ of the composition and there are no logical errors caused by $\mathcal{C}^\alpha$. 
The only way for $\cF$ to be a non-trivial undetectable fault configuration is if it is 
a non-trivial undetectable fault configuration in $\mathcal{C}^\beta \circ \mathcal{C}^\gamma$, which completes the proof.
\end{proof}

\begin{figure}[h]
    \includegraphics[width=1.0\linewidth]{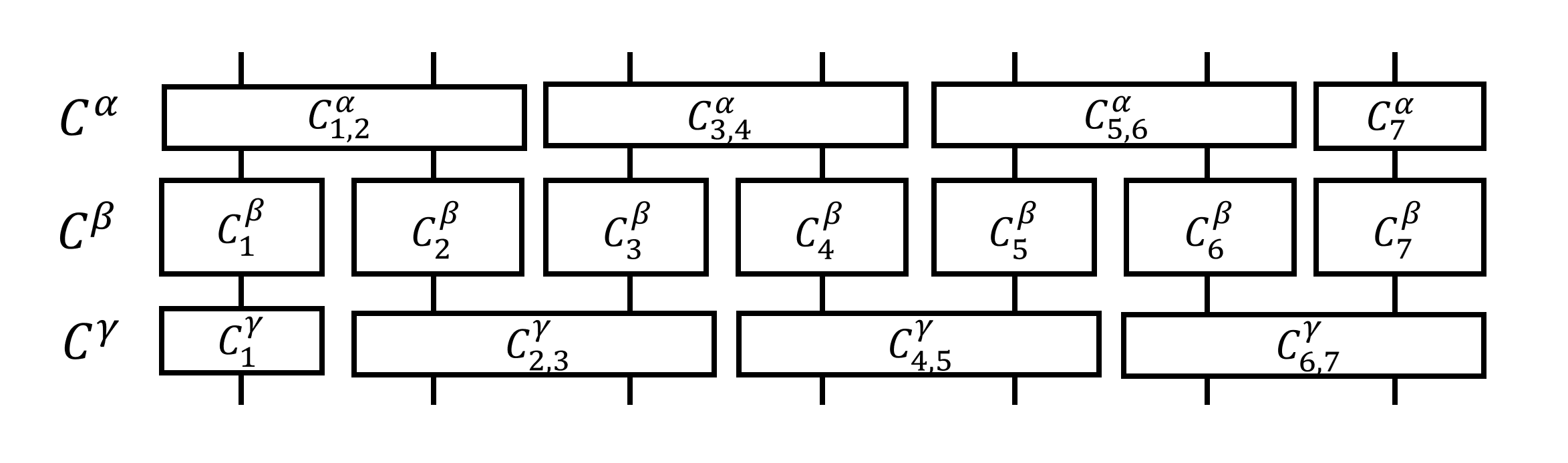}
    \caption{
    An example of applying Theorem~\ref{thm:stabilizer channel-composition} to a stabilizer channel composed of elementary stabilizer channels in space and time. 
    Assume that all the elementary stabilizer channels $\mathcal{C}^{\alpha}_{j},\mathcal{C}^{\alpha}_{j,j+1},\mathcal{C}^{\gamma}_{j},\mathcal{C}^{\gamma}_{j,j+1}$
    are time-local, and $\mathcal{C}^{\beta}_{j}$ are time-local with fault-set partition size at least $d$.
    Additionally assume that the stabilizer channel composition $\mathcal{C}^{\beta}_{j} \circ \mathcal{C}^{\gamma}_{j}$,
    $\mathcal{C}^{\alpha}_{j} \circ \mathcal{C}^{\beta}_{j}$, 
    $\mathcal{C}^{\alpha}_{j,j+1} \circ (\mathcal{C}^{\beta}_{j} \otimes \mathcal{C}^{\beta}_{j+1})$,
    and $ (\mathcal{C}^{\beta}_{j} \otimes \mathcal{C}^{\beta}_{j+1}) \circ \mathcal{C}^{\gamma}_{j,j+1}$ have fault distance $d$.
    Then the overall stabilizer channel $C^\alpha \circ C^\beta \circ C^\gamma$ has fault distance $d$ because $C^\alpha \circ C^\beta$
    and $C^\beta \circ C^\gamma$ satisfy conditions of Theorem~\ref{thm:stabilizer channel-composition},
    as the tensor products of distance $d$ time-local stabilizer channels. 
    } 
    \label{fig:space-composition}
\end{figure}

We conclude with a set of conditions on a set of stabilizer channels that ensure that their 
compositions have fault distance $d$. 
For simplicity, we consider the setting of stabilizer channels with the same number of input and output logical qubits, acting on one or two logical qubits
encoded in the same code, although it is straightforward to generalize these conditions to broader settings.
Consider two finite sets of time-local stabilizer channels $\mathfrak{C}^\text{short}$
and $\mathfrak{C}^\text{long}$. 
Stabilizer channels in $\mathfrak{C}^\text{long}$ are required to be time-local with fault-set partition of at least size $d$ and act on one logical qubit.
Additionally, we require that for any pair of one-qubit stabilizer channels $C^\alpha$ from $\mathfrak{C}^\text{short}$, $C^\beta$ from $\mathfrak{C}^\text{long}$
compositions $C^\alpha \circ C^\beta$, $C^\beta \circ C^\alpha$ have fault distance $d$. 
Similarly, we require that for any two-qubit stabilizer channel $C^\alpha$ from $\mathfrak{C}^\text{short}$, and one-qubit stabilizer channels $C^\beta,C^{\gamma}$ from $\mathfrak{C}^\text{long}$
compositions $C^\alpha \circ (C^\beta \otimes C^\gamma)$, $(C^\beta  \otimes C^\gamma) \circ C^\alpha$ have fault distance $d$. 
Theorem~\ref{thm:stabilizer channel-composition} implies that any composition of above stabilizer channels that interleaves short and long stabilizer channels is a
stabilizer channel of fault distance at least $d$, as illustrated by Figure~\ref{fig:space-composition}.
Above argument can be extended to the stabilizer channels that have different number of input and output logical qubits, 
as well as the stabilizer channel using multiple different codes for the logical qubits.

Note that the tools of this section can be used to analyze arbitrary arbitrary compositions of channels from a set of one-time-local distance-$d$ channels, at the expense of analyzing all possible compositions of $d+1$ such channels.

\section{Algorithms for fault distance and hook faults}
\label{sec:calculating-fault-dist}

Given a stabilizer channel and a local stochastic Pauli noise model with fault set $\faultset$, one can, in principle, calculate the fault distance directly by enumerating faults to find the lowest weight non-trivial fault.
However, the time complexity of such a brute-force approach will be exponential in the number of elementary faults in $\faultset$. 
Sometimes it is possible to express the channel check matrix as a graph in which the undetectable faults correspond to cycles, allowing efficient calculation of the fault distance. 
In \sec{graphic-channel-check} we review how this \emph{syndrome graph} can arise for a stabilizer channel, and point out existing efficient algorithms developed in the 1960s-1980s (which we believe are not widely known in the quantum computing community) that can be used to put a syndrome matrix into graph form.
In \sec{fault-dist-algo} we review existing algorithms to efficiently compute the fault distance given a syndrome graph.
In \sec{hook-fault-algo} we provide new simple algorithms to identifying hazardous hook faults and brazen hook faults.
In \sec{beyond-graphic} we review a number of approaches that can be used to find or bound the channel distance even when it is not possible to express the channel check matrix as a graph.
In the same section, we also point out some new connections between computing channel check distance and binary matroids.

\subsection{Channel check and logical effect matrices}
\label{sec:channel-check}
There are a number of approaches to efficiently compute complete set of channel checks, channel check matrix
and logical effect matrix for an arbitrary stabilizer circuit
\footnote{It it common to explicitly specify checks for a fault-tolerant circuit, here we review approaches that do not require that.}
One approach is based on space-time and outcome code of a stabilizer circuit~\cite{delfosse2023spacetime,delfosse2023simulation}.
The algorithms for a general form of a stabilizer circuit and the circuit's logical action from~\cite{kliuchnikov2023} can also be used for this purpose. 

\subsection{Graph-like channel check matrices}
\label{sec:graphic-channel-check}

Consider a channel which has a channel check matrix $A_{\checkset}$, expressed over the check basis $\checkset$.
We say the channel check matrix $A_{\checkset}$ is \emph{graph-like} if there exists a choice of check basis $\checkset'$ with the same span as $\checkset$ such that each column of $A_{\checkset'}$ has Hamming weight one or two.
The \emph{syndrome graph} $G=(V,E)$ is then constructed from $A_{\checkset'}$ as follows. 
A vertex is included in the vertex set $V$ for each row of $A_{\checkset'}$, along with an additional \emph{boundary vertex}. 
For each weight-2 column of $A_{\checkset'}$, an edge connecting the two vertices corresponding to that row's support is included in the edge set $E$. 
For each weight-1 column of $A_{\checkset'}$, an edge connecting the boundary vertex to the vertex corresponding to that row's support is included in $E$.
By construction, every edge $e$ in $E$ corresponds to an elementary fault $F$ in $\faultset$ and there is a logical effect associated with every edge $r(e) = A_L v_F$, and the edge has a weight given by $\weight(F)$ as defined in \sec{stochastic-noise}.

To relate undetectable faults to cycles in the graph $G$ we recall the definition of the cycle space of a graph.
We call a subset of edges $\mathcal{E}' \subset \mathcal{E}$ a cycle if the sub-graph $(V,\mathcal{E}')$ has even-degree vertices. 
With every set of edges we associate an indicator vector $v_\mathcal{E}$.
As the symmetric difference of two cycles is again a cycle, vectors $v_\mathcal{E}$ for which $\mathcal{E}$ is a cycle form a vector space over $\mathbb{F}_2$ \emph{the cycle space of the graph}~\cite{Diestel2017}.
Importantly, there exist efficient algorithms that can determine if $A_\checkset$ is graph-like and construct the corresponding decoding graph~\cite{Fujishige1980,Bixby1980,Truemper2014,Walter2012}.

\subsection{Graph-based fault distance algorithms}
\label{sec:fault-dist-algo}

There are a number of existing approaches that can be used to compute the fault-distance given a decoding graph $G$. 
One approach is to modify Dijkstra's algorithm, which is an approach alluded to in and implemented in \cite{Gidney2021stim}.

A second approach is to generalize the algorithm for finding minimal weight logical operators of CSS surface codes from \cite{Breuckmann2017} (attributed therein to a communication with Sergei Bravyi).
In this approach, one uses the fact that each row $(A_L)_i$ of the logical effect matrix $A_L$ defines a set $E_i \subset E$ of edges of the decoding graph $G$.
A minimal weight non-trivial fault then corresponds to a cycle that has odd overlap with at least one of the $E_i$s.
Finding a minimal weight cycle that has an odd overlap with a subset of edges $E_i$ is an example of the shortest odd cycle problem \cite{DePinaThesis}.
A key technique used to solve such problems is to reduce the shortest odd cycle problem in a graph $(V,E)$
to $|V|$ shortest path problems in a related graph with twice the number of vertices and edges, constructed based on $E_i$ \cite{DePinaThesis}. 
For completeness, we provide \algo{shortest-odd-cycle} for shortest odd cycle problem in \app{shortest-odd-cycle}.
Notably, this approach generalizes slightly beyond graph-like matrices, as discussed in \ref{sec:beyond-graphic}.

The asymptotic runtime of both approaches to computing the fault-distance given a decoding graph is the same.
For $n_L$ being the number of rows of the logical effect matrix, the runtime is $O(n_L(|V||E|+|V|^2\log|V|))$.

If something is known about the syndrome graph $G$ and its non-trivial cycles, one can often speed up both algorithms.
For example, a common scenario for topological codes is that the syndrome graph is a cellulation of an underlying smooth topological manifold, in which case the cycles associated with logical fault configurations correspond to non-contractible loops, while cycles with trivial action on the channel correspond to contractible loops.
In some such cases, one can be sure that any non-trivial cycle will pass through the boundary vertex, and as such we can considerably reduce the run time of the algorithm by only considering
the cycles that pass through the vertex.
If the upper-bound on the weight of minimal cycles is known, one can also speed-up the algorithms.
We provide further details on how to speed-up the second approach in~\app{shortest-odd-cycle}.

\subsection{Hook fault algorithm}
\label{sec:hook-fault-algo}

Here we provide \algo{fault-set-in-min} that can be used to find all hazardous and brazen hook faults (as defined in \sec{hook-faults}) for a channel with fault set $\faultset$ with respect to a fault subset $\faultset_\text{sub} \subset  \faultset$.
This algorithm requires a subroutine which calculates the fault-distance of the channel with fault sets that are subsets of $\faultset$, and if that subroutine is efficient, so too is \algo{fault-set-in-min}.
When the channel check matrix of $\faultset$ is graph-like, an efficient subroutine is provided by that described in \sec{fault-dist-algo}.

\algo{fault-set-in-min} is actually quite general and we believe it could be of broad interest beyond hook faults.  
The algorithm identifies if a given weight-$w$ fault configuration $F$ is in a minimal logical fault configuration for the channel with fault set $\faultset$ with unit weights.
A subroutine which finds the fault distance $d(\faultset)$ (such as that described in \sec{fault-dist-algo}) may be constructive and find a particular minimum-weight logical fault configuration, making it possible to check if $F$ is contained in the discovered minimum-weight logical configuration.
However there could be many minimal configurations (with some containing $F$ and others not) and we would not know which the fault distance algorithm will discover.
We can overcome this however by working with an augmented fault set $\tilde{\faultset}$ which is formed from $\faultset$ by including an additional elementary fault $\tilde{F}$ which has the same syndrome and logical effect as the fault configuration $F$, but with a weight $w-\epsilon$, where $0 < \epsilon < 1$.
This ensures that if $F$ is contained in \emph{any} minimal (weight-$d(\faultset)$) fault configuration of the channel with $\faultset$, then $\tilde{F}$ will be contained in a minimum-weight logical fault configuration of the channel with $\tilde{\faultset}$, which will have weight $d(\tilde{\faultset}) = d(\faultset)-\epsilon$. 
Whereas if $F$ is not contained in a minimal fault configuration of the channel with $\faultset$, then any logical fault configuration containing $\tilde{F}$ of the channel with $\tilde{\faultset}$ must have weight greater than $d(\faultset)+1 -\epsilon > d(\faultset)$, and as such $d(\tilde{\faultset}) = d(\faultset)$.
This forms the basis of 
\algo{fault-set-in-min}.

\begin{algorithm}[H]
    \caption{Inclusion in min-weight configuration}    
    \label{alg:fault-set-in-min}
    \begin{algorithmic}[1]
    \Require Fault set $\faultset$ for a channel, 
    \Blank fault configuration $F \subset \faultset$ with weight $w$, 
    \Blank subroutine to find fault-distance.
    \Ensure True or False: $F$ is in a min-weight logical fault configuration for the channel with $\faultset$ with unit weights.\
    \State Find fault-distance $d(\faultset)$ of the channel with fault set $\faultset$, with all elements of $\faultset$ having unit weight.
    \State Define fault set $\tilde{\faultset}$ by appending to $\faultset$ an elementary fault $\tilde{F}$ with the same syndrome and logical action as $F$.
    \State Find fault-distance $d(\tilde{\faultset})$ of the channel with fault set $\tilde{\faultset}$, with all elements of $\faultset$ having unit weight, except $\tilde{F}$ which has weight $w - \epsilon$ for $0 < \epsilon < 1$.
    \If{$d(\tilde{\faultset})  = d(\faultset) - \epsilon$}
        \Return{True.}
    \Else{}
        \Return{False.}
    \EndIf 
    \end{algorithmic}
\end{algorithm}

To check if a given hook fault $F \in \faultset \setminus \faultset_\text{sub}$ is hazardous we must know whether or not $F$ is in a minimal weight non-trivial fault configuration for $\faultset$ with unit weights.
This can be checked straight-forwardly using \algo{fault-set-in-min} with inputs being the full fault set $\faultset$ and $F$ with weight 1 (since $F$ is an elementary fault in $\faultset$).

To check if a given hook fault $F \in \faultset \setminus \faultset_\text{sub}$ is brazen we need to solve three sub-problems.
Sub-problem (i): Check if $F$ is equivalent to some fault configuration in $\faultset_\text{sub}$.
This is readily solved by Gaussian elimination.
Sub-problem (ii): Find a minimal-weight configuration $F' \subset \faultset_\text{sub}$ equivalent to $F$, with weight $w'$.
Sub-problem (iii): 
Check that a minimal weight non-trivial fault configuration for $\faultset_\text{sub}$ contains the fault configuration $F'$.
Sub-problem (iii) can be checked straight-forwardly using \algo{fault-set-in-min} with inputs being the fault subset $\faultset_\text{sub}$ and $F'$ with weight $w'$ (which has weight $w'$ since it is an elementary fault in $\faultset$) as inputs.

Note that if the channel check matrix for $\faultset$ is graph-like then the efficient fault-distance algorithm described in \sec{fault-dist-algo} can be used as the subroutine in \algo{fault-set-in-min} to identify hazardous hook faults and to solve sub-problem (iii) for brazen hook faults.
Furthermore, when $\faultset$ is graph-like we can efficiently solve the additional sub-problem (ii) using a minimum-weight decoder needed to identify brazen hook faults, provided that the weight $w'$ of $F'$ satisfies $w' < d(\faultset)-1$.
As we will see, this solution of sub-problem (ii) also solves sub-problem (i) more efficiently than Gaussian elimination as a bi-product.
Since $F$ is an elementary fault in $\faultset$, it forms an edge connecting the two vertices $v_1$ and $v_2$ in the graph defined over $\faultset$, and the fault's syndrome is $\sigma = \{v_1,v_2\}$.
If one or both of these vertices are missing from the graph defined over $\faultset_\text{sub}$, then there can be no $F' \subset \faultset_\text{sub}$ equivalent to $F$.
If both $v_1$ and $v_2$ are present in the graph defined over $\faultset_\text{sub}$, one can use a standard decoder to find a minimum-weight correction with the syndrome $\sigma$. 
If no such correction exists, then there can be no $F' \subset \faultset_\text{sub}$ equivalent to $F$.
But if it does, then the minimum-weight correction is $F'$ as we now prove.
Let $F''$ be a minimum-weight correction with syndrome $\sigma$.
For $F''$ to be equivalent to $F$ we further require that their logical actions are the same.
If the weight of the undetectable fault configuration $F + F''$ is less than $d(\faultset)$, then $F + F''$ must have trivial logical action and therefore $F''$ is equivalent to $F$.
If the weight of $F+F''$ is greater than or equal to $d(\faultset)$, then the weight of $F''$ is at least $d(\faultset)-1$, and the weight of $F'$ must also be at least $d(\faultset)-1$ by virtue of $F''$ having minimal weight for a fault configuration of $\faultset_\text{sub}$ with syndrome $\sigma$.
This contradicts our assumption.

\subsection{Beyond graph-like channel check matrices}
\label{sec:beyond-graphic}

Here we discuss a number of cases where the channel check matrix is not graph-like, but where we can use additional structure to calculate or at least bound the fault distance of the channel.

{\bf Upper bounding the fault distance}.---
When the channel check matrix is not graph-like, an upper bound on the fault distance can be found by finding any non-trivial undetectable fault configuration. 
One way to achieve this is to use a BP-OSD decoding approach as pointed out in Ref.~\cite{panteleev2021}.
Here we propose another method to form upper bounds using a graph $G_\text{res}$ that can be constructed, which we call the \textit{restricted syndrome graph}.
In what follows, we explain how this works.
First we choose a check basis (the restricted syndrome graph can depend on this choice), and then build the graph $G_\text{res}$ by adding edges as normal for those columns which have Hamming weight one or two (it can beneficial to choose a basis that maximizes the number of such columns).
Then finding the fault distance from this graph $G_\text{res}$ using the algorithm in \sec{fault-dist-algo} will provide an upper bound for the fault distance of the channel.
This is because we will have found the smallest non-trivial fault given a restricted elementary fault set, which will still be non-trivial for the full fault set but may not be minimal. 

{\bf Lower bounding the fault distance}.---
We can also lower bound the fault distance by constructing a different graph, which we call the \textit{reduced syndrome graph} $G_\text{red}$.
In what follows, we explain how this works.
We start by building $G_\text{res}$ as before, but now we will also include edges for the remaining columns of the channel check matrix as follows.
Suppose that $s > 2$ is the Hamming weight of one of the remaining columns, and let $w$ be the weight of the fault that column corresponds to. 
We add a new vertex to the graph (similar to the boundary vertex), and add $s$ new edges to the graph, with each edge having weight $w/s$ and connecting to the new vertex.
These edges connect pairs of the vertices in the fault (and connect one vertex to the boundary vertex if $s$ is odd).
We assign all of the residual error from the elementary fault to one of these new edges, and set the other residual errors to be trivial.
This is done for all remaining columns to produce $G_\text{red}$.
Suppose that the true fault distance of the channel is $d$, and let $d_\text{red}$ be the fault distance of $G_\text{red}$ (which can be obtained by running the algorithm in \sec{fault-dist-algo} on $G_\text{red}$).
The key here is to note that the minimum weight non-trivial faults for the true channel can all be considered as the minimum-weight one by the graph algorithm for $G_\text{red}$, and it is proposed then it will have the correct weight.
If another fault is proposed which contains just some of the new edges without the full fault, then it must have lower weight than the true lowest weight non trivial logical fault. 
Hence we obtain a lower bound.

Checking if a check matrix is `almost' graph-like is an interesting open question.
Let us consider a simple special case of this problem.
\textbf{Given} a set of linearly independent channel checks $\Sigma$ of size $n$, 
\textbf{find} a set of $n-1$ linearly-independent checks from the span of $\Sigma$
such than the corresponding channel check matrix is graph-like. 
This problem is closely related to the problem of recognising even-cycle binary matroids 
for which an efficient algorithm has recently been discovered~\cite{EvenCycleMatroids,Guenin2023}.

{\bf If the fault set `separates' into fault-sets with known fault distances }.---
In many stabilizer channels of interest, the channel check matrix is not graph-like, but we can split the decoding problem into independent parts that are graph-like.
This happens for example due to $Y$ errors in surface codes, which result in four rather than two checks being flipped for the standard choice of checks in the surface code. 
We can treat the $Y$ as being composed of an $X$ and a $Z$, and handling $X$ and $Z$ errors separately.
We now formalize this scenario more abstractly, using an approach similar to that described in Refs.~\cite{delfosse2023,Gidney2021stim}.
We require the following from the fault set, channel check and logical effect matrices:
\begin{enumerate}
    \item The fault set $\faultset$ separates into $K+1$ disjoint fault sets $\faultset_k$, so that fault distance with respect to each of the first $K$ fault sets alone is at least $d$. 
    \item For the first $K$ fault sets the spans $\Lambda_k$ of the columns of the check matrices with respect to $\faultset_k$ are linearly independent subspaces.
    \item Each elementary fault from $\faultset_{K+1}$ triggers the same checks and has the same logical effect as a sum of elementary faults from the first $K$ fault sets. 
    The sum includes at most one elementary fault from each $\faultset_k$.
\end{enumerate}
If the above properties hold, the fault distance is at least $d$. 
This is because any undetectable fault of weight $d$ corresponds to an undetectable fault of weight at most $d$ with respect to each of $\faultset_k$.
Moreover, a non-trivial fault of weight $d$ must correspond to a non-trivial fault of weight at most $d$ with respect to one of $\faultset_k$.

{\bf If the channel check matrix is a regular matroid}.
It is possible to efficiently find the fault distance for a more general class of check matrices, namely the matrices that correspond to binary regular matroids \cite{Truemper2014}.
The matrices that are graph-like are a proper subset of those which correspond to binary regular matroids.
There exist efficient algorithms and their C++ implementations that check if a given matrix corresponds to a regular matroid \cite{Truemper2014, Walter2012}.
An efficient algorithm for solving an analog to an odd cycle problem in regular binary matroids is presented in \cite{DePinaThesis}.
We are not aware of any practical examples where channel check matrix is a regular matroid but not graph-like. 
Identifying such examples, as well as decoding algorithms for this class of check matrices is an open question.
Another promising case is for channel check matrices related to perturbed graphic matroids, which have an efficient algorithm for a shortest cycle problem \cite{Geelen2018}. 
It is an open problem to find an efficient algorithm shortest odd cycle in perturbed graphic matroids.

\section{Examples with surface codes}

\label{sec:examples}

In this section we illustrate some of the main concepts discussed in earlier sections by working through two stabilizer channel design examples.
We use the following strategy to design circuit implementations of a given logical operation with fault distance that match the code distance $d$:
\begin{enumerate}
    \item \textbf{Code-deformation sequence}.---
    First find a sequence of codes and which have at least code distance $d$ (for all logical Pauli operators which are not measured out by the channel).
    The sequence should be chosen such that the first and last rounds match the input and output codes, and the logical operators evolve as required by the logical operation as described in \sec{code-def-channel-logical-action}.
    
    \item \textbf{Phenomenological implementation}.---
    Then, one specifies a phenomenological version of the channel in terms of generator measurements for the code sequence.
    Under a simple phenomenological noise model, modifications are made to ensure a fault distance of the channel is $d$ calculated as in \sec{calculating-fault-dist} such as repeating some rounds, or adding intermediate codes to the code deformation sequence.
    
    \item \textbf{Circuit implementation}.---
    Lastly, one specifies an explicit stabilizer circuit implementation of the channel in terms of allowed instructions. 
    This circuit is built from extraction circuits that measure the generators of the phenomenological implementation, and as such there is a natural restriction of the faults of the circuit $\faultset$ to the phenomenological fault subset $\faultset_\text{sub} \subset \faultset$ in \sec{hook-faults}.
    To achieve fault distance $d$, we ensure that the extraction circuits are \textit{valid}, i.e. they act correctly when performed together in the overall circuit and also seek to eliminate brazen hook faults with respect to the phenomenological fault subset. 
    Eliminating these brazen hook faults may not be sufficient to achieve full circuit distance, in which case it can be beneficial to expand the fault subset to $\faultset'_\text{sub} \supset \faultset_\text{sub}$ to include diagonal edges in the decoding graph and eliminate all brazen faults with respect to this expanded subset $\faultset'_\text{sub}$.
    At each stage, it may be necessary to modify the Phenomenological implementation by increasing round repetitions etc.
\end{enumerate}

In \sec{surface-codes} we fix the setting for our surface code examples.
In \sec{ZZ-meas-channel} we design a stabilizer channel to non-destructively measure the logical $X_1 X_2$ operator of a pair of surface codes.
This is a standard example, but acts as a warm-up to clarify our broad approach to design fault-tolerant stabilizer channels.
Then in \sec{hadamard-channel-def}, \sec{hadamard-channel-phen} and \sec{hadamard-channel-circ}  we design a stabilizer channel to fault-tolerantly implement the logical Hadamard on a surface code patch at circuit level.
In this second example we see that a number of intricacies arise which must be handled to achieve a fault distance $d$, some of which we believe have not been noted in prior literature. 
This adds to the collection of intricacies arising in fault-tolerant circuit design \cite{gidney2023inplace,ChamberlandCampbell2022,gidney2023bacon,McEwen2023,gidney2023cleaner,geher2023tangling}.

\subsection{Surface code stabilizer channels}
\label{sec:surface-codes}

\begin{figure}[h]
    (a)\includegraphics[width=0.45\linewidth]{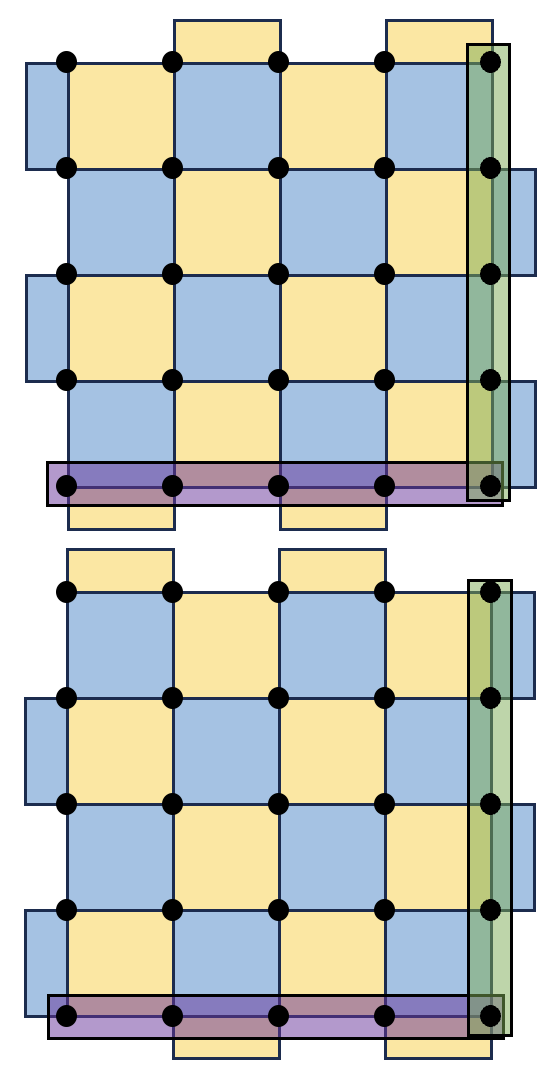}
    (b)\includegraphics[width=0.44\linewidth]{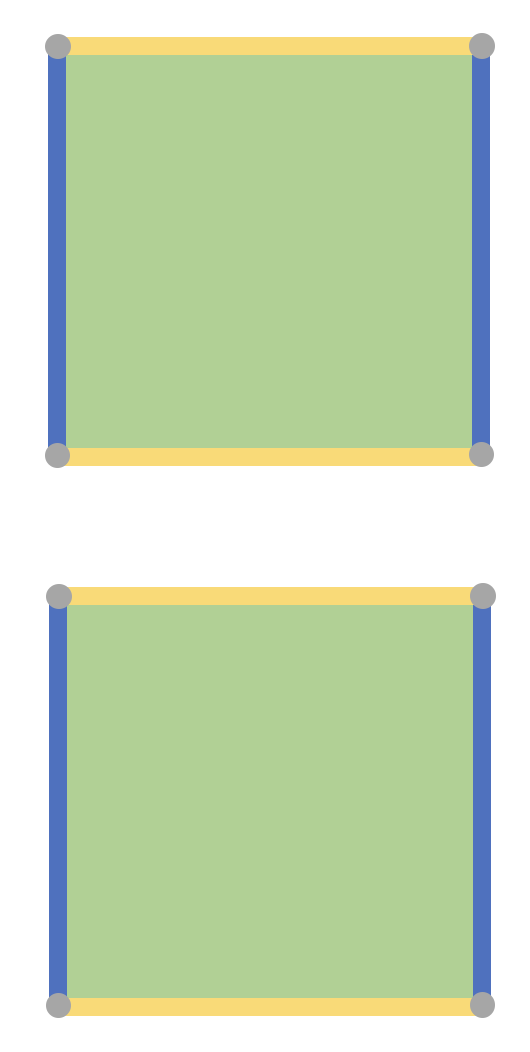}
    \caption{
        (a) a pair of distance $d = 5$ surface code patches with slightly different boundary implementations. 
        Qubits are placed on vertices, and $X$-type ($Z$-type) stabilizer generators are represented by plaquettes colored in blue (yellow).
        We consider the logical action with respect to the standard basis, with a horizontal purple $X$ string is a logical $X$ representative, and a vertical green $Z$ string is a logical $Z$ representative. 
        (b) The topological features of the two patches are the same.
        We represent the two types of boundary with yellow (blue) lines, and grey dots at the intersection of boundary types.
        We provide a brief review of the topological viewpoint of surface codes in \app{topo-viewpoint}.
    }
    \label{fig:surfacecode}
\end{figure}

By the term ``surface code'', we refer to a broad class of quantum error correcting codes defined on qubits embedded in two-dimensional surfaces with a pattern of local stabilizer generators tiling the surface.
There are many variants due to the range of surfaces that can be chosen, a choice of how the qubits are embedded in the surface, and some local freedom in the tiling both in the bulk, but also along the boundaries.
Among the most familiar surface codes are the family of square-patch surface codes depicted in \fig{surfacecode}, which we will call the \emph{standard patch}.
In this section, we are interested in stabilizer channels that implement logical operations on surface codes.
We will assume the following throughout:
\begin{itemize}

    \item \textbf{Qubit layout}.---
        Physical qubits can be placed on the sites of two lattices on $\mathbb{R}^2$:
        the \textit{data lattice} consisting of integer points $(i,j)$,
        and the \textit{ancilla lattice} consisting of half-integer points $(i+\frac{1}{2}, j+\frac{1}{2})$.
        In most figures we just show the data lattice.
        
    \item \textbf{Input and output codes}.---
        The input and output codes for the channels we consider are a pair of surface code patches as shown in \fig{surfacecode}(a). 
        All of our figures will be shown for the $d=5$ case but generalize to arbitrary odd integer distance $d$.
        
    \item \textbf{Instantaneous stabilizer codes}.---
        When viewed as a code deformation channel, instantaneous stabilizer codes are surface code patches with each check being weight-four or weight-two restrictions of an alternating checkerboard pattern.
        
    \item \textbf{Phenomenological instruction set}.---
        We allow joint-Pauli preparations and measurements supported on sets of up to four data qubits contained within a unit square.
        We also allow for the simultaneous shift of all qubits in a region by a uniform displacement vector, even if the shifted region overlaps with the original region.
        
    \item \textbf{Circuit instruction set}.---
        We allow the circuit operations and faults specified in \sec{stochastic-noise}, with two-qubit Clifford unitaries between nearest-neighbor data-ancilla pairs.
        We further assume that any generators measured by circuits are done so using a single ancilla qubit and a CNOT circuit as shown in \fig{def-circuits} (and discuss hook faults in the same figure).
        
    \item \textbf{Fix-ups}.---
        In this section we follow the convention that stabilizer channels are circuits with all conditional Paulis at the end as in~\fig{enc-and-dec}.
    
    \item \textbf{Channel checks}.---
        We construct the check set in the phenomenological implementation (and use the same checks for the circuit implementation) in such a way that the fault sets separate into graph-like decoding matrices as described in \sec{beyond-graphic}. 
        For each Pauli $P$ measured in a phenomenological round we consider both the previous and the next round.
        If $P$ is a product of Pauli operators measured in the previous round we add the corresponding check.
        If $P$ is a product of Pauli operators measured in the next round we add the corresponding check.
        We do not add the same check to the set of checks twice.
        When constructing the checks for the first and last rounds respectively, we follow the same principle but use the input syndrome in place of a round before the first, and the pre-phase-fixed output syndrome in place of a round after the last~\footnote{Recall that we follow the convention where the conditional Paulis are placed between the last round and the output syndrome of a stabilizer channel (as in~\fig{enc-and-dec}).
        To obtain the pre-phase-fixed output syndrome, we commute the output syndrome past the conditional Paulis to find the values they would have immediately after the final round.}.

    \item \textbf{Noise}.---
        We assume the standard circuit noise model in \sec{stochastic-noise}, with unit weights for faults. 
        The fault set for the phenomenological noise is bit-flips of measurement outcomes and i.i.d. Pauli noise on qubits between rounds. 
        Note that these phenomenological faults can always be considered as a strict subset of the faults of the circuit noise model in accordance with the hook fault definitions in \sec{hook-faults}.
    
\end{itemize}

\begin{figure}
    (a)\includegraphics[width=0.29\linewidth]{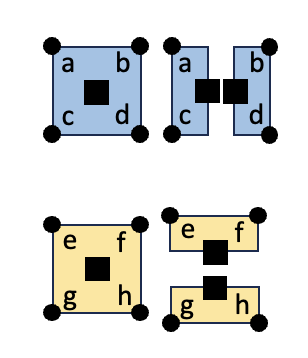}
    (b)\includegraphics[width=0.24\linewidth]{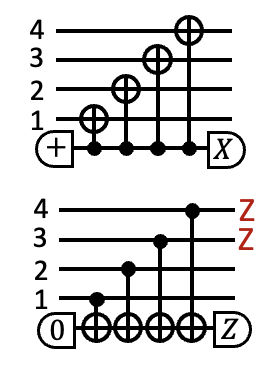}
    (c)\includegraphics[width=0.32\linewidth]{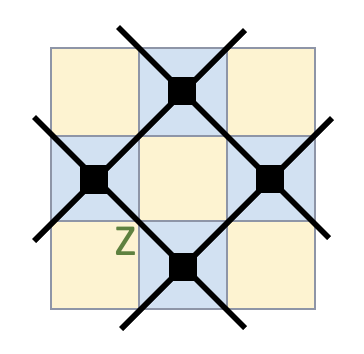}
    (d)\includegraphics[width=0.28\linewidth]{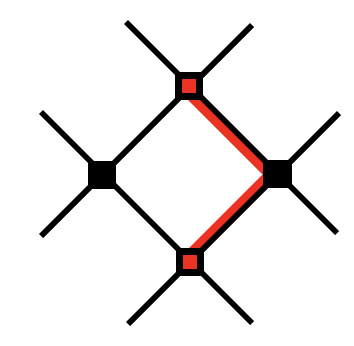}
    (e)\includegraphics[width=0.28\linewidth]{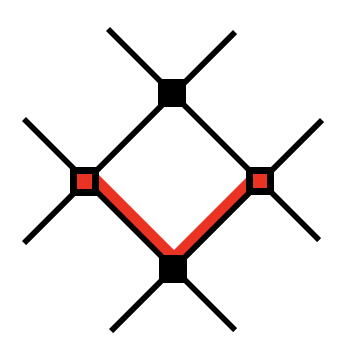}
    (f)\includegraphics[width=0.28\linewidth]{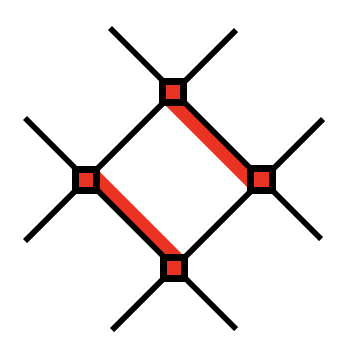}
    (g)\includegraphics[width=0.45\linewidth]{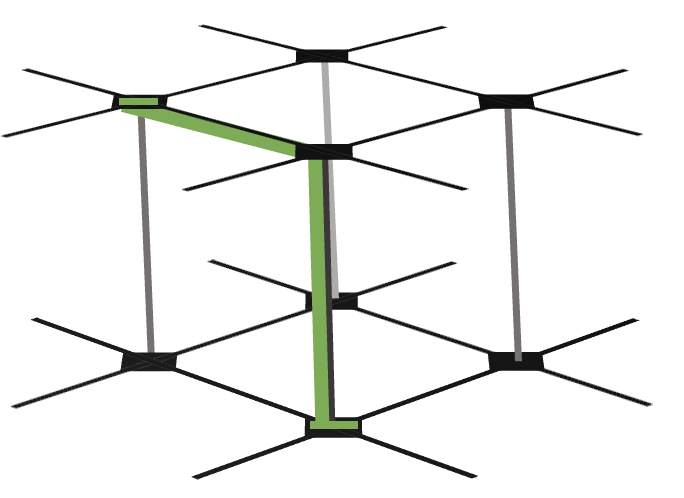}
    (h)\includegraphics[width=0.45\linewidth]{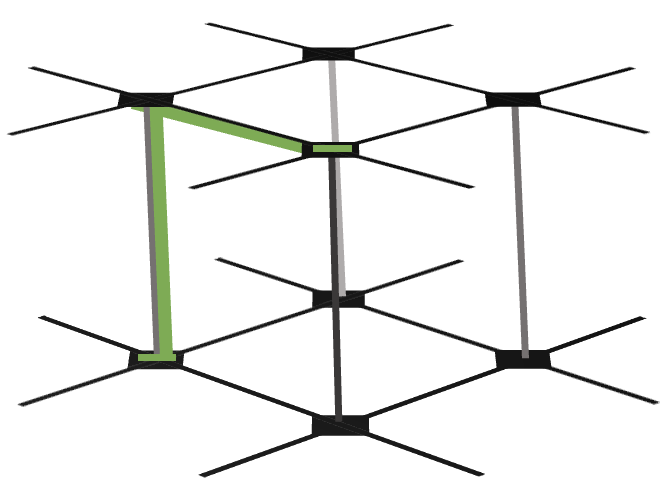}
    \caption{
    (a,b) The circuits used to measure the $X$-type and $Z$-type generators are specified by an assignment of the letters $a,b,c,d, e, f, g, h$ in (b) to a time ordering as in (a).
    Some single faults in these circuits can propagate into weight-two Pauli errors on the data qubits, such as the $ZZ$ error indicated in red.
    We call these space-like hook faults with respect to the phenomenological fault set.
    (c) A section of the decoding graph for $Z$ errors.
    (d,e,f) Depending on the relative time ordering of $e,f,g,h$, the $ZZ$ error in (b) produces different patterns of space-like hook faults with respect to the phenomenological fault set in the $Z$ error decoding graph.
    Specifically, (d) arises if $\{f,h\} =\{3,4\}$, (e) arises if $\{g,h\} =\{3,4\}$ and (f) arises if $\{g,f\} =\{3,4\}$.
    (g,h) Depending on the relative time ordering of $d$ and $a$, a single-qubit $Z$ fault as shown in green in (c) between time steps $d$ and $a$ results in different space-time hooks.
    Specifically, (g) arises if $d<a$, while (h) arises if $d>a$.
    }
    \label{fig:def-circuits}
\end{figure}

\subsection{XX measurement}
\label{sec:ZZ-meas-channel}

In this warm-up example we apply some of the techniques we have introduced to carry out the three-step design flow to construct a fault-tolerant implementation of the logical $X_1 X_2$ measurement of a pair of surface code patches. 
In this example, we reproduce the well-known lattice-surgery implementation of this operation.

\begin{figure}
    \includegraphics[width=1.0\linewidth]{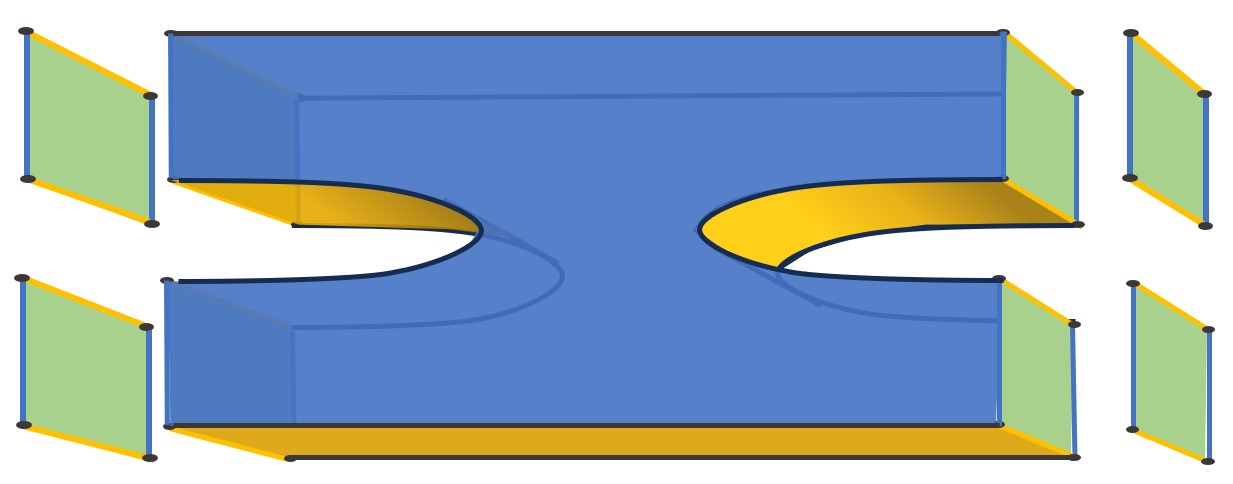}
    \caption{
    A topological representation of a logical $X_1 X_2$ measurement of a pair of adjacent surface code patches, with time moving from left to right using the lattice surgery operations patch-merge followed by patch-split.
    }
    \label{fig:ZZ-measurement-topo}
\end{figure}

\begin{figure}
    \includegraphics[width=1.0\linewidth]{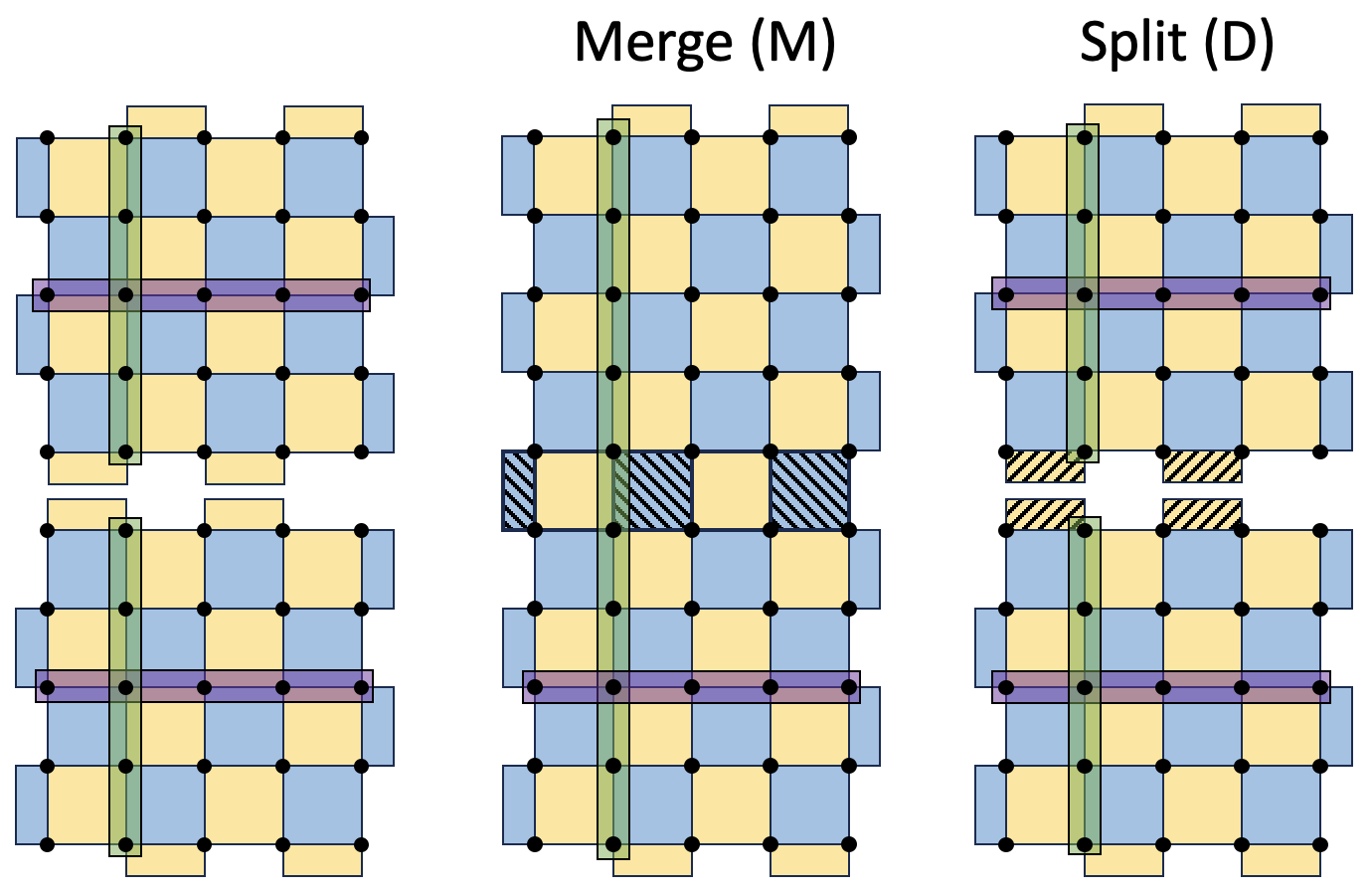}
    \caption{
    A sequence of three instantaneous stabilizer groups in a code deformation sequence that implements the same logical channel as in \fig{ZZ-measurement-topo}. 
    The input code is depicted on the left, with a stabilizer group $\stab_\text{in}$ corresponding to the stabilizer groups of the two input surface code patches. 
    We include $X$ and $Z$ logical operator representatives in purple and green (using the freedom in which representatives we choose to avoid overlap for visual clarity).
    The merge round results in the middle stabilizer group, $\stab_\text{M}$ (with M for merge).
    The split round results in the right stabilizer group, $\stab_\text{D}$, which contains the output stabilizer group $\stab_\text{out} =\stab_\text{in}$ (along with the logical $X_1 X_2$ operator), which again corresponds to the stabilizer groups of the two input surface code patches.
    For each deformation round, we shade those stabilizer generators which were not present in the previous instantaneous stabilizer group.
    The logical $X_1 X_2$ outcome $f(O)$ of the channel is given by the parity of the outcomes of the shaded blue generators in the merge round.
    }
    \label{fig:ZZ-measurement-code-def}
\end{figure}

\textbf{Code deformation sequence}.---
We begin with a topological viewpoint of how this logical gate can be applied as shown in \fig{ZZ-measurement-topo} (an understanding of the topological viewpoint of surface codes will not be necessary to understand the results in this section, but we include a brief overview in \app{topo-viewpoint} for completeness).
To form an explicit code-deformation sequence from this topological picture, we consider three time-slices which correspond to the sequence of surface codes shown in \fig{ZZ-measurement-code-def}.
In this code deformation sequence, an input pair of surface code patches (with stabilizer group $\stab_\text{in}$) are first merged together into a single surface code patch (with stabilizer group $\stab_M$), and then split once more.
Note that each code in this sequence has code distance $d$.

The logical action of this code deformation sequence is straight forward to calculate using the approach in \sec{code-def-channel-logical-action}.
The input and output code basis are the standard choices of logical $X$ and $Z$ for both patches as illustrated in \fig{ZZ-def-logical-action}.

\begin{figure}
    \includegraphics[width=0.9\linewidth]{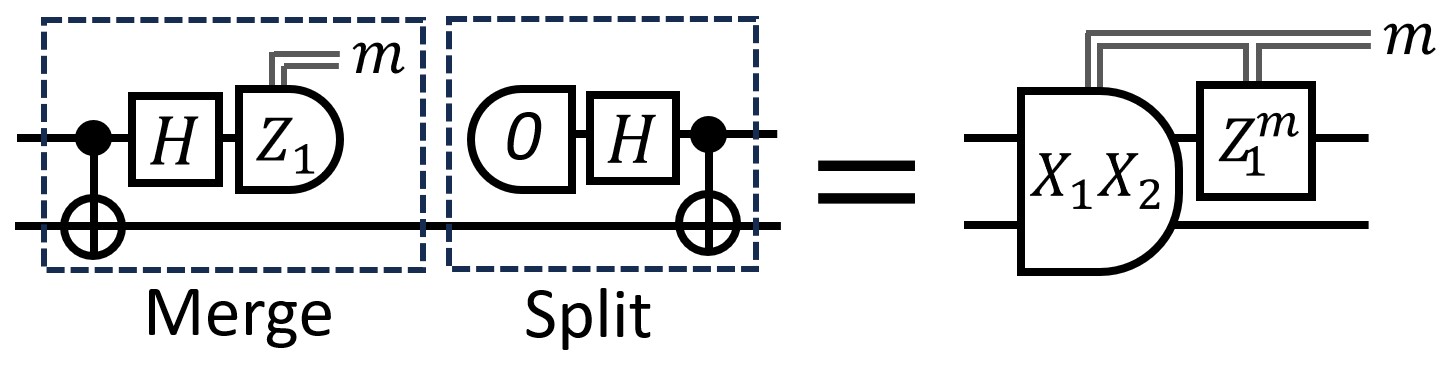}
    \caption{
    The logical action of the code deformation sequence depicted in \fig{ZZ-measurement-code-def} is equivalent to a logical $X_1 X_2$ measurement up to a conditional Pauli (which ensures that the output state is stabilized by logical $X_1 X_2$).
    Let the outcome vector of the circuit be $O$.
    The logical outcome $f(O)$ of $X_1 X_2$ is obtained from the parity of the set of shaded $X$-type checks in the measurement of $\stab_M$.
    }
    \label{fig:ZZ-def-logical-action}
\end{figure}

\textbf{Phenomenological implementation}.---
The phenomenological channel consists of the following rounds:
\begin{enumerate}
    \item $d$ rounds: measure the generators of $\stab_M$.
    \item $1$ round: measure the generators of $\stab_\text{out}$.
\end{enumerate}

It is well-known that the channel check matrix in this example has the property described in \sec{beyond-graphic} whereby it separates into two graph-like matrices.
Specifically, this occurs when we separate faults into (i) $X$ Paulis and $Z$-check flips producing the $X$ decoding graph $G_X$, and (ii) $Z$ Paulis and $X$-check flips, forming the $Z$ decoding graph $G_Z$ as shown in \fig{ZZ-def-decoding-graphs}.
Each of the remaining faults are $Y$ Paulis described as the combination of precisely one $X$ Pauli fault in $G_X$ and one $Z$ Pauli fault in $G_Z$.

\begin{figure*}
    (a)\includegraphics[width=0.8\linewidth]{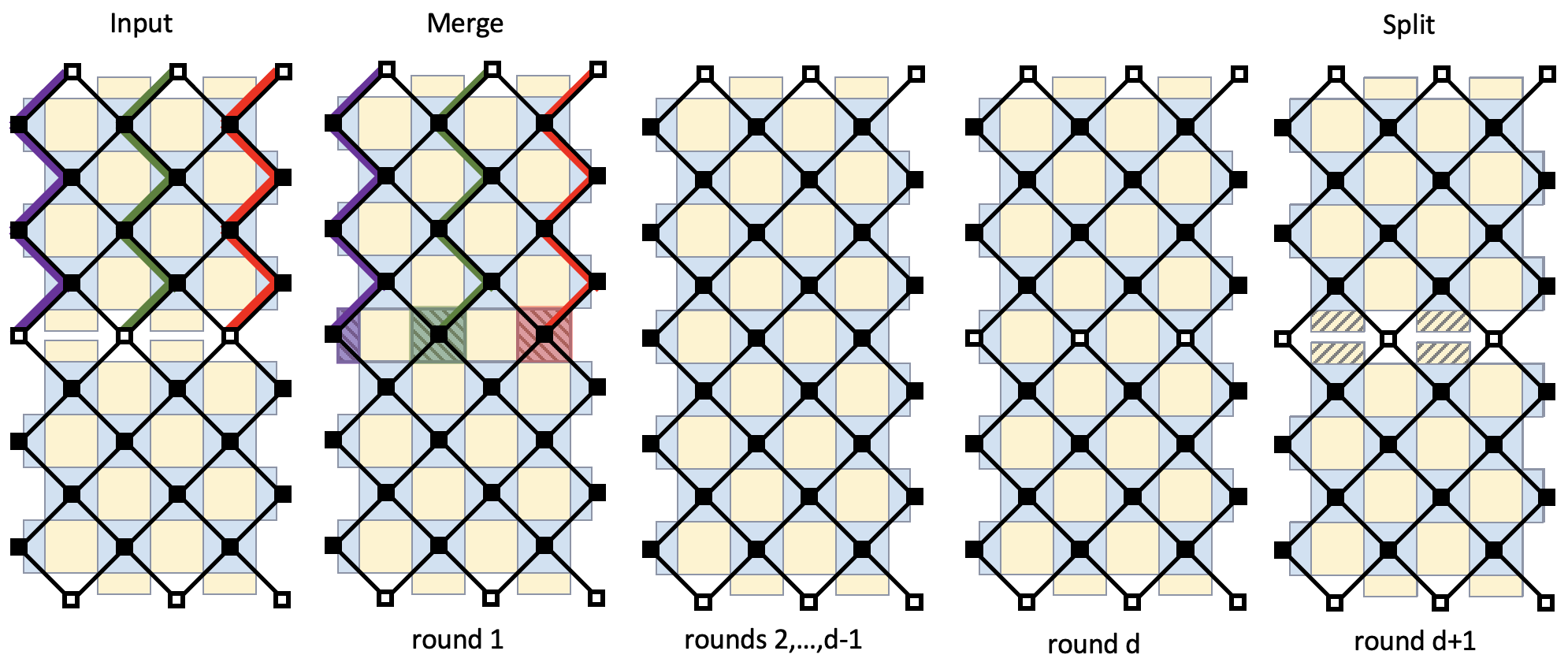}
    (b)\includegraphics[width=0.82\linewidth]{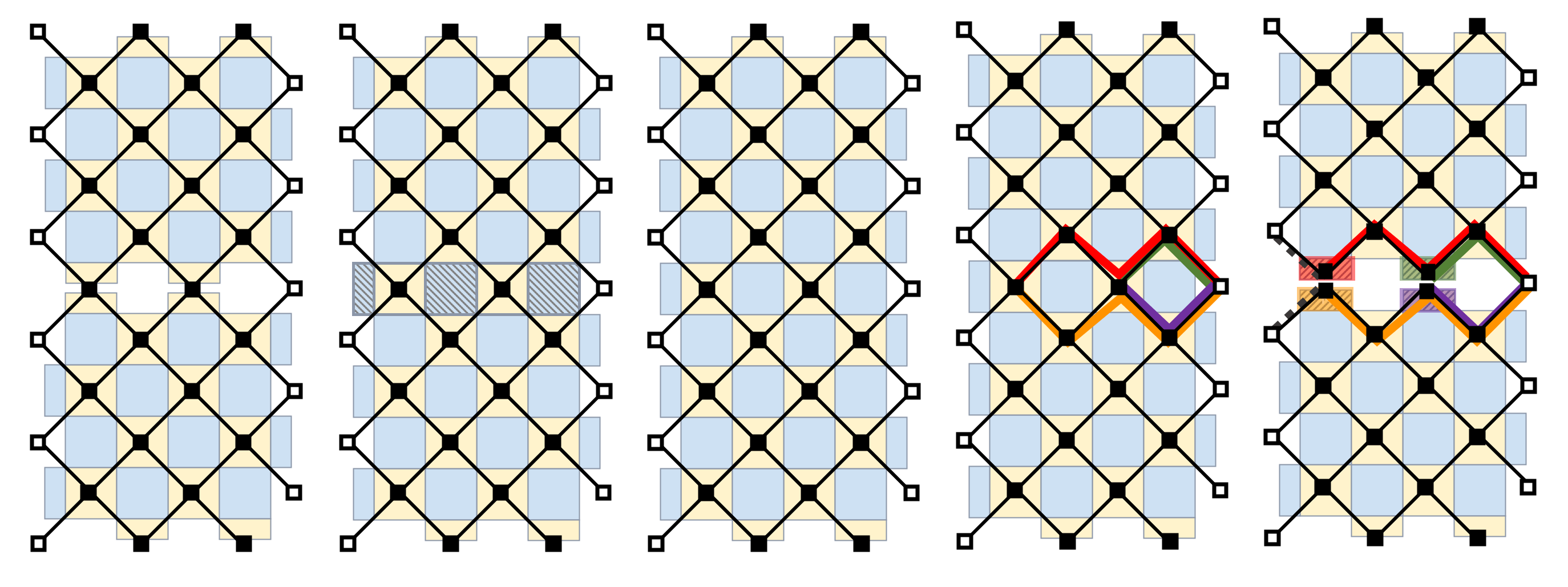}
    \caption{
    (a) and (b) show time-slices of the $Z$ decoding graph $G_Z$ and the $X$ decoding graph $G_X$ respectively for each step in the phenomenological implementation of the XX measurement channel.
    Each solid vertex corresponds to a check, and each empty vertex is identified with the boundary vertex in the graph. 
    The solid vertices of the time slice of the decoding graph shown on top of round $j$ correspond to checks that are parities of measurement outcomes in rounds $j$ and $j+1$.
    Each black edge corresponds to a fault caused by a single X or Z fault on the qubit at the center of that edge after round $j$ and before round $j+1$.
    The checks associated with the vertices at the ends of any edge are flipped when the fault corresponding to that edge occurs.
    Each set of colored edges of the same color (which can span multiple layers) shows a Pauli fault configuration equivalent to a single measurement outcome flip highlighted with the same color (as in \sec{flip-fault-effect}).
    We use the colored edge set to represent the edge associated with that measurement outcome flip in the decoding graph. 
    To form full decoding graphs from these time slices, we also include a vertical edge between pairs of filled vertices that have the same location in consecutive time slices.
    }
    \label{fig:ZZ-def-decoding-graphs}
\end{figure*}

To see why the first code deformation step is repeated $d$ times to form the phenomenological implementation of the channel, it is important to understand the lowest-weight non-trivial logical fault configurations of this phenomenological channel.
Clearly, the space-like logical operators of the instantaneous stabilizer groups (indicated on \fig{ZZ-measurement-code-def}) can result in logical faults.
Simply by choosing the code patches to be distance $d$ as shown ensures those logical faults are all of size at least $d$.
There are time-like logical faults too: a repeated outcome flip of any one of the shaded checks in $\stab_M$ of \fig{ZZ-measurement-code-def} corresponds to a string of vertical edges in \fig{ZZ-measurement-code-def} connecting the boundary vertex to itself and will not be seen by any channel check and will result in the outcome $\mathfrak{f}(O)$ from being flipped.
These time-like logical faults require us to use $d$ repeated rounds -- if fewer repetitions were made then the undetectable logical outcome flip could occur with fewer than $d$ faults.

Looking at the decoding graphs $G_X$ and $G_Z$ in \fig{ZZ-def-decoding-graphs}, one may wonder why the high-weight colored faults that arise due to flips of the random outcomes do not result in a reduction of the distance.
For $G_Z$, the merge step introduces vertical weight-($d$) strings (red, green, purple) - these lie along minimum-weight logical $Z$ operators of the stabilizer code after that time step.
However, these minimum-weight logical $Z$ operators have weight $2d$, and as such these faults do not reduce the distance below $d$.
The reason that the flips of the shaded measurement outcomes in round $(d+1)$ of \fig{ZZ-def-decoding-graphs} do not reduce the fault distance of $G_X$ is more subtle.
Take for example the weight-4 undetectable fault configuration formed from the red and orange faults and the two dashed edges in \fig{ZZ-measurement-code-def}.
A careful analysis of this fault configuration shows it corresponds to the application of a stabilizer immediately before round $d$, and then the application of a logical $X_1X_2$ Pauli operator immediately before round $d+1$.
However, this logical operator has trivial logical action as defined in \sec{3b}, since the logical $X_1X_2$ operator is stabilized by the output state (given that the channel measures this operator).

\textbf{Circuit implementation}.---
The circuit-level implementation of the channel requires additional ancilla qubits as shown in \fig{ZZ-measurement-hooks}.
The channel consists of the same rounds as the phenomenological channel, but with the weight-two and weight-four Pauli measurements replaced by extraction circuits as described in \fig{def-circuits}.
As noted earlier in this section, we (i) require that these circuits are valid, and (ii) avoid brazen hook faults with respect to the phenomenological fault set. 

To ensure the circuits are valid amounts to careful checks on the ordering of any interleaved CNOTs such that the Pauli measurements in the code deformation sequence are correctly implemented in the absence of faults. 
To avoid brazen hook faults one could in principle check each circuit fault individually using \algo{fault-set-in-min}.
However, this assumes we already have a circuit in mind to test.
It can instead be useful to motivate good candidate circuits by ensuring they avoid some known patterns that result in brazen hooks.
One such pattern is to ensure that time-like hook faults of the type shown in \fig{def-circuits} do not have large overlap with the minimum-weight logical operators of the instantaneous stabilizer groups of the channel as discussed in \fig{ZZ-measurement-hooks}.
As is already known~\cite{ChamberlandCampbell2022}, we find that the standard extraction circuits for surface codes avoid these known patterns.

\begin{figure}
    \includegraphics[width=1.0\linewidth]{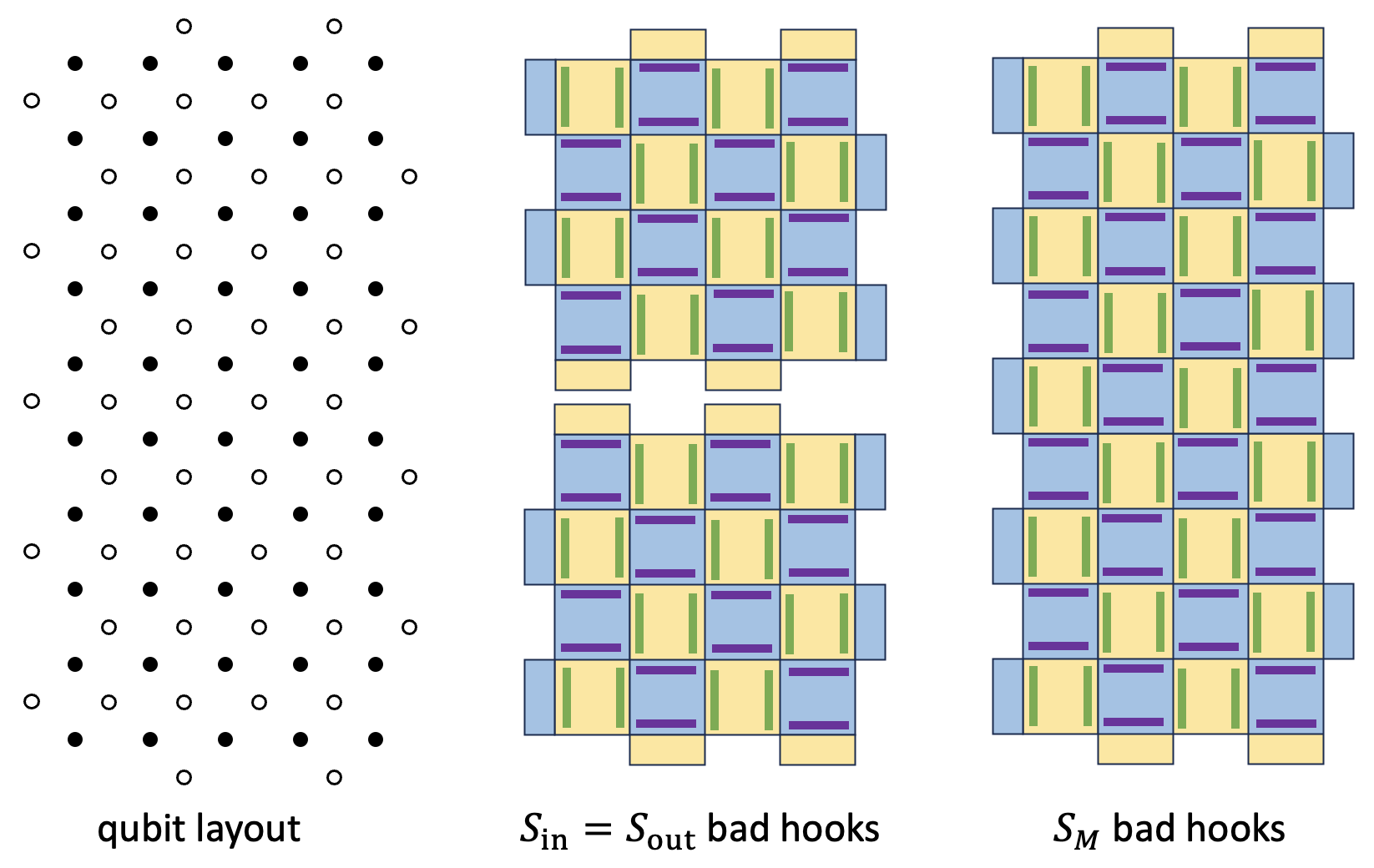}
    \caption{
    We include qubits on the ancilla lattice (white) for the circuit implementation of the logical XX measurement.
    Given the structure of the  two stabilizer codes that appear in the code deformation sequence, there are a number of brazen hook faults that must be avoided.
    Specifically, all horizontal X-type and vertical Z-type hook faults are brazen and must be avoided throughout the channel.
    This is achieved by the standard surface code CNOT ordering~\cite{Fowler2012}, which corresponds to the labelling $(a,b,c,d) = (1,3,2,4)$ and $(e,f,g,h) = (1,2,3,4)$ in \fig{def-circuits}.
    }
    \label{fig:ZZ-measurement-hooks}
\end{figure}

The checks are the same as for the phenomenological channel. 
Given these choices, the resulting channel check matrix once more separates into graph-like parts $G_Z$ for the subset of faults that only flip $X$-checks, and $G_X$ for the subset of faults that only flip $Z$-checks, with all other faults decomposing into these types.
The fault distance of this circuit implementation of the channel is $d$.

\subsection{Hadamard: deformation sequence}
\label{sec:hadamard-channel-def}

We begin our design of a fault-tolerant implementation of a logical Hadmard gate by proposing a code deformation sequence based on a topological viewpoint of this gate, which is very natural as shown in \fig{hadamard-topo}.
The code is rotated, such that the boundaries where the logical $X$ and $Z$ string operators terminate are swapped, and then a transverse Hadamard is applied to the data qubits which exchanges all $X$ and $Z$ stabilizers and logical operators. 
This leaves the code oriented as it was initially, but with the logical $X$ operator having been mapped to the logical $Z$ operator and vice-versa, as is required for the logical Hadamard.

\begin{figure}
\includegraphics[width=1.0\linewidth]{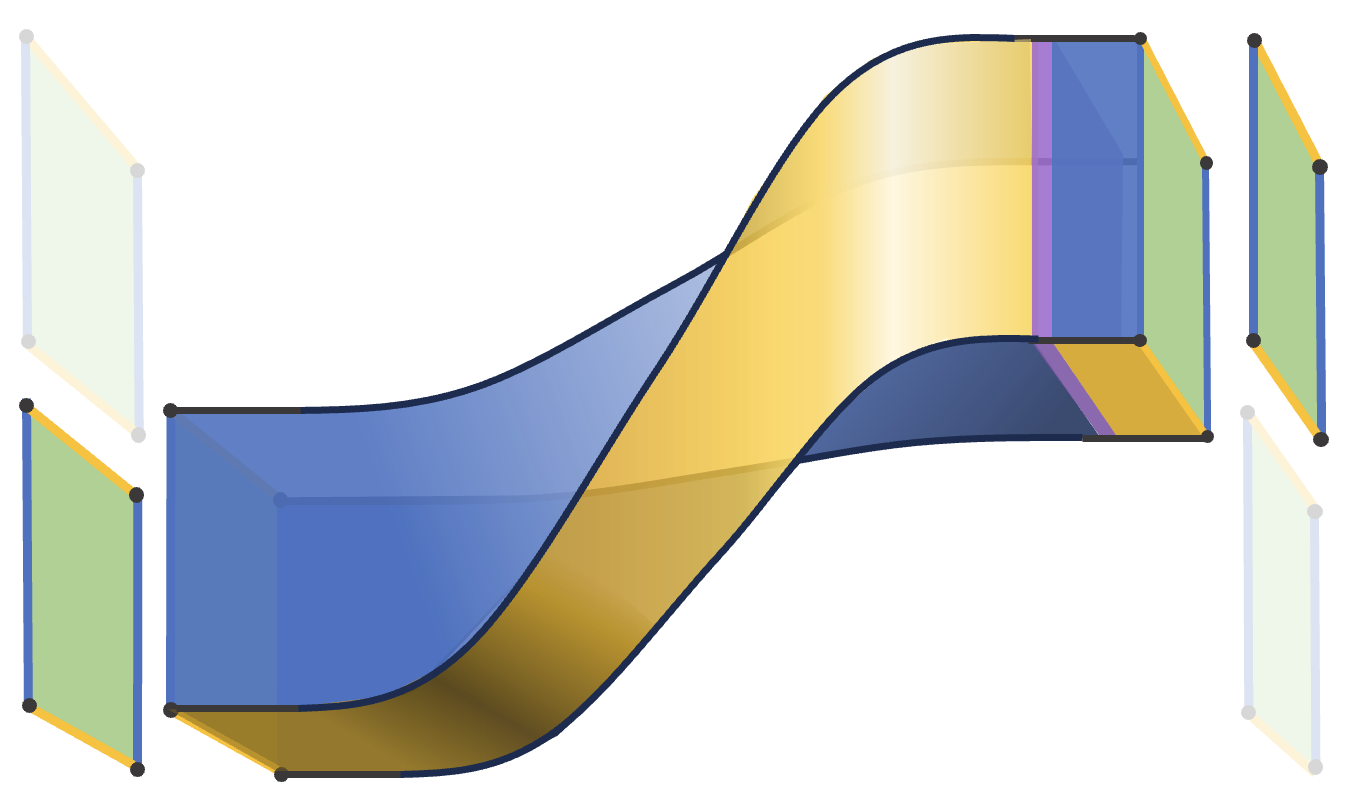}
    \caption{
    A topological representation of a the logical Hadamard on a surface code patch, with time moving from left to right.
    We focus on the ``Hadamard + move'', where the input and output surface codes are on adjacent patches.
    }
    \label{fig:hadamard-topo}
\end{figure}

In \fig{hadamard-code-def} we give a code deformation sequence which implements the logical Hadamard, or more specifically, the logical Hadamard along with a patch move. 
This can be considered an explicit implementation of the topological viewpoint of the channel in \fig{hadamard-topo}.
We note that if one requires the patch of the output to be in the same location as the input patch, one can use a logical move operation. 
We focus on the "Hadamard plus move" here because it has a lower time cost by avoiding the patch move and therefore may be beneficial in practice.

\begin{figure*}[h]
    \includegraphics[width=1.0\linewidth]{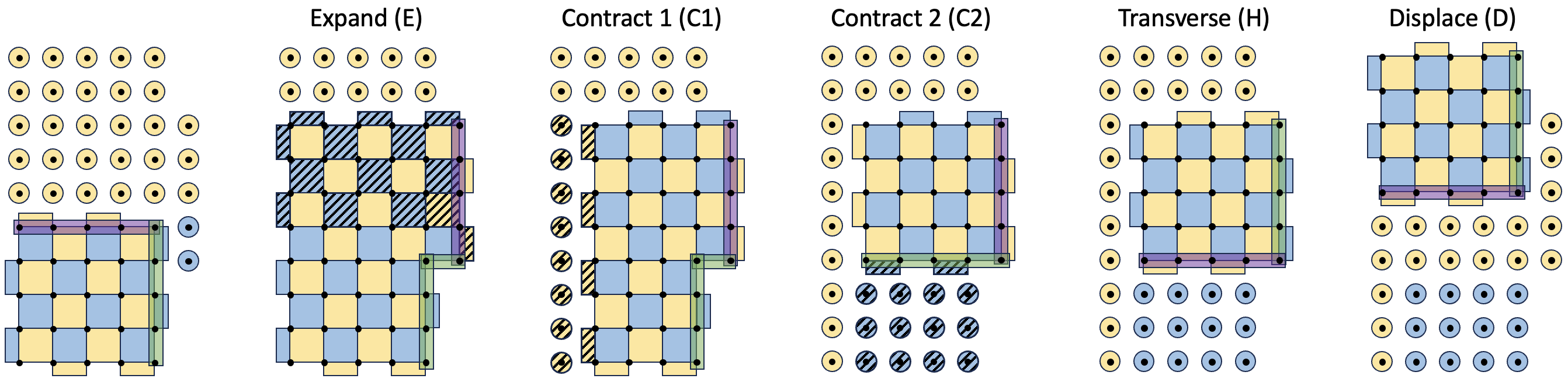}
    \caption{
    \textbf{Code deformation sequence:}
    A sequence of six instantaneous stabilizer groups in a code deformation sequence that implements the logical Hadamard. 
    For each, we shade those stabilizer generators which were not present in the previous instantaneous stabilizer group.
    The first and last stabilizer groups contain the stabilizer groups for the input and output code patches.
    Note that in this implementation, the logical qubit is placed next to its starting point.
    The penultimate stabilizer group is obtained by a transverse Hadamard on the data qubits in the patch (which exchanges stabilizer generator types).
    In the last code deformation step, a unitary move operation is applied.
    }
    \label{fig:hadamard-code-def}
\end{figure*}

\begin{figure*}[h]
    (a)\includegraphics[width=0.31\linewidth]{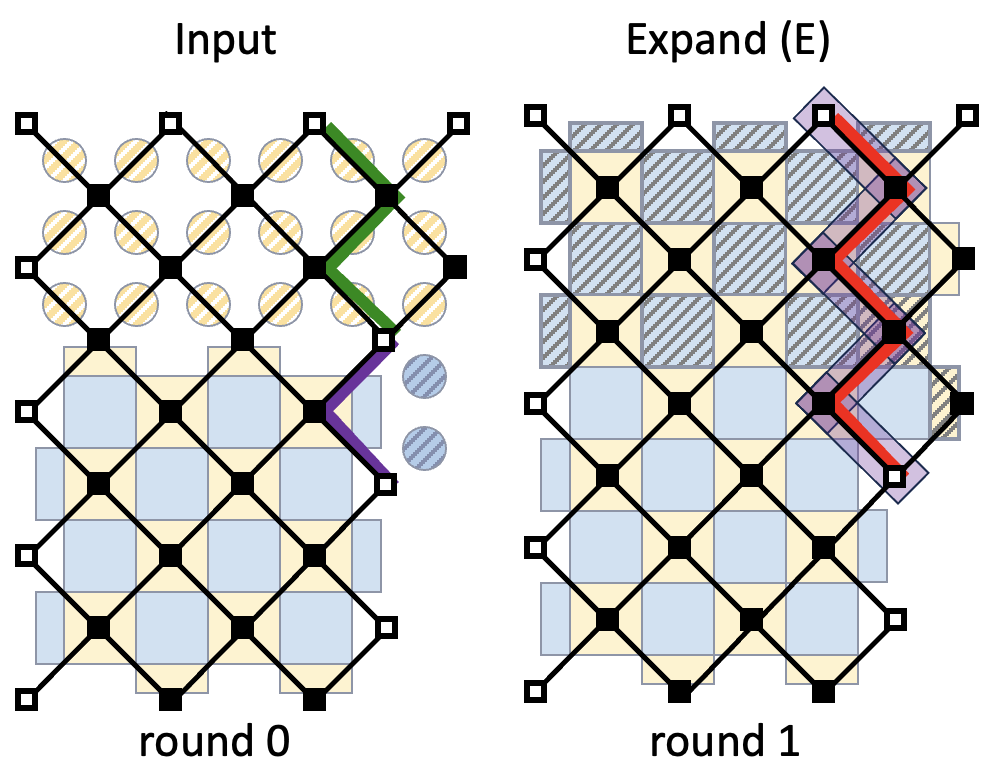}
    (b)\includegraphics[width=0.64\linewidth]{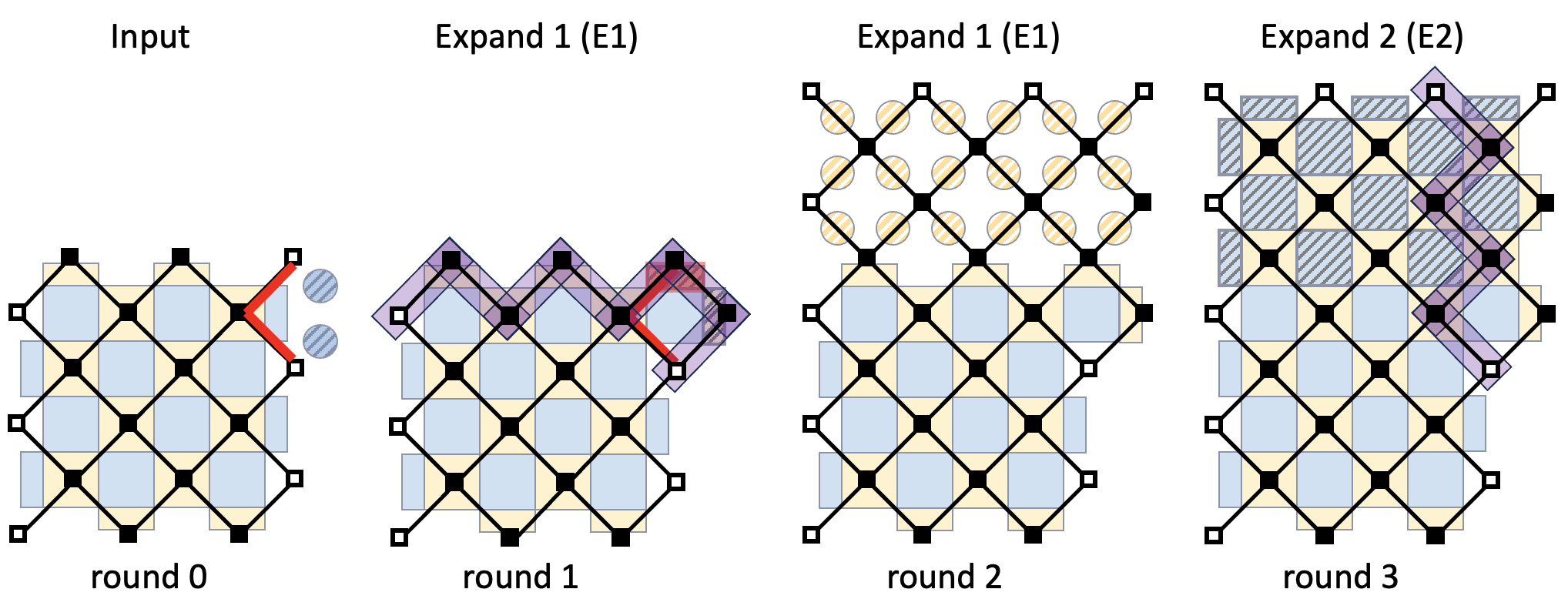}
    \caption{
    \textbf{Phenomenological implementation:} 
    To achieve a fault distance $d$ in the phenomenological implementation, we separate the expansion stage E into two substages E1 and E2. 
    (a) The reason can be seen from the round 0 and round 1 time-slices of the $X$ decoding graph $G_X$.
    Here, we use the same conventions to produce decoding graph time slices as in \fig{ZZ-def-decoding-graphs}, taking the measurement of the input code to be round 0 and both rounds 1 and 2 to be measurements of the `Expand (E)' generators.
    The fault highlighted in green in the first time slice is an undetectable non-trivial fault. 
    This is because the green fault combined with purple fault configuration (which is a stabilizer) is equivalent to the red fault configuration in the panel `Expand (E)', which is a logical operator with non-trivial logical action.
    (b) We see that the time-slices of the $X$ decoding graph for the two substages E1 and E2 avoid this issue. 
    Potentially problematic random outcome flip in E1 is equivalent to 
    weight two Pauli fault that does not overlap with a minimum weight logical operator.
    }
    \label{fig:Hadamard-phenom-extra-step}
\end{figure*}

\subsection{Hadamard: phenomenological implementation}
\label{sec:hadamard-channel-phen}

When we build the phenomenological implementation, we might at first hope that, like in our logical $XX$ measurement example, we simply need to ensure that some of the codes formed in the code deformation sequence we found in \sec{hadamard-channel-def} need to be repeated in order to achieve a fault distance $d$.
However, this does not work --- as is discussed in \fig{Hadamard-phenom-extra-step}, the steps of the code deformation sequence introduce a weight $d-2$ undetectable non-trivial fault configuration.
This can be avoided by including an additional code deformation step, by splitting the expand step into two sub-steps as shown in the same figure.

The phenomenological channel consists of the following rounds:
\begin{enumerate}
    \item $2$ rounds: measure the generators of $\stab_\text{E1}$.
    \item $d$ rounds: measure the generators of $\stab_\text{E2}$.
    \item $d$ rounds: measure the generators of $\stab_\text{C1}$.
    \item $1$ round: measure the generators of $\stab_\text{C2}$.
    \item $1$ round: apply transverse Hadamard to data qubits.
    \item $O(1)$ rounds: shift data qubits distance $O(1)$ to final position.
    This can be achieved by a sequence of shifts using qubits in the ancilla lattice (this is the same approach as we use in the circuit implementation of the channel, and we provide more explicit details for that in \fig{Hadamard-circuits}.
\end{enumerate}

\subsection{Hadamard: circuit implementation}
\label{sec:hadamard-channel-circ}

As we saw in the logical $XX$ measurement example, the choice of circuits to implement the generator measurements is very important.
In \fig{Hadamard-expanded-allowed-hooks} we consider a number of the surface code patches in isolation and analyze which space-like hook faults would be brazen for each check in the various patches, which constrains the circuits that can be used to perform generator extraction.

\begin{figure}[t]
    (a)\includegraphics[width=0.93\linewidth]{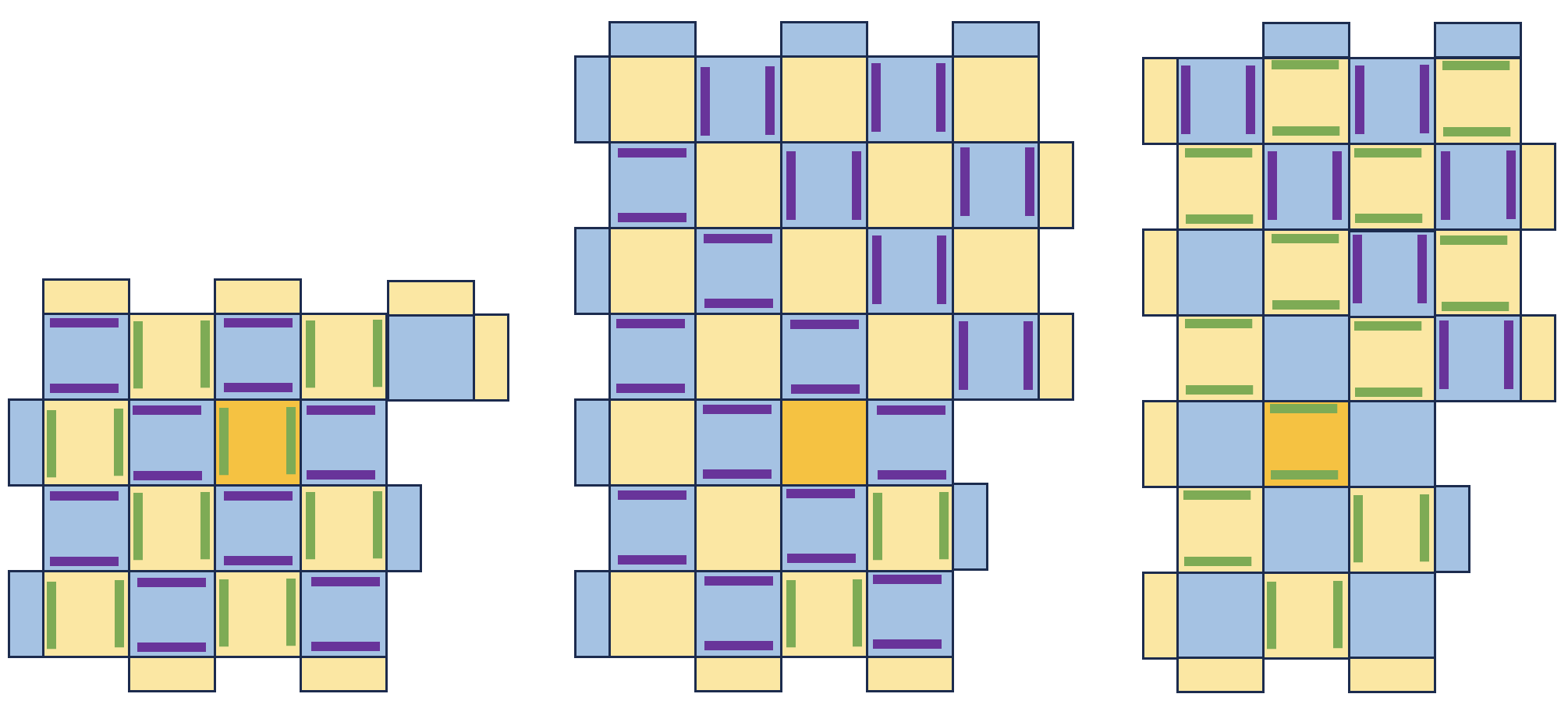}
    (b)\includegraphics[width=0.97\linewidth]{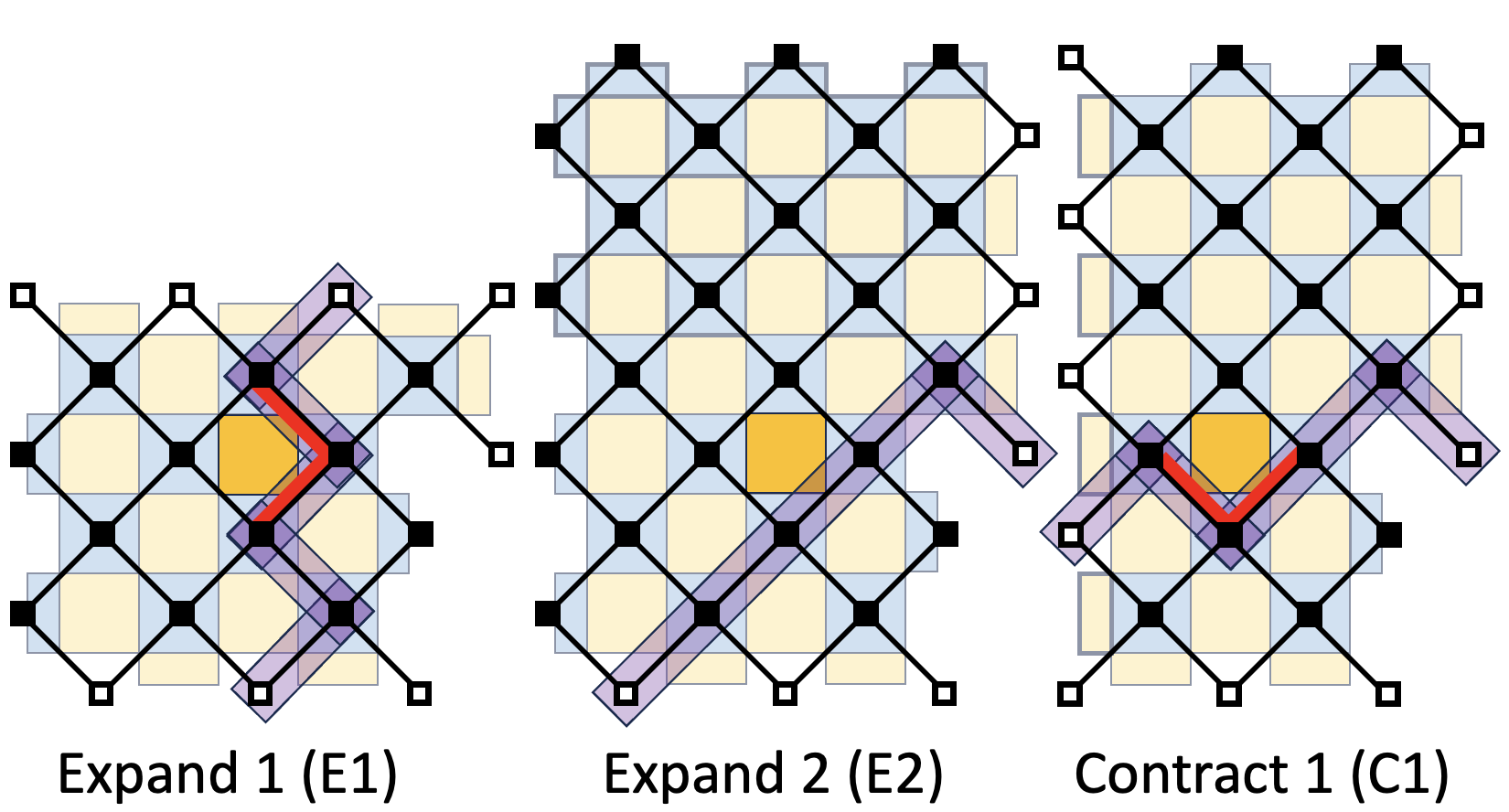}
    \caption{
    (a) Potential brazen hook faults that could arise from generator extraction circuits for each of the three non-standard surface code patches that appear in the phenomenological implementation. 
    In each case, one must take care to ensure that these brazen hook faults do not appear in the generator extraction circuits for the given code.
    (b) An example of how these potential brazen hook faults are identified for the highlighted yellow $Z$-type stabilizer generator.
    }
    \label{fig:Hadamard-expanded-allowed-hooks}
\end{figure}

From \fig{Hadamard-expanded-allowed-hooks}, it is clear that although some generators appear in multiple code patches (such as that highlighted in yellow), the circuit requirements can depend on the patch. 
However, choosing a circuit that avoids brazen hook faults for each stabilizer code independently is not sufficient to achieve the fault distance of the channel. 
One also needs to consider the effects of hazardous hook faults, which align with minimum-weight logical fault configurations of $\faultset$ but not with align with minimum-weight logical fault configurations of the phenomenological subset $\faultset_\text{sub}$. 

As we saw in \fig{def-circuits}, circuit faults can introduce diagonal space-time edges to the decoding graph. 
This allows many minimum-weight logical fault configurations of the phenomenological noise $\faultset_\text{sub}$ to deform in the time direction to form minimum-weight logical fault configurations $\faultset$ as shown in \fig{Hadamard-deformed-logicals}. 
We can analyze such faults by extending $\faultset_\text{sub}$ to form $\faultset'_\text{sub}$ by including diagonal faults from $\faultset$.
Then, eliminating brazen hook faults with respect to $\faultset'_\text{sub}$ is needed to ensure full distance.

\begin{figure}[t]
    \includegraphics[width=0.7\linewidth]{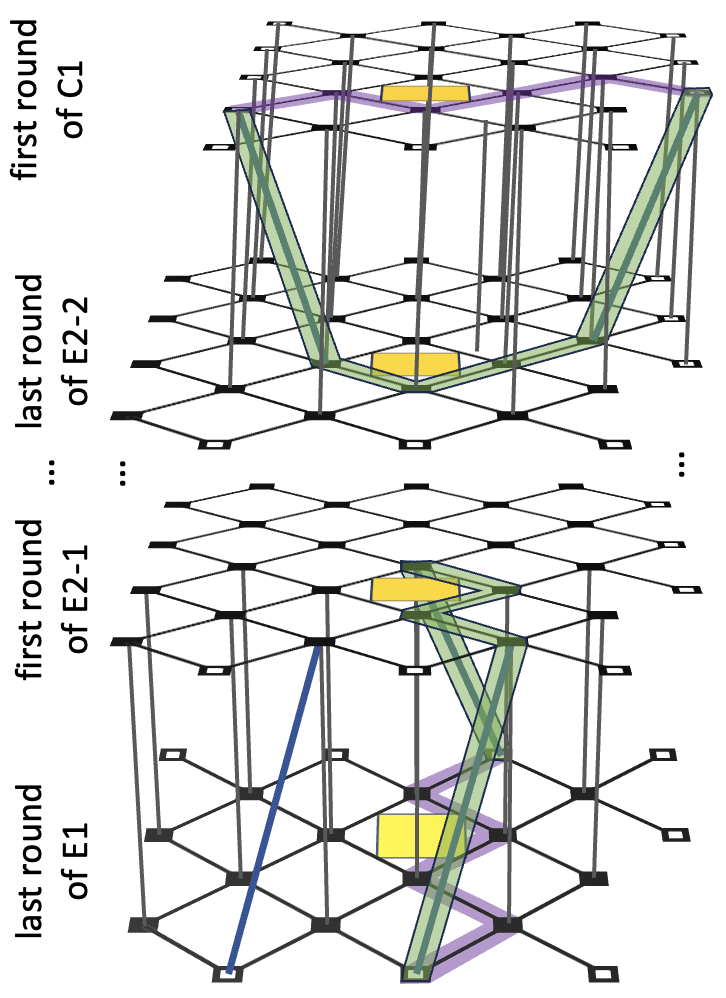}
    \caption{
    Here we show parts of the decoding graph for $Z$-type errors with circuit noise.
    Many of the edges here are also present in the phenomenological graph, but there are also many diagonal edges too (of which we only include a few for visual clarity).
    For example, a diagonal edge is present in the graph for any $Z$ fault on a data qubit between two CNOTs that connect it to two different $X$-stabilizers, leading to detection events in different rounds.  
    An effect of diagonal edges is to remove the penalty for some deformations of logical faults in the time dimension.
    For example, in the graphs here, the weight-$d$ logical operators of the codes C1 and E1 (purple) can be deformed into fault-weight-$d$ logical faults of the channel (green).
    Which paths can be taken depends on which diagonal edges are introduced by the circuits.
    We include these diagonal edges in an extended fault subset $\faultset'_\text{sub}$ and avoid brazen hooks with respect to $\faultset'_\text{sub}$.
    This introduces the requirement to avoid hooks associated with logical operators not just of the current instantaneous stabilizer group, but also those in the recent past and near future. 
    Here, this results in different requirements for those circuits which are in the earlier and later E2 rounds, as illustrated on the yellow face. 
    }
    \label{fig:Hadamard-deformed-logicals}
\end{figure}

To avoid brazen hook faults with respect to $\faultset'_\text{sub}$ due to diagonal edges, we can ensure that the brazen space-like hook faults with respect to $\faultset_\text{sub}$ (which we already found for each time-slice of the phenomenological decoding graph in \fig{Hadamard-expanded-allowed-hooks}) also don't occur in any of the $O(d)$ rounds before or after that round.
This leads us to split stage E2 into two substages E2-1 and E2-2, each consisting of $O(d)$ rounds, with different syndrome extraction circuits.
We take a more abstract viewpoint of this property in \fig{Hadamard-circuits}(a), which also shows that C1 must be split into two substages.

The precise number of rounds for each stage in needed in \fig{Hadamard-circuits}(a) depends on the precise details of the circuits used -- for example which diagonal edges are present in the decoding graph. 
In \fig{Hadamard-circuits}(b) we provide explicit circuit families, and use the distance algorithm to show that the specified numbers of rounds is sufficient to achieve full distance for $d=3,5,7,9$.
We observe that about $d/2$ rounds are sufficient for E2-1 for example, which can be understood by inspecting the diagonal space-time edge structure of the circuits in E2-1, and noting that only half of them point in the $y$-direction, such that after $d/2$ rounds of E2-1, there is no-longer any need to avoid $Z$ hook faults in the $y$-direction.

The circuit implementation of the logical Hadamard then consists of the following rounds:
\begin{enumerate}
    \item $2$ rounds: measure the generators of $\stab_\text{E1}$ with E1 circuits.
    \item $(d+1)/2$ rounds: measure the generators of $\stab_\text{E2}$ with E2-1 circuits.
    \item $(d-1)/2$ rounds: measure the generators of $\stab_\text{E2}$ with E2-2 circuits.
    \item $d$ rounds: measure the generators of $\stab_\text{C1}$ with C1-1 circuits.
    \item $d-3$ rounds: measure the generators of $\stab_\text{C1}$ with C1-2 circuits.
    \item $1$ round: measure the generators of $\stab_\text{C2}$ with C2 circuits.
    \item $1$ round: apply transverse Hadamard to data qubits.
    \item $O(1)$ rounds: translate the patch a distance $O(1)$ to final position.
\end{enumerate}
Note that the total number of time steps in the circuit implementation is more than the phenomenological implementation, essentially because C1 has been broken into C1-1 and C1-2 and each require $d$ rounds, but also some syndrome extraction circuits need 5 rather than 4 CNOT time steps.

\begin{figure*}[h]
    (a)\includegraphics[width=0.95\linewidth]{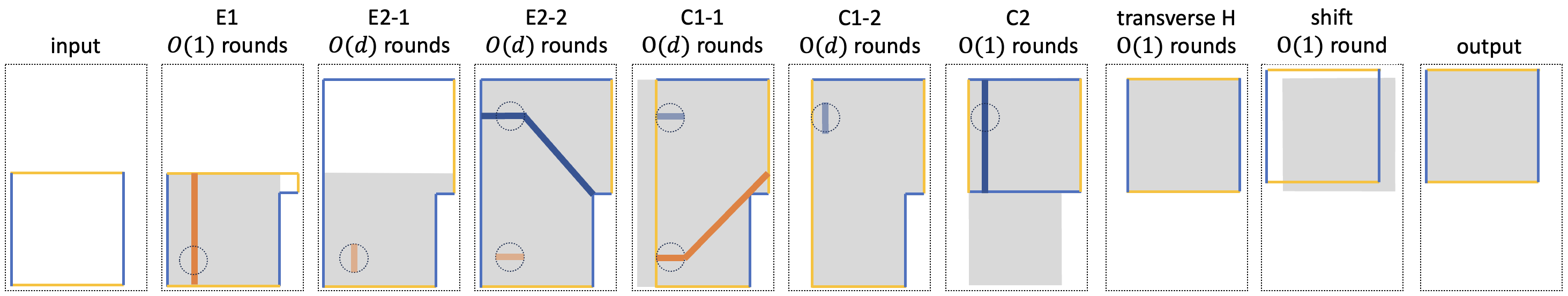}
    
    \vspace{1cm}
    (b)\includegraphics[width=0.95\linewidth]{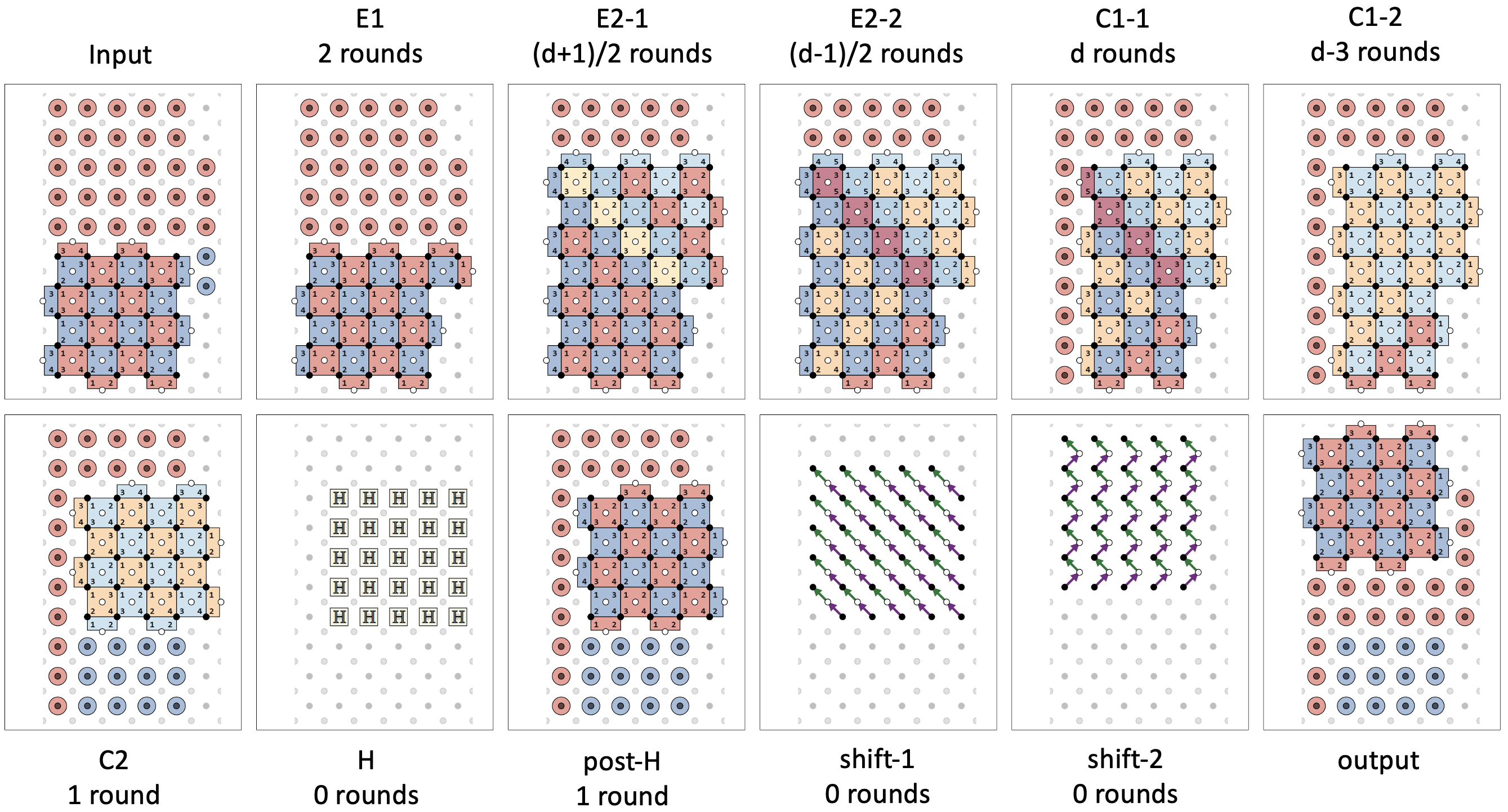}
    \caption{
    (a) To achieve a fault distance $d$ in the circuit implementation of the Hadamard, we separate each of E2 and C1 into two substages which we call E2-1, E2-2 and C1-1, C1-2 respectively. 
    We can see that the orange logical $Z$ operator, which is in the $y$-direction in E1, can be deformed along diagonal space-time edges to terminate at a later time in E2-1, which necessitates the use of a circuit in E2-1 which does not have $Z$-type hook faults in the $y$-direction.
    On the other hand, the logical operator of the patch C1 can be deformed along diagonal space-time edges to terminate at an earlier time in E2-2 with a component in the $x$-direction.
    This places the requirement that circuits in E2-2 do not have $Z$-type hook faults in the $y$-direction.
    The number of rounds of each must be $O(d)$, since diagonal edges can be taken along the $O(d)$ path.
    A very similar argument leads us to split C1 into C1-1 and C1-2 with different $X$-type hook faults.
    (b) Here we provide explicit circuit families, and use the distance algorithm to show that the specified numbers of rounds is sufficient to achieve full distance for $d=3,5,7,9,11$.
    We conjecture these circuits will achieve full distance for all $d$.
    }
    \label{fig:Hadamard-circuits}
\end{figure*}

\textbf{Relation to other work}.---
Prior to the recent work in Ref.~\cite{geher2023b}, the known approaches to implement the logical Hadamard with surface codes can be viewed as variants of that originally proposed in Ref.~\cite{Horsman2012} in which a transverse Hadamard is applied to the data qubits of 
a surface code patch, and the patch is rotated by first expanding the patch and then shrinking it.
It is assumed in Ref.~\cite{Horsman2012} that 
at the phenomenological-level, the expansion and shrink can both be accomplished by a single step followed by $d$ rounds of stabilizer measurements.
We find however that this results in the fault distance dropping below the code distance, and we propose a modified full-distance version of this phenomenological implementation that includes a few additional patch deformation steps to avoid this problem.
This problem is also avoided by the phenomenological implementations of Ref.~\cite{Vuillot2019} and Ref.~\cite{Fowler2018}, but with different space and time cost tradeoffs.
To extend the full-distance phenomenological implementation to form a full-distance circuit implementation with single-ancilla measurement extraction circuits, we find that the circuits must vary both spatially and temporally as the patch is deformed, and yet more idling steps must be added.

Recently, Ref.~\cite{geher2023b} was released after our initial draft was completed with the construction provided here.
Therein, two approaches to implement the Hadamard are provided.
The first is quite similar to ours, and again can be viewed as a variant of that in Ref.~\cite{Horsman2012}.
The second is a little different, where the authors have deformed the domain wall from a time-like application of the transverse Hadamard into a space like domain wall which cuts through one of the codes in their code sequence.
We note that their circuit-level implementation is more efficient than that which we present here -- taking $2d +O(1)$ error correction cycles rather than $3d + O(1)$ as we find here.

\section{Concluding remarks}
\label{sec:concl}

We foresee the tools developed in this paper finding application for designing and verifying the fault tolerance of future stabilizer channels, especially for Floquet, LDPC and space-time codes.

In what follows we list some other directions that would be interesting to explore in future work:
\begin{itemize}
    \item \textbf{Non-Clifford operations.}---
    Our framework does not explicitly address non-Clifford operations.
    We can include non-Clifford state preparation in our framework via the following argument.
    Assume that a noisy non-Clifford state preparation channel $\mathcal{T}'$ with output code $\mathcal{S}$ is equivalent to a perfect state preparation $\mathcal{T}$
    followed by a noisy channel $\mathcal{I}$ consisting of identity gates on each of the output qubits of $\mathcal{T}$ with stochastic Pauli noise.
    For any stabilizer channel $\mathcal{C}$ with input code $\mathcal{S}$, probability of logical failure of $\mathcal{T}' \circ \mathcal{C}$
    is then the same as probability of logical failure of stabilizer channel $\mathcal{I} \circ \mathcal{C}$. 
    Therefore, all the tools we have developed for analysis of stabilizer channels apply to this more general scenario.
    One open question is: if and how decoder for channel $\mathcal{C}$ can be benefit from the measurement outcomes of the channel $\mathcal{T'}$.
    The second open question is: does assumption that $\mathcal{T}'$ is equivalent to $\mathcal{T} \circ \mathcal{I}$ captures practically relevant scenarios and 
    what are practically useful generalisations of the assumption.
    
    \item \textbf{Adaptive logical circuits.}---
    We did not explicitly consider logical circuits with adaptive operations conditional on measurement outcomes.
    Adaptive operations are necessary when injecting magic states\cite{Bravyi2005} and are a useful tool for reducing resource requirements for various quantum algorithms \cite{Jones2013,Gidney2018,gidney2021cccz}.
    One way to address this is to ensure that for any possible fixed set of logical outcome we apply a Clifford channel with fault distance at least $d$.
    This can be achieved by relying on our channel composition results in \sec{combining-stabilizer-channels}.
    Assuming elementary faults occur with probability $O(p)$, 
    the probability of the adaptive channel failure is at most $O(p^{(d/2)}) + p'$, where $p'$ is the probability of the failure to predict the correct outcome of the logical measurements during the channel execution.
    Ensuring that $p'$ is small might require the adaptive channel design modifications and is related to the data backlog problem \cite{Terhal2015}.
    Design constraints on adaptive logical circuits is an active area of research closely related to efficient scalable decoders.
    
    \item \textbf{Efficient and scalable decoders.}---
    We have not considered decoding in much detail. 
    Implementing large-scale quantum algorithms fault-tolerantly requires real-time parallel window  decoding~\cite{skoric2023parallel} 
    due to data backlog problem \cite{Terhal2015}.
    In the main text we only discussed global non-real-time decoders for graph-like check matrices.
    There are many practical examples where check matrices are not graph-like 
    and require extending decoders beyond graph-like check matrices\cite{delfosse2023,honeycombmemory}.
    
    \item \textbf{Entropic contributions to channel failure.}---
    Here we have looked at the fault distance $d$, which corresponds to the size of the minimum-weight logical fault configuration of the channel.
    This relates to the exponent $t$ of the lowest-order contributions to the logical failure rate $C \cdot p^t$.
    Another object that would be interesting to study is the number of minimum-weight logical fault configurations, which relates to the prefactor $C$, known as the entropic contribution~\cite{beverland2019}.
    Maybe new definitions and algorithms for classifying hook faults could aid the design of channels to reduce this prefactor.

    \item \textbf{More general noise models.}---
    While the class of stochastic Pauli noise models that we have assumed in this formalism is broad (allowing for high-weight errors and measurement outcome flips), it does not cover a number of practically motivated noise models.
    It would be interesting to explore extensions of the formalism that apply to soft measurement outcomes~\cite{Pattison2021}, coherent errors~\cite{bravyi2018surface_code_coherent_noise} and amplitude damping. 
    
    \item \textbf{Further FT conditions to prove a threshold.}--- 
    Our results do not imply a new version of the threshold theorem.
    In Ref.~\cite{aliferis2005}, the canonical example of the threshold theorem, there are two fault-tolerant conditions (see also Sec 1.1 of \cite{chamberland2018} for a brief discussion).
    The first condition concerns the number of faults that can be corrected in a gadget when the output is corrected perfectly, and is very similar to what we define here as a fault-distance. 
    The second condition involves the case where there is no perfect decoding at the end, and the output of one gadget is decoded and fed as the input to the next gadget.
    It is used to ensure that faults in far-separated components do not lead do logical faults, such that most fault configurations larger than $d$ are actually correctable.
    This is necessary to establish the threshold theorem.
    Our definition of time-local channels is a step towards this goal.
    
\end{itemize}

\textbf{Acknowledgements.}---
We thank Adam Paetznick, Nicolas Delfosse, Jeongwan Haah and Marcus Silva for useful discussions.
This work was initiated when S.H. was an intern at Microsoft Quantum, and was completed while M.E.B. was a researcher at Microsoft Quantum.

\bibliography{bib_fault_distance.bib}

\appendix

\section{Shortest odd cycle problem}
\label{app:shortest-odd-cycle}

For completeness, we provide a well-known algorithm~\ref{alg:shortest-odd-cycle} for shortest odd cycle problem (page 95, \cite{DePinaThesis}).
It can be used to find the channel check distance when the channel check matrix is graph-like, as discussed in \sec{fault-dist-algo}.
The runtime of the algorithm is $O(|V|(|E|+|V|\log|V|))$. It is dominated by $|V|$ calls to the shortest path algorithm in line~\ref{line:sssp}.
Each call has complexity $O(|E|+|V|\log|V|)$.
Note that if we know that the shortest cycle must include a vertex from some set $V' \subset V$,
the complexity reduces to $O(|V'|(|E|+|V|\log|V|))$ because we can replace $V$ with $V'$ in line~\ref{line:sssp}.
If there is an upper-bound on the cycle weight (length), this can further speed up the calls to the shortest paths algorithm in line~\ref{line:sssp}.

\begin{algorithm}[H]
    \caption{Shortest odd cycle}    
    \label{alg:shortest-odd-cycle}
    \begin{algorithmic}[1]
    \Require Graph $(V,E)$, weights $w : E \rightarrow \mathbb{R}^{\ge 0}$, $E'\subset E$
    \Ensure Shortest cycle $E_{\min}$ with $|E_{\min} \cap E'|$ odd
    \State Let $V^\pm = \left\{ v^+, v^- : v \in V \right\}$, use $|v^+| = v, |v^-| = v$
    \State Let $G^\pm=(V^\pm,E^\pm)$ be a graph with empty edge-set $E^\pm$
    \For{ every edge $(v_1,v_2) \in E $}
    \If{ $(v_1,v_2) \in E'$ }
        \State Add $(v_1^+,v_2^-)$, $(v_1^-,v_2^+)$ with weight $w(v_1,v_2)$ to $E^\pm$
    \Else 
        \State Add $(v_1^+,v_2^+)$, $(v_1^-,v_2^-)$ with weight $w(v_1,v_2)$ to $E^\pm$
    \EndIf
    \EndFor
    \State For $v \in V$, find $P_v$ the shortest path between $v^+,v^-$ \label{line:sssp}
    \State Let $P$ be the  shortest path among $\left\{ P_v : v \in V \right\}$
    \State \Return $\{ (|v^\pm_1|,|v^\pm_2|) : \text{ for edge } (v^\pm_1,v^\pm_2) \in P \}$
    \end{algorithmic}
\end{algorithm}

There is a simple intuitive explanation of the algorithm correctness.
Imagine $V^\pm = \{ v^\pm : v \in V \} $ as two `floors' of a building, with edges of $E^\pm$ derived from $E'$
corresponding to staircases between floors $V^+, V^-$.
Edges of $E^\pm$ derived from the rest of the edges $E \setminus E'$ correspond to staying on the same floor.
Every walk from `floor' $V^+$ to `floor' $V^-$
takes an odd number of staircases.
Every cycle in $(V,E)$ that passes through vertex $v$ and has an odd overlap with $E'$ corresponds to a walk between $v^+$ and $v^-$, and vice-versa.

\section{Topological viewpoint of surface codes}
\label{app:topo-viewpoint}

Since their inception~\cite{Kitaev2003}, surface codes have been understood from a topological viewpoint; see \fig{surfacecode}(b).
When the microscopic details are ignored,
the planar \textit{code patch} in \fig{surfacecode}(a) can be viewed as a topological object $\patch$ embedded in 2D space, consisting of a closed connected 2D region with a number of additional labeled features as shown in \fig{surfacecode}(b).
We say that the region is in the \textit{topological phase}, while outside the region is in the \textit{trivial phase}.
The boundary which separates these two phases is bi-colored with yellow and blue, where each yellow (blue) boundary section is referred to as a primal (dual) boundary.
The intersection point of a primal boundary and a dual boundary is referred to as a \textit{corner}.

\begin{figure}
    \includegraphics[width=1.0\linewidth]{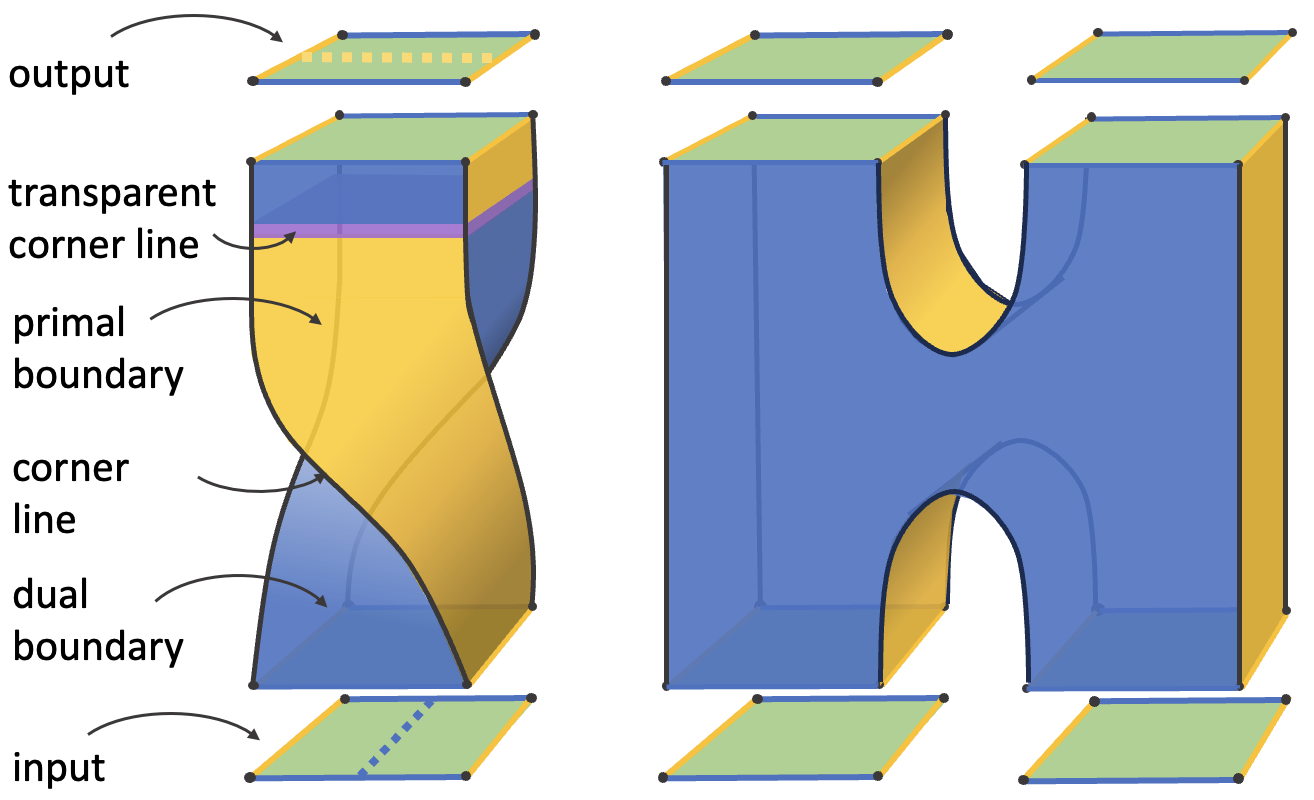}
    \caption{
    Topological descriptions of channels which implement logical Hadamard gate (left) and a logical $ZZ$ measurement (right) in surface codes. 
    For both channels, primal boundaries and dual boundaries are colored in yellow and blue respectively.
    The purple time slice is a domain wall, and the line where this section intersects the boundary of $\topochannel_H$ is a transparent corner line.
    The input and output codes of the channel are at the bottom and on the top of each topological object.
    }
    \label{fig:topo-viewpoint}
\end{figure}

The topological description of surface codes in turn leads to topological description of the logical operations that can be applied to them that includes approaches such as braiding punctures~\cite{Raussendorf2006,raussendorf2007fault,Raussendorf2007b,Bombin2008,Bombin2009,fowler2009high,Fowler2012}, lattice surgery \cite{Horsman2012} and braiding twist defects~\cite{Bombin2010,hastings2014reduced}.
We briefly review some aspects of this topological description of logical operations following the more comprehensive discussion presented in Sec III.B of Ref.~\cite{Bombin2021}.

A stabilizer channel of surface codes can be viewed as a topological object $\topochannel$ embedded in (2+1)-D spacetime, consisting of a closed connected 3D region with a number of labeled features which are generalizations of features of 2D surface code patches to 3D; see \fig{topo-viewpoint}.
The region inside $\topochannel$ is in the topological phase, while the region outside $\topochannel$ is in the trivial phase. 
Roughly speaking, $\topochannel$ records the history of deformation from $\patch_{\textrm in}$ to $\patch_{\textrm out}$ which are the 2D boundaries of $\topochannel$ at the start and end.
The entire boundary of $\topochannel$ is divided into components, with each component categorized as either a \textit{primal boundary surface} or a \textit{dual boundary surface}.
Two adjacent boundary surfaces must have different types.
A \textit{corner line} is the intersection line between a primal boundary surface and a dual one.
A \textit{domain wall} is the boundary between two bulk regions of $\mathcal{T}$.
A domain wall can either terminate in the bulk of the toric code phase where a \textit{twist line} emerges, or terminate on a corner line.
A corner line on the boundary of a domain wall is said to be \textit{transparent}.

While a patch $\patch$ is a topological object with no specified locations of qubits, we can consider explicit cellulations of it with vertices and plaquettes colored yellow and blue and recover explicit surface codes.
For example, such as the distance $d$ surface code shown in \fig{surfacecode}(a) can be recovered by a particular cellulation of $\patch$.

A $Z$-type ($X$-type) non-trivial logical operator forms a string in $\topochannel$ that connects two different primal (dual) boundaries or forms a topologically non-trivial primal (dual) loop in $\topochannel$.
With the presence of domain walls, a string can start as $Z$-type from a primal boundary, cross a domain wall becoming $X$-type, and then end at a dual boundary.
If the string ends at two boundaries that share a transparent corner line, then the string is considered as a trivial logical operator.
The action of a topological channel can be inferred by considering these logical operators.

In Ref.~\cite{Bombin2021}, a strategy was given to specify code deformation channels via the cellulation of a topological description $\topochannel$.
In this work we use a less precise method, and use the topological description of the Hadamard and $ZZ$ measurement depicted in \fig{topo-viewpoint} to loosely suggest an initial code deformation sequence, which we then analyze and modify in the following sections using the techniques we have presented in this paper.

\end{document}